\documentclass[11pt]{gen-j-l}

\usepackage[hmargin=1.3in,vmargin=1in]{geometry}

\usepackage{amsmath,amsthm,amssymb,amsfonts,verbatim}

\usepackage[dvipsnames,usenames]{color}
\usepackage{hyperref}

\makeatletter
\def\@tocline#1#2#3#4#5#6#7{\relax
  \ifnum #1>\c@tocdepth 
  \else
    \par \addpenalty\@secpenalty\addvspace{#2}%
    \begingroup \hyphenpenalty\@M
    \@ifempty{#4}{%
      \@tempdima\csname r@tocindent\number#1\endcsname\relax
    }{%
      \@tempdima#4\relax
    }%
    \parindent\z@ \leftskip#3\relax \advance\leftskip\@tempdima\relax
    \rightskip\@pnumwidth plus4em \parfillskip-\@pnumwidth
    #5\leavevmode\hskip-\@tempdima
      \ifcase #1
       \or\or \hskip 2em \or \hskip 2em \else \hskip 3em \fi%
      #6\nobreak\relax
    \dotfill\hbox to\@pnumwidth{\@tocpagenum{#7}}\par
    \nobreak
    \endgroup
  \fi}
\makeatother

 \providecommand{\F}{\mathbb{F}}

\title[AG codes \& near-optimal list decoding]{Optimal rate list decoding over bounded alphabets using algebraic-geometric codes}

\author[V. Guruswami]{Venkatesan Guruswami}
\address{Computer Science Department \\  Carnegie Mellon University \\ Pittsburgh, USA.}
\email{venkatg@cs.cmu.edu}
\thanks{Extended abstracts announcing these results were presented at the 2012 and 2013 ACM Symposia on Theory of Computing (STOC)~\cite{GX-stoc12,GX-stoc13}. This is a merged and significantly revised version of these conference papers, that accounts for the explicit subspace designs that were constructed in \cite{GK-combinatorica} subsequent to \cite{GX-stoc13}, makes some simplifications and improvements in the construction of h.s.e sets in Section~\ref{sec:hse} compared to \cite{GX-stoc12}, and reorganizes the material and flow substantially.  \\
~ \\
The research of V. Guruswami was supported in part by a Packard Fellowship and NSF grants CCF-0963975, CCF-1422045, and CCF-1814603. Some of this work was done during visits by the author to Nanyang Technological University.}
\author[C. Xing]{Chaoping Xing}
\address{School of Electronic Information and Electrical Engineering\\ Shanghai Jiao Tong University; and\\
Division of Mathematical Sciences \\ School of Physical \&  Mathematical Sciences \\ Nanyang Technological University \\ Singapore.}
\email{xingcp@ntu.edu.sg}
\thanks{The research of C. Xing was supported  by the National Research Foundation, Prime Minister's Office, Singapore under its Strategic Capability Research Centres Funding Initiative; and the Singapore MoE Tier 1 grants RG25/16 and RG21/18. }

\newtheorem{lemma}{Lemma}[section]
\newtheorem{theorem}[lemma]{Theorem}
\newtheorem{cor}[lemma]{Corollary}
\newtheorem{fact}[lemma]{Fact}
\newtheorem{obs}[lemma]{Observation}
\newtheorem{claim}[lemma]{Claim}
\newtheorem{defn}{Definition}

\newcommand{\dims}{\kappa}

\theoremstyle{remark}
\newtheorem{rmk}{Remark}

\newcommand{\eps}{\varepsilon}
\renewcommand{\epsilon}{\varepsilon}
\renewcommand{\le}{\leqslant}
\renewcommand{\ge}{\geqslant}

\def\eqdef{\stackrel{\mathrm{def}}{=}}
\def\period{\Delta}
\def\proj{\mathrm{proj}}
\def\hse{\mathsf{HSE}}
\def\degree{\lambda}
\newcommand{\mv}[1]{{\mathbf{#1}}}

\def\ZZ{\mathbb{Z}}
\def\PP{\mathbb{P}}

\def \mL {\mathcal{L}}

\def \mP {\mathcal{P}}

\def\cL{{\mathcal L}}
\def\cl{{\mathcal L}}

\def\Pin{{P_{\infty}}}
\def\g{{\mathfrak{g}}}

\def\Fm{{\F_{q^m}}}

\newcommand{\Ga}{\alpha}
\newcommand{\Gb}{\beta}
\newcommand{\Gg}{\gamma}     
     
\newcommand{\Ge}{\epsilon}

\newcommand{\Gk}{\kappa}

\newcommand{\Gs}{\sigma}

\newcommand{\s}{\sigma}
\def\FH{\mathsf{FH}}

\def\FGS{\mathsf{FGS}}

\def \bi {{\bf i}}

\def \by {{\bf y}}

\def\Supp {{\rm Supp }}

\def\Aut {{\rm Aut }}

\def\FD{{\mathrm{F}}}

\newcommand{\sd}{subspace design}
\newcommand{\sds}{subspace designs}
\newcommand{\dm}{\Lambda}
\def\rs{\mathsf{RS}}

\begin{document}

\maketitle
\thispagestyle{empty}

\enlargethispage{1.2cm}

\begin{abstract}

We construct two classes of algebraic code families which are efficiently list decodable with small output list size from a fraction $1-R-\eps$ of adversarial errors where $R$ is the rate of the code, for any desired positive constant $\eps$. The alphabet size depends only on $\eps$ and is nearly-optimal.

The first class of codes are obtained by \emph{folding} algebraic-geometric codes using automorphisms of the underlying function field. The second class of codes are obtained by restricting evaluation points of an algebraic-geometric code to rational points from a \emph{subfield}. In both cases, we develop a linear-algebraic approach to perform list decoding, which pins down the candidate messages to a subspace with a nice ``periodic" structure.

To prune this subspace and obtain a good bound on the list-size, we pick subcodes of these codes by pre-coding into certain \emph{subspace-evasive} sets which are guaranteed to have small intersection with the sort of periodic subspaces that arise in our list decoding. We develop two approaches for constructing such subspace-evasive sets. The first is a Monte Carlo construction of \emph{hierearchical subspace-evasive} (h.s.e) sets which leads to excellent list-size but is not explicit. The second approach exploits a further \emph{ultra-periodicity} of our subspaces and uses a novel construct called \emph{subspace designs}, which were subsequently constructed explicitly and also found further applications in pseudorandomness.

To get a family of codes over a fixed alphabet size, we instantiate our approach with algebraic-geometric codes based on the Garcia-Stichtenoth tower of function fields. Combining this with pruning via h.s.e sets yields codes list-decodable up to a $1-R-\eps$ error fraction with list size bounded by $O(1/\eps)$, matching the existential bound for random codes up to constant factors. Further, the alphabet size can be made $\exp(\tilde{O}(1/\eps^2))$ which is not much worse than the lower bound of $\exp(\Omega(1/\eps))$. The parameters we achieve are thus quite close to the  existential bounds in all three aspects---error-correction radius, alphabet size, and list-size--- simultaneously. This construction is, however, Monte Carlo and the claimed list decoding property only holds with high probability. Once the code is (efficiently) sampled, the encoding/decoding algorithms are deterministic with a running time $O_\eps(N^c)$ for an absolute constant $c$, where $N$ is the code's block length.

Using subspace designs instead for the pruning, our approach yields a deterministic construction of an algebraic code family of rate $R$ with efficient list decoding from $1-R-\eps$ fraction of errors over an {alphabet of constant size} $\exp(\tilde{O}(1/\eps^2))$.  The list size bound is upper bounded by a very slowly growing function of the block length $N$; in particular, it is at most $O(\log^{(r)} N)$ (the $r$'th iterated logarithm) for any fixed integer $r$. The explicit construction avoids the shortcoming of the Monte Carlo sampling at the expense of a worse list size.


\end{abstract}

\newpage
\enlargethispage{2.5cm}
\tableofcontents

\newpage
\section{Introduction}

An error-correcting code $C$ of block length $N$ over a finite alphabet $\Sigma$ maps a set $\mathcal{M}$ of messages into codewords in $\Sigma^N$. The rate of the code $C$, denoted $R$, equals $\frac{1}{N} \log_{|\Sigma|} |\mathcal{M}|$. In this work, we will be interested in codes for adversarial noise, where the channel can arbitrarily corrupt any subset of up to $\tau N$ symbols of the codeword. The goal will be to correct such errors and recover the original message/codeword efficiently. It is easy to see that information-theoretically, we need to receive at least $RN$ symbols correctly in order to recover the message (since $|\mathcal{M}| = |\Sigma|^{RN}$), so we must have $\tau \le 1-R$.

Perhaps surprisingly, in a model called list decoding, recovery up to this information-theoretic limit becomes possible. Let us say that a code $C \subseteq \Sigma^N$ is $(\tau,\ell)$-list decodable if for every received word $\mv{y} \in \Sigma^N$, there are at most $\ell$ codewords $\mv{c} \in C$ such that $\mv{y}$ and $\mv{c}$ differ in at most $\tau N$ positions. Such a code allows, in principle, the correction of a fraction $\tau$ of errors, outputting at most $\ell$ candidate codewords one of which is the originally transmitted codeword.

The probabilistic method shows that a random code of rate $R$ over an
alphabet of size $\exp(O(1/\eps))$ is with high probability
$(1-R-\eps,O(1/\eps))$-list decodable~\cite{elias91}. However, it is
not known how to construct or even randomly sample such a code for
which the associated algorithmic task of list decoding (i.e., given
$\mv{y} \in \Sigma^N$, find the list of codewords within fractional
radius $1-R-\eps$) can be performed efficiently. This work takes a big
step in that direction, giving a randomized construction of such
efficiently list-decodable codes over a slightly worse alphabet size
of $\exp(\tilde{O}(1/\eps^2))$. We note that the alphabet size needs
to be at least $\exp(\Omega(1/\eps))$ in order to list decode from a
fraction $1-R-\eps$ of errors, \footnote{ The best trade-off between rate $R$ and list decoding radius $\tau$ is the Gilbert-Varshamov bound, i.e., $R\le 1-H_q(\tau)$, where $ H_q(\tau)$ is the $q$-ary entropy function $x\log_q(q-1)-x\log_qx-(1-x)\log_q(1-x)$. The function $1-H_q(\tau)$ is equal to $1-\tau-\Ge$
if the alphabet size is at least $\exp(\Omega(1/\eps))$.} so this is close to optimal. For the
list-size needed as a function of $\eps$ for decoding a $1-R-\eps$
fraction of errors, the best lower bound is only $\Omega(\log
(1/\eps))$~\cite{GN14}, but as mentioned above, even random coding
arguments only achieve a list-size of $O(1/\eps)$, which our
construction matches up to constant factors. We also give a fully {\em
  deterministic} construction with a list-size that is very slowly
growing as a function of the block length.

We now review some of the key results on algebraic list decoding
leading up to this work. A more technical comparison with related work
appears in Section~\ref{sec:intro-related-work}.
The work of Sudan~\cite{sudan} used bivariate polynomial interpolation to give the first efficient list decoding algorithm for Reed-Solomon codes, which for rates $R$ below $1/3$ corrected a fraction of errors exceeding the $(1-R)/2$ bound achievable by unique decoding. Guruswami and Sudan~\cite{GS99} introduced multiplicities in the interpolation step and gave an efficient list decoding algorithm that could correct an error-fraction $1-\sqrt{R}$. The multiplicities also offered an avenue to incorporate ``soft" information about varying reliability of different symbols, which was developed by Koetter and Vardy~\cite{KV} to give an influential algebraic soft-decision decoder for Reed-Solomon codes. The $1-\sqrt{R}$ bound remained the largest known efficiently list-decodable error-fraction for any value of rate $R$ till Parvaresh and Vardy~\cite{PV-focs05} gave a variant of Reed-Solomon codes list-decodable up to error fraction $1-O(R \log (1/R))$ which beats the $1-\sqrt{R}$ bound for low-rates.

Building on the Parvaresh-Vardy work together with further new algebraic ideas, Guruswami
and Rudra~\cite{GR-FRS} gave the first
construction of codes that achieved the optimal trade-off between rate
and list-decoding radius, i.e., enabled list decoding up to a fraction
$1-R-\eps$ of worst-case errors with rate $R$. They showed that a variant of Reed-Solomon
(RS) codes called {\em folded} RS codes admit such a list decoder. For
a decoding radius of $1-R-\eps$, the code was based on bundling
together disjoint windows of $m =\Theta(1/\eps^2)$ consecutive symbols
of the RS codeword into a single symbol over a larger alphabet. As a
result, the alphabet size of the construction was
$N^{\Omega(1/\eps^2)}$. It was also shown in \cite{GR-FRS} that ideas based on code concatenation and expander
codes can be used to bring down the alphabet size to
$\exp(\tilde{O}(1/\eps^4))$ which is independent of the block length. However, the resulting codes lose some important and powerful features
such as list recovery and soft decoding of the folded RS codes. Also,
the decoding time complexity as well as proven bound on worst-case
output list size for folded RS codes were $N^{\Omega(1/\eps)}$.\footnote{The list size for decoding folded RS codes was shown to be bounded by a constant depending only on $\eps$ in subsequent work~\cite{KRS-pc}.}

Our main final result statement is the following, offering two constructions, one randomized and one deterministic, of variants of algebraic-geometric (AG) codes that are list-decodable with optimal rate. These appear as Theorems~\ref{thm:final-monte-carlo} and \ref{thm:cap-achieving-gs-deterministic} in the final section of the paper.

\begin{theorem}[Main]
\label{thm:main-intro}
For any $R \in (0,1)$ and positive constant $\eps \in (0,1)$, there is
\begin{enumerate}
\item[{\rm (i)}] a Monte Carlo construction of a family of codes of rate at least
  $R$ over an alphabet size $\exp(O(\log(1/\eps)/\eps^2))$ that are
  encodable and $(1-R-\eps, O(1/(R\eps))$-list decodable in
  $O_\eps(N^c)$ time\footnote{We use the $O_\eps(\cdot)$ notation to hide constant factors that depend on $\eps$.}, where $N$ is the block length of the code and
  $c$ is an absolute positive constant.
\item[{\rm (ii)}] a deterministic construction of a family of codes of rate at
  least $R$ over an alphabet size $\exp(O(\log^2(1/\eps)/\eps^2))$ that
  are encodable and $(1-R-\eps,L(N))$-list decodable in $O_\eps(N^c)$
  time, for a list size that satisfies $L(N) = o(\log^{(r)} N)$ (the
  $r$'th iterated logarithm) for any fixed integer $r$.
\end{enumerate}
\end{theorem}

The first part of Theorem \ref{thm:main-intro} is achieved through folded algebraic-geometric codes. To fold  algebraic-geometric codes, we first find suitable automorphisms of the ground function field. The list of possible candidate messages output by the list decoder has exponential size, but is contained in a well structured subspace. To prune down the list size, we only encode messages that belong to so-called hierarchical subspace-evasive sets, which are chosen to have small intersection with the structured subspaces arising in the decoding. To make use of subspace-evasive sets efficiently, we have to: (i) give an efficient pseudorandom construction of these sets;  and (ii) encode the messages to  subspace-evasive sets efficiently. We refer to Section \ref{sec:tech} for details.

The second part of Theorem \ref{thm:main-intro} is obtained through usual algebraic-geometric codes with evaluation points over subfields. As in the first part,  the list of possible candidate messages belongs to a subspace that is  well structured, specifically with a property we called ultra-periodicity (Definition~\ref{def:ultra-periodic}). The approach based on hierarchical subspace-evasive sets in the first part leads to excellent list size; however, we only know randomized constructions of hierarchical subspace-evasive sets. To obtain a deterministic list decoding, we prune down the list of possible solutions through subspace designs (see Section \ref{sec:tech} for details).

We note that our Monte Carlo construction gives codes that are quite
close to the existential bounds in three aspects simultaneously ---
the trade-off between error fraction $1-R-\eps$ and rate $R$, the
list-size as a function of $\eps$, and the alphabet size of the code
family (again as a function of $\eps$).  Even though these codes are
not fully explicit, they are ``functionally explicit" in the sense
that once the code is (efficiently) sampled, with high probability the
polynomial time encoding and decoding algorithms deliver the claimed
error-correction guarantees for {\em all} allowed error pattern. The
explicit construction avoids this shortcoming at the expense of a
slightly worse list size.

\subsection{Prior and related work}
\label{sec:intro-related-work}

Let us recap a bit more formally the construction of folded RS codes from \cite{GR-FRS}. One begins with the Reed-Solomon encoding of a polynomial $f \in \F_q[X]$ of degree $< k$ consisting of the evaluation of $f$ on a subset of field elements ordered as $1,\gamma,\dots,\gamma^{N-1}$ for some primitive element $\gamma \in \F_q$ and $N< q$. For an integer ``folding" parameter $m \ge 1$ that divides $N$, the folded RS codeword is defined over alphabet $\F_q^m$ and consists of $n/m$  blocks, with the $j$'th block consisting of the $m$-tuple $(f(\gamma^{(j-1)m}),f(\gamma^{(j-1)m+1}),\ldots,f(\gamma^{jm-1}))$.
The algorithm in \cite{GR-FRS} for list decoding these codes  was based on the algebraic identity $\overline{f(\gamma X)} = \overline{f(X)}^q$ in the residue field $\F_q[X]/(X^{q-1}-\gamma)$ where $\overline{f}$ denotes the residue $f \bmod {(X^{q-1}-\gamma)}$. This identity is used to solve for $f$ from an equation of the form $Q(X,f(X),f(\gamma X),\dots,f(\gamma^{s-1} X)) = 0$ for some low-degree nonzero multivariate polynomial $Q$. The high degree $q> n$ of this identity, coupled with $s \approx 1/\eps$, led to the large bounds on list-size and decoding complexity in \cite{GR-FRS}.

One possible approach to reduce $q$ (as a function of the code length) in this construction would be to work with algebraic-geometric codes based on function fields $K$ over $\F_q$ with more rational points. However, an automorphism $\sigma$ of $K$ that can play the role of the automorphism $f(X) \mapsto f(\gamma X)$ of $\F_q(X)$ is only known (or even possible) for very special function fields. This approach was used in \cite{Gur-cyclo-ANT} to construct list-decodable codes based on cyclotomic function fields using as $\sigma$ certain Frobenius automorphisms. These codes improved the alphabet size to polylogarithmic in $N$, but the bound on list-size and decoding complexity remained $N^{\Omega(1/\eps)}$.

Subsequently, a linear-algebraic approach to list decoding folded RS codes was discovered in \cite{salil-NOW,Gur-ccc11}. Here, in the interpolation stage, which is common to all list decoding algorithms for algebraic codes~\cite{sudan,GS99,PV-focs05,GR-FRS}, one finds a {\em linear} multivariate polynomial $Q(X,Y_1,\dots,Y_s)$ whose total degree in the $Y_i$'s is $1$. The simple but key observation driving the linear-algebraic approach is that the equation $Q(X,f(X),\dots,f(\gamma^{s-1} X)) = 0$ now becomes a linear system in the coefficients of $f$. Further, it is shown that the solution space has dimension less than $s$, which again gives a list-size upper bound of $q^{s-1}$.
Finally, since the list of candidate messages fall in an affine space, it was noted in \cite{Gur-ccc11} that one can bring down the list size by carefully ``pre-coding" the message polynomials so that their $k$ coefficients belong to a ``subspace-evasive set" (which has small intersection with every $s$-dimensional subspace of $\F_q^k$).
This idea was used in \cite{GW-tit13} to give a {\em randomized} construction of $(1-R-\eps,O(1/\eps))$-list decodable codes of rate $R$.
However, the alphabet size and runtime of the decoding algorithm both remained $N^{\Omega(1/\eps)}$.
Similar results were also shown in \cite{GW-tit13,Kopparty15} for univariate multiplicity codes, where the encoding of a polynomial $f$ consists of the evaluations of $f$ and its first $m-1$ derivatives at distinct field elements.

Concurrently with the conference version of part of this work reported in \cite{GX-stoc12}, Dvir and
Lovett \cite{DL-stoc12} gave an elegant construction of explicit subspace evasive sets
based on certain algebraic varieties. Furthermore, Ben-Aroya and  Shinkar \cite{AS14} improved the result of \cite{DL-stoc12} slightly by using an elementary  construction. Their results yield an
explicit version of the codes from \cite{Gur-ccc11}, albeit with a
worse list size bound of $(1/\eps)^{O(1/\eps)}$. This work and
\cite{DL-stoc12,AS14} are incomparable in terms of results. The big advantage of
\cite{DL-stoc12,AS14} is the deterministic construction of the code. The
benefits in our work are: (i) both constructions in the present paper give codes over an alphabet size that is a constant independent of $N$,
whereas in \cite{DL-stoc12} the $N^{\Omega(1/\eps^2)}$ alphabet size of
folded RS codes is inherited; (ii)
our first Monte Carlo construction ensures list-decodability with a list-size of $O(1/\eps)$ that is much
better and in fact matches the full random construction up to constant factors,\footnote{As
  mentioned above, the bound in \cite{DL-stoc12} is $(1/\eps)^{O(1/\eps)}$
  and it seems very difficult to get a sub-exponential dependence on
  $1/\eps$ with the algebraic approach relying on Bezout's theorem to
  construct subspace-evasive sets.} and (iii) our second construction gives a deterministic algorithm as well with almost constant list size (and constant alphabet size). Another important feature is that both our work and \cite{DL-stoc12,AS14} achieve a decoding complexity of $O_\eps(N^c)$ with exponent independent of $\eps$.

Our paper presents two class of codes: folded algebraic-geometric
codes and usual algebraic-geometric codes with evaluation points over
subfields. For both the classes of codes, we can apply hierarchical
subspace-evasive sets as well as subspace design to prune down the
list size by taking certain subcodes. This is because of the
``periodic" structure of the subspace in which the candidate messages
are pinned down by the linear-algebraic list decoder is similar in
both cases. Thus, we can obtain both randomized and deterministic
algorithms from each of the two classes of codes. In total, we have
four combinations of constructions.  To illustrate both algebraic
approaches, we decide to focus on two combinations, i.e., (i) folded
algebraic-geometric codes with hierarchical subspace-evasive sets; and
(ii) usual algebraic-geometric codes with evaluation points over
subfields with subspace designs. These are listed in
Figure~\ref{fig:comparison-table}. We note that the other two
combinations are also possible, as the pruning of the subspace of
solutions is ``black-box" with respect to its periodic structure.

In the table presented in Figure~\ref{fig:comparison-table}, we list  previous results and those in this paper.  The major improvement of this work is to bring down the alphabet size to constant, while at the same time ensuring small list size and low decoding complexity where the exponent of the polynomial run time does not depend on $\eps$. Our folded algebraic-geometric subcodes achieve a list size matching the fully random constructions up to constant factors, together with alphabet size not much worse than the lower bound $\exp(\Omega(1/\Ge))$. On the last line, our  algebraic-geometric subcodes give a deterministic list decoding with almost constant list size and optimal decoding radius.

\subsection{Subsequent works and open questions}
The challenge of decoding up to radius approaching the optimal bound $(1-R)$ with rate $R$ along with good list and alphabet size is, for the most part, solved by our work. There are still some goals that have not been met. One is to get a fully deterministic construction with constant list-size and alphabet size (as a function of $\eps$), and construction/decoding complexity $O_\eps(N^c)$. This has been almost achieved by Kopparty, Ron-Zewi,  Saraf, and Wootters~\cite{KRS-pc}. They prove that the list-size for list-decoding folded Reed-Solomon codes is itself, without any pruning by subspace evasive sets, bounded by a constant. They then combine it with several other tools from algebraic coding theory and pseudorandomness to construct codes of rate $R$ list-decodable up to a $(1-R-\eps)$ error fraction with
constant list and alphabet size (depending only on $\eps$) and decoding complexity $O_\eps(N^c)$ (in fact the exponent $c$ can be made arbitrarily close to $1$).
An exciting recent result by Guo and Ron-Zewi~\cite{GuoRonZewi20} achieves both constant list-size and alphabet within our framework, via improved subspace evasive sets for the ultra-periodic subspaces output by the list decoder.

Another challenge is to construct a $(1-R-\eps,L)$-list decodable code of rate $R$ (for list size $L$ bounded by a polynomial in the block length), over an alphabet of size $\exp(O(1/\eps))$, which is the asymptotically optimal size. All known constructions over a constant-sized alphabet known so far have alphabet size at least $\exp(\Omega(1/\eps^2))$. Finally, the various algebraic and expander-based techiques that have led to progress on list decoding only work over large alphabets. The challenge of efficient optimal rate list decoding over say the binary alphabet, even for the simpler model of erasures, remains wide open. The best known constructions are obtained via concatenation, and are list-decodable up to the so-called Blokh-Zyablov bound~\cite{GR-BlokhZyablov}.

\begin{center}
\begin{figure}
\begin{tabular}{|c|c|c|c|c|c|}
\hline
Code & Construction & Alphabet size & List size & Decoding time &  Reference \\
\hline
Folded RS/derivative  & Explicit & $N^{O(1/\eps^2)}$ & $N^{O(1/\eps)}$ & $N^{O(1/\eps)}$ & \cite{GR-FRS,GW-tit13} \\
Folded RS subcode & Randomized &  $N^{O(1/\eps^2)}$ &$O(1/\eps)$ & $N^{O(1/\eps)}$ & \cite{GW-tit13} \\
{\bf Folded RS subcode} & Explicit & $N^{O(1/\eps^2)}$ & $(1/\eps)^{O(1/\eps)}$ & $N^{O(1)} 2^{1/\eps^{O(1)}}$ & \cite{DL-stoc12} \\
Folded cyclotomic & Explicit$^{\ast}$ & $(\log N)^{O(1/\eps^2)}$ & $N^{O(1/\eps^2)}$ & $N^{O(1/\eps^2)}$ & \cite{Gur-cyclo-ANT} \\
{\bf Folded AG subcode} & Randomized & $\exp(\tilde{O}(1/\eps^2))$ & $O(1/\eps)$ & $N^{O(1)}  2^{1/\eps^{O(1)}}$ &Thm. \ref{thm:main-intro}(i) \\
{\bf AG subcode} & Explicit & $\exp(\tilde{O}(1/\eps^2))$ & $2^{2^{2^{(\log^* N)^2}}}$ & $N^{O(1)} (1/\eps)^{O(1)}$ & Thm. \ref{thm:main-intro}(ii)  \\
\hline
\end{tabular}
\caption{{\small $N$ in the above table stands for the length of codes. Parameters of various constructions of codes that enable list decoding $(1-R-\eps)$ fraction of errors, with rate $R$. The last two lines are from this work. ``Explicit" means the code can be constructed in deterministic polynomial time (the $^{\ast}$ for folded cyclotomic is because of requirement of an irreducible polynomial of high degree, which can be sampled and then checked (for a "Las Vegas" construction)). The rows with first column in boldface are not dominated by other constructions. The last line gives the first deterministic construction of algebraic codes for efficient optimal rate list decoding over constant-sized alphabets.}}
\label{fig:comparison-table}
\end{figure}
\end{center}

\subsection{Organization} The paper is organized as follows. In Section \ref{sec:tech}, we describe the detailed techniques of our paper including algebraic approaches and pseudorandomness. Following  the section on techniques, in Section \ref{sec:periodic} we introduce periodic and ultra-periodic subspaces, give definitions and basic properties. In Section \ref{sec:FF}, we recall some basic results on function fields and algebraic-geometric codes.
To illustrate our ideas in an algebraically simpler (and perhaps more
practical) setting, in Section \ref{sec:hermitian} we give a
construction based on a tower of Hermitian field
extensions~\cite{Shen93}. This is capable of giving a similar result to our best ones based on the Garcia-Stichtenoth tower, albeit with
alphabet size and list-size upper bound polylogarithmic in the code
length.
In Section \ref{sec:hse} we first introduce hierarchical
subspace-evasive sets, then show that random sets are hierarchical
subspace-evasive with high probability. We also present a pseudorandom
construction of hierarchical subspace-evasive sets, which also allow for efficient encoding and efficient computation of intersection with periodic subspaces.

Folded algebraic-geometric codes from the Garcia-Stichtenoth tower are studied
in Section \ref{sec:FGS}. The list size, decoding radius and decoding
algorithm via local expansion are also discussed in this
section. Section \ref{sec:PL} is devoted to the discussion of pruning down
the list size for folded codes from both the Hermitian and the
Garcia-Stichtenoth towers using hierarchical subspace-evasive
sets. The second class of our codes, namely usual algebraic-geometric
codes with evaluation points over subfields is presented in Section
\ref{sec:lin-rs}. In this section, we first discuss list decoding for
the simpler Reed-Solomon case, and then generalize it to list decoding
of arbitrary algebraic geometric codes and finally instantiate the approach with the codes
from the Garcia-Stichtenoth tower. In Section \ref{sec:sd},
we introduce subspace designs and cascaded subspace designs, and discuss parameters of random and
explicit constructions of those. In the last section, the explicit construction
of subcodes of RS and AG subcodes based on subspace designs is
presented.

\section{Our techniques}\label{sec:tech}
We describe some of the main new ingredients that go into our work. We need both new algebraic insights and constructions, as well as ideas in pseudorandomness relating to (variants of) subspace-evasive sets. We describe these in turn below.

\subsection{Algebraic ideas}\label{subsec:AI}

It is shown in \cite{GS99} that one can list decode the usual algebraic-geometric codes up to the Johnson bound. On the other hand, one has not found list decoding algorithms of the usual algebraic-geometric codes beyond the Johnson bound. Thus, to list decode the algebraic-geometric codes beyond the Johnson bound, it is natural to consider some  variants of usual algebraic-geometric codes as one does for Reed-Solomon codes \cite{GR-FRS}. In this work, we present two new variants of algebraic-geometric codes--folded algebraic geometric codes and usual algebraic geometric codes with evaluation points over subfields. We describe these in turn.

\subsubsection{Folding AG codes}

The first approach is to use suitable automorphisms of function fields to fold the code. This approach was used for Reed-Solomon codes in \cite{GR-FRS} and for cyclotomic function field in \cite{Gur-cyclo-ANT}, though this was done using the original approach in \cite{GR-FRS} where the messages to be list decoded were pinned down to the roots of a higher degree polynomial over a large residue field.
As mentioned earlier, effecting this ``non-linear" approach in \cite{GR-FRS,Gur-cyclo-ANT} with automorphisms of more general function fields seems intricate at best. In this work we employ the linear-algebraic list decoding method of \cite{GW-tit13}. However, the correct generalization of the linear-algebraic list decoding approach to the function field case is also not obvious. One of the main algebraic insights in this work is noting that a possible way to generalize the linear-algebraic approach to codes based on algebraic function fields is to rely on the {\em local power series expansion} of functions from the message space at a suitable rational point.  (The case for Reed-Solomon codes being the expansion around $0$, which is a finite polynomial form.)
%
%

Working with a suitable automorphism which has a ``diagonal" action on
the local expansion lets us extend the linear-algebraic decoding
method to AG codes (here by  a ``diagonal" action, we mean that this action gives rise to equations on coefficients of a polynomial that are diagonal). Implementing this for specific AG codes requires
an explicit specification of a basis for an associated message
(Riemann-Roch) space, and the efficient computation of the local
expansion of the basis elements at a special rational point on the
curve. We show how to do this for two towers of function fields: the
Hermitian tower~\cite{Shen93} and the asymptotically optimal
Garcia-Stichtenoth tower~\cite{GS95,GS96}. The former tower is quite
simple to handle --- it has an easily written down explicit basis, and
we show how to compute the local expansion of functions around the
point with all zero coordinates. However, the Hermitian tower does not
have bounded ratio of the genus to number of rational points, and so
does not give constant alphabet codes (we can get codes over an
alphabet size that is polylogarithmic in the block length
though). Explicit basis for Riemann-Roch spaces of the
Garcia-Stichtenoth tower were constructed in \cite{SAKSD01}. Regarding
local expansions, one major difference is that we work with local
expansion of functions at the point at infinity, which is fully
``ramified" in the tower. For both these towers, we find and work with
a nice automorphism that acts diagonally on the local expansion, and
use it for folding the codes and decoding them by solving a linear
system.

\subsubsection{Restricting evaluation points to a subfield}

The second approach is to work with ``normal" algebraic-geometric codes, based on evaluating functions from a Riemann-Roch space at some rational places, except we use a constant field extension of the function field for the function space, but restrict to evaluating at rational places over the original base field.
Let us give a brief idea why restricting evaluation points to a subfield enables correcting more errors.
The idea behind list decoding results for folded RS (or derivative) codes in \cite{GR-FRS,GW-tit13} is that the encoding of a message polynomial $f \in \F_Q[X]$ includes the values of $f$ and closely related polynomials at the evaluation points. Given a string not too far from the encoding of $f$, one can use this property together with the ``interpolation method" to find an algebraic condition that $f$ (and its closely related polynomials) must satisfy, eg. $A_0(X) + A_1(X) f(X) + A_2(X) f^q(X) + \cdots + A_s(X) f^{q^{s-1}}(X)\equiv 0\pmod{x^{q-1}-\Gg}$ in the case of folded Reed-Solomon codes~\cite{GR-FRS} (here  $\Gg$ is a primitive element of $\F_q$, and the $A_0,A_1,\dots,A_s$ are low-degree polynomials found by the decoder). The solutions $f(X)$ to this equation form an affine space, which can be efficiently found (and later pruned for list size reduction when we pre-code messages into a subspace-evasive set).

For Reed-Solomon codes as in Definition \ref{def:RS}, the encoding only includes the values of $f$ at $\Ga_1,\Ga_2,\dots,\Ga_n$. But since $\Ga_i \in \F_q$, we have $f(\Ga_i)^q = f^\sigma(\Ga_i)$ where $f^\sigma$ is the polynomial obtained by the action of the Frobenius automorphism that maps $y \mapsto y^q$ on $f$ (formally, $f^\sigma(X) = \sum_{j=0}^{k-1} f_j^q X^j$ if $f(X) = \sum_{j=0}^{k-1} f_j X^j$). Thus the decoder can ``manufacture" the values of $f^\sigma$ (and similarly $f^{\sigma^2}, f^{\s^3}$, etc.) at the $\Ga_i$. Applying the above approach then enables finding a relation $A_0(X) + A_1(X) f(X) + A_2(X) f^\s(X) + \cdots + A_s(X) f^{\s^{s-1}}(X)=0$, which is again an $\F_q$-linear condition on $f$ that can be used to solve for $f$. We remark here that this approach can also be applied effectively to linearized polynomials, and can be used to construct variants of Gabidulin codes that are list-decodable up to the optimal $1-R$ fraction of errors (where $R$ is the rate) in the rank metric~\cite{GWX16}.

To extend this idea to algebraic-geometric codes, we work with constant extensions $\F_{q^m} \cdot F$ of algebraic function fields $F/\F_q$. The messages belong to a Riemann-Roch space over $\Fm$, but they are encoded via their evaluations at $\F_q$-rational points. For decoding, we recover the message function $f$ in terms of the coefficients of its local expansion at some rational point $P$. (The Reed-Solomon setting is a special case when $F =\F_q(X)$, and $P$ is $0$, i.e., the zero of $X$.)
To get the best trade-offs, we use AG codes based on a tower of function fields due to Garcia and Stichtenoth~\cite{GS95,GS96} which achieve the optimal trade-off between the number of $\F_q$-rational points and the genus. For this case, we recover messages in terms of their local expansion around the point at infinity $\Pin$ which is also used to define the Riemann-Roch space of messages. So we treat this setting separately (Section \ref{subsec:decoding-gs}), after describing the framework for general AG codes first.

\subsection{Pseudorandomness}
The above algebraic ideas enable us to pin down the messages into a
structured subspace of dimension linear in the message length. The specific structure of the subspace is a certain ``periodicity" --- there is a subspace $W \subset \F_q^m$ such that once $f_0,f_1,\dots,f_{i-1}$ (the first $i$ coefficients of the message polynomial) are fixed, $f_i$ belongs to a coset of $W$. We now describe our ideas to prune this list, by restricting (or ``pre-coding") the message polynomials to belong to carefully constructed pseudorandom subsets that have small intersection with any periodic subspace.

\subsubsection{Hierarchical subspace-evasive sets}
 The first approach  follows along the lines of \cite{GW-tit13} and we only encode messages in a
{\em subspace-evasive set} which has small intersection with low-dimensional
subspaces. Implementing this in our case, however, leads to several
problems. First, since the subspace we like to avoid intersecting much
has large dimension, the list size bound will be linear in the code
length and not a constant  like in our final result (by ``a constant", we mean that the list size is independent of the code length and dependent on $\Ge$). More severely, we
cannot go over the elements of this subspace to prune the list as that
would take exponential time. To solve the latter problem, we observe
that the subspace has a special ``periodic" structure, and exploit
this to show the existence of large ``hierarchically subspace evasive"
(h.s.e) subsets which have small intersection with the projection of
the subspace on certain prefixes.  Isolating the periodic property of
the subspaces, and formulating the right notion of evasiveness w.r.t
to such subspaces, is an important aspect of this work. 

We also give a construction of good h.s.e sets using limited wise independent sample spaces, in a manner enabling the efficient iterative computation of the final list of intersecting elements. Further our construction allows for efficient indexing into the h.s.e set which leads to
 an efficient encoding algorithm for our code). As a further ingredient, we note that the number of possible subspaces that arise in the decoding is much smaller than the total number of possibilities. Using this together with an added trick in the h.s.e set construction, we are able to reduce the list size to a constant.

\subsubsection{Subspace designs}

The approach based on h.s.e sets leads to excellent list size; however, we only know randomized constructions of h.s.e sets with the required properties. Our second approach to prune the subspace of possible solutions is based on {\em subspace designs} and leads to deterministic subcode constructions.
More precisely speaking, the coefficients $f_0,f_1,\dots,f_{k-1}$ of the message polynomial (which belong to the extension field $\F_{q^m}$) are pinned down by the linear-algebraic list decoder to a periodic subspace with the property that there is an $\F_q$-subspace $W \subset \F_{q^m}$ such that once $f_0,f_1,\dots,f_{i-1}$ are fixed, $f_i$ belongs to a coset of $W$.
Our idea then is to restrict $f_i$ to belong to a subspace $H_i$ where $H_1,H_2,\dots,H_k$ are a collection of subspaces in $\F_q^m$ such that for any $s$-dimensional subspace $W \subset \F_q^m$, only a small number of them have non-trivial intersection with $W$. More precisely, we require that $\sum_{i=1}^k \dim(W \cap H_i)$ is small. We call such a collection $\{H_i\}_{i=1}^k$ as a {\em subspace design} in $\F_q^m$. We feel that the concept of subspace designs is interesting in its own right, and view the introduction of this notion in Section \ref{sec:sd} as a key contribution in this work. Indeed, subsequent work by Forbes and Guruswami~\cite{FG-random15} highlighted the central role played by subspace designs in
``linear-algebraic pseudorandomness'' and in particular how they lead to rank condensers and dimension expanders.

A simple probabilistic argument shows that, with high probability, any   $q^{\Omega(\eps m)}$  subspaces of dimension $(1-\eps) m$ that are randomly chosen have small total intersection with every $s$-dimensional $W$. This construction can also be derandomized, though the construction complexity of the resulting codes becomes quasi-polynomial with this approach for the parameter choices needed in the construction.

Fortunately, in a follow-on to \cite{GX-stoc13}, Guruswami and
Kopparty gave explicit constructions of subspace designs with
parameters nearly matching the random
constructions~\cite{GK-combinatorica}.  One can pre-code with this
subspace design to get explicit list-decodable sub-codes of
\emph{Reed-Solomon codes} whose evaluation points are in a subfield
(Section~\ref{subsec:put-together-rs}). However, this construction
inherits the large field size of Reed-Solomon codes.


For explicit subcodes of algebraic-geometric codes using subspace
designs we need additional ideas. The dimension $k$ in the case of AG
codes is much larger than the alphabet size $q^m$ (in fact that is the
whole point of generalizing to AG codes). So we cannot have a subspace
design in $\F_q^m$ with $k$ subspaces. We therefore use several
``layers" of subspace designs in a cascaded fashion (Section
\ref{subsec:csd}) --- the first one in $\F_q^m$, the next one in
$\F_q^{m_1}$ for $m_1 \gg q^{\sqrt{m}}$, the third one in $\F_q^{m_2}$
for $m_2 \gg q^{\sqrt{m_1}}$ and so on.  Since the $m_i$'s increase
exponentially, we only need about $\log^* k$ levels of subspace
designs. Each level incurs about a factor $1/\eps$ increase in the
dimension of the ``periodic subspace" ($W$ when we begin) at the
corresponding scale.  With a careful technical argument and choice of
parameters, we are able to obtain the bounds of Theorem
\ref{thm:main-intro}(ii).

\section{Periodic subspaces}
\label{sec:periodic}
In this section we formalize a certain ``periodic" property of affine subspaces that will arise in our list decoding application.


We begin with some notation.
For a vector $\mv{y}=(y_1,y_2,\dots,y_m)^T \in \F_q^m$ and positive
integers $t_1 \le t_2 \le m$, we  denote by
$\proj_{[t_1,t_2]}(\mv{y}) \in \F_q^{t_2-t_1+1}$ its projection onto
coordinates $t_1$ through $t_2$, i.e.,
$\proj_{[t_1,t_2]}(\mv{y})=(y_{t_1},y_{t_1+1},\dots,y_{t_2})^T$. When
$t_1=1$, we use $\proj_t(\mv{y})$ to denote
$\proj_{[1,t]}(\mv{y})$. By default, we treat vectors as column vectors. These notions are extended to subsets of
strings in the obvious way: $\proj_{[t_1,t_2]}(S) = \{
\proj_{[t_1,t_2]}(\mv{x}) \mid \mv{x} \in S\}$.

For an affine space $H$, its {\em underlying subspace} is the subspace $S$ such that $H$ is a coset of $S$.

\begin{defn}[Periodic (affine) subspaces]
	\label{def:periodic-subspaces}
	For positive integers $r,b,\period$ with $r < \period$ and $\dims := b\period$, an affine subspace $H \subset \F_q^\dims$ is said to be $(r,\period,b)$-periodic if there exists a matrix $B \in \F_q^{\period \times \period}$ whose kernel $\mathrm{ker}(B)$ has dimension at most $r$, and vectors $\mv{a}_\ell \in \F_q^\period$ and matrices $A_\ell \in \F_q^{\period \times (\ell-1)\period}$ for $1 \le \ell \le b$, such that every $\mv{x} \in H$ satisfies the following equations
	for $\ell=1,2,\dots,b$:
		\begin{equation}
	\label{eqn:proj-period}
		\mv{a}_\ell +  A_{\ell} \cdot \proj_{(\ell-1)\period}(\mv{x})  + B  \cdot \proj_{[(\ell-1)\period+1,\ell \period]}(\mv{x})  = 0  \ .
		\end{equation}
In other words, the projections of the subspace onto blocks of
contiguous $\period$ symbols, conditioned on any prefix, always belong
to an affine shift of the subspace $W := \mathrm{ker}(B)$ of dimension
at most $r$.

For dimensions $\dims$ not necessarily divisible by $\period$, we say that an affine subspace $H \subseteq \F_q^{\dims}$ is $(r,\Delta)$-periodic if there is exists a $(r,\period,b)$-periodic subspace $H' \subseteq \F_q^{b\period}$ for $b = \lceil \frac{\dims}{\period} \rceil$ such that $H = \proj_{[1,\dims]}(H')$.

We will call $W$ the \emph{recurring subspace} of the periodic subspace $H$.


\end{defn}

\begin{defn}[Representing periodic affine subspaces]
  \label{def:periodic-rep}
  The matrices $A_i$ and vectors $\mv{a_i}$, $i=1,2,\dots,b$, and the matrix $B$, or equivalently the system of equations \eqref{eqn:proj-period}, can be used to specify the $(r,\period,b)$-periodic subspace $H$, and this is the representation of periodic subspaces that will naturally arise in our list decoders.

\end{defn}

The motivation for the above definition will be clear when we present
our linear-algebraic list decoders, which will pin down the messages
that must be output within an $(s-1,m,k)$-periodic (affine)
subspace. (Here $q^m$ will be the alphabet size of the code, $k$ its
dimension, and $s$ will be a parameter of the algorithm that governs
how close the decoding performance approaches the Singleton bound.)

The following properties of periodic affine spaces follow directly from the definition.
\begin{claim}
\label{clm:periodic-subspace}
Let $H$ be an $(r,\period,b)$-periodic affine subspace. Then for each  $j=1,2,\dots,b$,
\begin{enumerate}
\item the projection of $H$ to the first $j$ blocks of $\period$ coordinates, $\proj_{j\period}(H) = \{ \proj_{j\period}(\mv{x})  \mid \mv{x}\in H \}$, has dimension at most $j r$. (In particular $H$ has dimension at most $b r$.)
\item for each $\mv{a} \in \F_q^{(j-1)\period}$, there are at most $q^r$ extensions $\mv{y} \in \proj_{j\period}(H)$ such that $\proj_{(j-1)\period}(\mv{y})=\mv{a}$.
\end{enumerate}
\end{claim}

\smallskip \noindent {\bf Ultra-periodic subspaces.} For our result on pre-coding algebraic-geometric codes with subspace designs, we will exploit an even stronger property that holds for the subspaces output by the linear-algebraic list decoder. We formalize this notion below.

\begin{defn}[Ultra-periodic subspace]
\label{def:ultra-periodic}
	For positive integers $r,b,\period$ with $r<\period$, an affine subspace $H$ of $\F_q^{\dims}$ for $\dims=b\period$ is said to be
$(r,\period,b)$-{\em ultra periodic} if there exist vectors $a_{\ell}\in \F_q^{\period}$ and matrices $B_{\ell} \in \F_q^{\period \times \period}$ with $\dim(\mathrm{ker}(B_\ell))\le r$ for $\ell=1,2,\dots,b$, such that every $\mv{x} \in H$ satisfies the following equations for  $\ell=1,2,\dots,b$:
\begin{equation}
\label{eqn:ultra-period}
\mv{a}_\ell +  \sum_{i=1}^{\ell} B_{\ell-i+1}  \cdot \proj_{(i-1)\period+1,i \period}(\mv{x})   = 0  \ .
\end{equation}
In other words, the space $H$ is defined by equations that have a lower-triangular ``Toeplitz"  block-diagonal structure, with the blocks on the diagonal being $B_1$, the blocks on the next lower diagonal being $B_2$, the next diagonal having $B_3$, and so on.

For ambient dimensions $\dims$ not necessarily divisible by $\period$,
we say that an affine subspace $H \subseteq \F_q^{\dims}$ is
$(r,\Delta)$-ultra periodic if there is exists a $(r,\period,b)$-ultra
periodic subspace $H' \subseteq \F_q^{b\period}$ for $b = \lceil
\frac{\dims}{\period} \rceil$ such that $H = \proj_{[1,\dims]}(H')$.
\end{defn}

%

We have the below observation that follows from the definition of ultra-periodicity.
\begin{obs}
\label{obs:ultra-period-scales}
If a subspace $H$ of $\F_q^{\dims}$ is $(r,\period)$-{\em ultra periodic}, then for every integer $\ell$, $1 \le \ell \le \frac{\dims}{\period}$, $H$ is $(\ell r, \ell \period)$-periodic.
\end{obs}
 Thus ultra-periodicity captures the fact that the subspace is periodic not only for blocks of size $\period$, but also for block sizes that are multiples of $\period$. Thus the subspace looks periodic in multiple ``scales" simultaneously. As with periodic subspaces, an ultra-periodic subspace is defined by
equations of the form \eqref{eqn:ultra-period}, and this is how we
will specify the subspace.

\section{Preliminaries on function fields and algebraic-geometric codes}\label{sec:FF}
For convenience of the reader, we start with some background on global function fields over finite fields. The reader may refer to \cite{stich-book,NX01} for detailed background on function fields and algebraic-geometric codes.

\subsection{General background on function fields}

For a prime power $q$, let $\F_q$ be the finite field of $q$ elements. An { algebraic function field}  over
$\F_q$ in one variable is a field extension $F \supset \F_q$ such
that $F$ is a finite algebraic extension of  $\F_q(x)$ for some
$x\in F$ that is transcendental over $\F_q$. The field $\F_q$ is called the full constant field of $F$ if the algebraic closure of $\F_q$  in $F$ is $\F_q$ itself. Such a function field is also called a global function field. From now on, we always denote by  $F/\F_q$  a function field $F$ with the full constant field $\F_q$.

\subsubsection{Valuations, Places, and Divisors}

A discrete valuation of $F/\F_q$ is a map from $F$ to $\ZZ\cup\{+
\infty\}$ satisfying certain properties (see \cite[Definition 1.19]{stich-book}). Then each discrete valuation $\nu$ from $F/\F_q$ to $\ZZ\cup\{+
\infty\}$ defines a valuation ring $O=\{f\in F:\; \nu(f)\ge 0\}$ that is a local ring \cite[Theorem 1.1.13]{stich-book}. The maximal ideal $P$ of $O$ is  given by $P=\{f\in F:\; \nu(f)> 0\}$ and it is called a {\it place}. We denote  the valuation $\nu$ and the local ring $O$ corresponding to $P$ by $\nu_P$ and $O_P$, respectively. The residue class field $O_P/P$, denoted by $F_P$, is a finite extension of $\F_q$. The extension degree $[F_P:\F_q]$ is called {\it degree} of $P$, denoted by $\deg(P)$. A place of degree one is called a {\it rational} place. For a nonzero function $z\in F$, the principal divisor  of $z$ is defined to be
${\rm div}(z)=\sum_{P\in\PP_F}\nu_P(z)P$. The zero and pole divisors of $z$ are defined to be ${\rm div}(z)_0=\sum_{\nu_P(z)>0}\nu_P(z)P$ and ${\rm div}(z)_{\infty}=-\sum_{\nu_P(z)<0}\nu_P(z)P$, respectively. Then we have $\deg({\rm div}(z))=0$, i.e, $\deg({\rm div}(z)_0)=\deg({\rm div}(z)_\infty)$.
For two functions $f,g\in F$ and a place $P$, we have $\nu_P(f+g)\ge \min\{\nu_P(f),\nu_P(g)\}$ and the equality holds if $\nu_p(f)\neq\nu_P(g)$ (note that $\nu_P(0)=+\infty$). This implies that $f+g\neq 0$ if $\nu_P(f)\neq\nu_P(g)$.

If $F$ is the rational function field $\F_q(x)$, then every discrete valuation of $F/\F_q$  is given by  either $\nu_{\infty}$ or $\nu_{p(x)}$ for an irreducible polynomial $p(x)$, where $\nu_{\infty}$ is defined by $\nu_{\infty}(f/g)=\deg(g)-\deg(f)$  and $\nu_{p(x)}(f/g)=a-b$ with $p(x)^a||f$ and  $p(x)^b||g$ for two nonzero polynomials $f,g\in\F_q[x]$. It is straightforward to verify that the degrees of places   corresponding to $\nu_{\infty}$ and  $\nu_{p(x)}$ are $1$ and $\deg(p(x))$, respectively.

Let $\PP_F$ denote the set of places of $F$. The divisor group, denoted by ${\rm Div}(F)$, is the free abelian group generated by all places in $\PP_F$. An element $G=\sum_{P\in\PP_F}n_PP$ of ${\rm Div}(F)$ is called a divisor of $F$, where $n_P=0$ for almost all $P\in\PP_F$. We denote $n_p$ by $\nu_P(G)$. The support, denoted by $\Supp(G)$, of $G$ is the set $\{P\in\PP_F:\; n_P\neq 0\}$. Thus, $\Supp(G)$ of a divisor $G$ is always a finite subset of $\PP_F$.

\subsubsection{Constant field extension}
One of our code constructions will be based on evaluations of functions at rational points over a \emph{subfield}. For this purpose, we will work with constant field extensions over $\F_{q^m}$ of a function field over a base field $\F_q$. We describe these now.

Let $F/\F_q$ be a function field. Fix an algebraic closure $\bar{F}$ of $F$. Then $\bar{F}$ contains the algebraic closure $\bar{\F}_q=\cup_{i=1}^{\infty}\F_{q^i}$ as well. Hence, for $m\ge 1$,  $\bar{F}$ contains the extension field $\F_{q^m}$ of $\F_q$.   The composite field $F_m:=\F_{q^m}\cdot F$ is defined to be the smallest subfield of $\bar{F}$ that contains both $F$ and $\F_{q^m}$. Then we have the following facts (see \cite[Propositions 3.6.1 and 3.6.3]{stich-book}):
\begin{itemize}
	\itemsep=0.5ex
\item[(i)] the full constant field of $F_m$ is $\F_{q^m}$;
\item[(ii)] each subset of $F$ that is linearly independent over $\F_q$ remains so over $F_m$;
\item[(iii)] $[F_m:\F_{q^m}(x)]=[F:\F_q(x)]$ for any $x\in F\setminus \F_q$;
\item[(iv)] a place $P$ of $F$ of degree $d$ splits into $\gcd(m,d)$ places of $F_m$ of degree $d/\gcd(m,d)$ (in the case of rational function fields, this means that an irreducible polynomial over $\F_q$ of degree $d$ is factorized into product of $\gcd(m,d)$ irreducible polynomials over $\F_{q^m}$ of degree $d/\gcd(m,d)$);
    \item[(v)] genus of $F_m$ is equal to genus of $F$.
\end{itemize}
A divisor $G=\sum_{P\in\PP_F}n_PP$ of $F$ can be viewed as the divisor $\sum_{P\in\PP_F}\sum_{P'|P}n_PP'$
of $F_m$. We still denote this divisor of $F_m$ by $G$. By (iv) of the above facts, a rational place $P$ of $F$ continues to be a rational place $P'$ of $F_m$. The valuation ring of $P's$ is the tensor product of $O_P$ with $\F_{q^m}$, i.e, $O_{P'}=O_P\otimes_{\F_q}\F_{q^m}$. If there is no confusion, we still denote $P'$ by $P$.

\subsubsection{Riemann-Roch spaces} For a divisor $G$ of $F/\F_q$, we define the \emph{Riemann-Roch space} associated with $G$ by
\[\mL(G):=\{f\in F^*:\; {\rm div}(f)+G\ge 0\}\cup\{0\},\]
where $F^*$ denotes the set of nonzero elements of $F$.
Then $\mL(G)$ is a finite dimensional space over $\F_q$ and its
dimension $\ell(G)$ is determined by the Riemann-Roch theorem which
gives
\[\ell(G)=\deg(G)+1-\g+\ell(W-G),\]
where $\g$ is the genus of $F$ and $W$ is a canonical divisor of degree $2\g-2$. Therefore, we
always have that $\ell(G)\ge \deg(G)+1-\g$ and the equality holds
if $\deg(G)\ge 2\g-1$ \cite[Theorems 1.5.15 and 1.5.17]{stich-book}.

Consider the finite extension $\F_{q^m}$ over $\F_q$ and the constant extension $F_m:=\F_{q^m}\cdot F$ over $F$. As a divisor $G$ of $F$ can be viewed as a divisor of $F_m$,  we can consider the Riemann-Roch space in $F_m$ given by
\[\mL_m(G):=\{f\in F_m^*:\; {\rm div}(f)+G\ge 0\}\cup\{0\}.\] Then it is clear that $\mL_m(G)$ contains  $\mL(G)$ and $\mL_m(G)$ is a finite dimensional vector space over $\F_{q^m}$. Furthermore, $\cL_m(G)$ is the tensor product of $\cL(G)$ with $\Fm$ (see \cite[Proposition 5.8 of Chapter II]{Sil85}). This implies that
\[\dim_{\Fm}(\cL_m(G))=\dim_{\F_q}(\cL(G))\]
and an $\F_q$-basis of $\cL(G)$ is also an $\Fm$-basis of $\cL_m(G)$.

\subsubsection{Automorphisms}  The automorphisms of the function field $F$ that fix $\F_q$ are denoted by $\Aut(F/\F_q)$. For an automorphism $\phi\in \Aut(F/\F_q)$ and
and a function $f\in F$, we denote by $f^{\phi}$ the action of $\phi$ on $f$.
 For a place $P$, define a map $\nu_{P^\phi}$ from $F$ to $\ZZ\cup\{+
\infty\}$ given by $f\mapsto \nu_P(f^{\phi^{-1}})$. Then one can show that $\nu_{P^\phi}$  indeed satisfies  the properties given in \cite[Definition 1.19]{stich-book} and hence it is a discrete valuation.  The valuation ring $O_{P^\phi}$ of $\nu_{P^\phi}$ is given by
{\small \[\{h\in F:\; \nu_{P^\phi}(h)\ge 0\}=\{h\in F:\; \nu_{P}(h^{\phi^{-1}})\ge 0\}\stackrel{ h=f^\phi}{=}\{f^\phi\in F:\; \nu_{P}(f)\ge 0\}=\{f^\phi:\; f\in O_P\}.  \]}
 and the maximal ideal of this valuation ring is $\{\phi(x):\; x\in P\}$. Therefore, this maximal ideal is a place of $F$, denoted by $P^\phi$. Moreover, $\phi$ induces an $\F_q$-isomorphism between the residue fields $F_P$ and $F_{P^\phi}$. Hence, we have $\deg(P)=\deg(P^\phi)$.

For a function $f$ and a rational place $P\in\PP_F$ with $\nu_P(f)\ge 0$, we denote by $f(P)$ the residue class of $f$ in the residue class field $F_P$ at $P$. If $\nu_P(f)\ge 0$ and $\nu_{P^{\phi}}(f)\ge 0$, then one has that $\nu_P(f^{\phi^{-1}})\ge 0 $.  Furthermore, there is an $\F_q$-isomorphism between $O_P$ and  $O_{P^\phi}$ given by $f\mapsto f^\phi$. This induces the identity map between $F_P=\F_q$ and $F_{P^\phi}=\F_q$. Hence,  $f(P)=f^\phi(P^\phi)$. Replacing $f$ by $f^{\phi^{-1}}$ gives $f(P^{\phi})=f^{\phi^{-1}}(P)$.

 For a divisor $G=\sum_{P\in\PP_F}m_PP$ we denote by $G^{\phi}$ the divisor $\sum_{P\in\PP_F}m_PP^{\phi}$. Therefore, we have
 \[\phi(\mL(G)):=\{f^{\phi}:\; f\in\mL(G)\}=\mL(G^{\phi}).\]

{\small Assume that $E/\F_q$ is a subfield of $F$ and $\phi$ is an
 automorphism of $\Aut(F/E)$. Then for a divisor $G$ of $F$ that is
 invariant under $\phi$, we have $\phi(\mL(G))=\mL(G)$.}

Next we consider the constant extension $F_m=\F_{q^m}\cdot F$. Let $\Gs$ be the Frobenius automorphism $\F_{q^m}/\F_q$, i.e., $\Gs(\Ga)=\Ga^q$ for any $\Ga\in\F_{q^m}$. Then $\Gs$ can be extended to  an  automorphism of $\Aut(F_m/F)$ given by $\Gs(f)=f$ for any $f\in F$ and  $\Gs(\Ga)=\Ga^q$ for any $\Ga\in\F_{q^m}$. If $P$ is a rational place of $F$, then $P$ remains to be a rational place $P'$ of $F_m$ and hence $\Gs(O_{P'})=\Gs(O_P\otimes_{\F_q}\F_{q^m})=O_P\otimes_{\F_q}\F_{q^m}=O_{P'}$. Thus, we have $(P')^\Gs=P'$.

\subsection{Algebraic-geometric codes} Let $\mP=\{P_1,P_2,\dots,P_N\}$ be a set of $N$ distinct rational places of a function field $F/\F_q$ of genus $\g$. Let $G$ be a divisor of $F$ with $\Supp(G)\cap\mP=\emptyset$. Then the
algebraic-geometric code defined by
\begin{equation}
\label{eq:ag-code-defn}
C(\mP,G):=\{(f(P_1),f(P_2),\dots,f(P_N)):\; f\in\mL(G)\}
\end{equation}
is an $\F_q$-linear code of length $N$. Furthermore, the dimension of $C(\mP,G)$ is equal to $\ell(G)$ if $N>\deg(G)$.

The (generalized) Reed-Solomon codes can be realized under the above framework of algebraic-geometric codes. More precisely speaking, the Reed-Solomon codes are algebraic-geometric codes based on rational function fields. Let us give the detail on construction of the Reed-Solomon codes
under the framework of algebraic-geometric codes.

Let $F=\F_q(x)$ be a rational function field. Let $\Ga_1,\Ga_2,\dots,\Ga_n$ be $n$ distinct elements of $\F_q$. Denote by $P_i$ the unique zero of $x-\Ga_i$ for $1\le i\le n$ and put $\mP=\{P_1,P_2,\dots,P_n\}$. Let $\Pin$ be the unique pole of $x$. Put $G=(k-1)\Pin$. Then the Riemann-Roch space $\mL(G)$ is the $\F_q$-space consisting of polynomials of degree less than $k$. By definition, we have
\begin{eqnarray*}
C(\mP,G)&=&\{(f(P_1),f(P_2),\dots,f(P_N)):\; f\in\mL(G)\}\\
&=&\{(f(\Ga_1),f(\Ga_2),\dots,f(\Ga_N)):\; f\in\F_q[x],\; \deg(f)\le k-1\}.
\end{eqnarray*}

The  codes considered in this paper are variations of the above algebraic-geometric codes, namely, folded algebraic-geometric codes and algebraic-geometric codes with evaluation points in a subfield.

A folded algebraic-geometric code is a code with each coordinate being a column vector $(f(P),f(P^\Gs),\dots,f(P^{\Gs^{m-1}}))^T\in\F_q^m$ for a function $f\in \mL(G)$, a rational place $P$ and an automorphism $\Gs\in\Aut(F/\F_q)$, where $T$ stands for transpose. This is a generalization of folded Reed-Solomon codes introduced in \cite{GR-FRS}. The main reason why a folded  algebraic-geometric code is used is that once a position is transmitted correctly, then one gets $m$ correct components $(f(P),f(P^\Gs),\dots,f(P^{\Gs^{m-1}}))$. Consequently, more interpolation equations are increased and list decoding radius is enlarged (see Lemma \ref{lem:x5.4}, for instance).

Similar to folded algebraic-geometric codes, introducing   algebraic-geometric codes with evaluation points in a subfield is for purpose of increasing list decoding radius as well. We choose $N$ rational places $P_1,P_2,\dots,P_N$ of a function field $F/\F_q$ and let $\Gs$ be the Frobenius automorphism  of $\F_{q^m}/\F_q$. Then  one has $P_i^\Gs=P_i$ for all $1\le i\le N$. Thus, once we have a correct position $f(P_i)$ for some function $f\in \mL_m(G)$, we get correct information for other $m-1$  elements $f(P_i)^{\Gs^j}=f^{\Gs^j}(P_i)$ for $i=1,2,\dots,m-1$. As a result,   list decoding radius is enlarged (see  Lemma \ref{lem:Q-is-good}, for instance).

\subsection{Background on Hermitian tower}
\label{subsec:herm-prelim}
In what follows, let $r$ be a prime power and let $q=r^2$. We denote by $\F_q$ the finite field with $q$ elements. The Hermitian function tower that we are going to use for our code construction was discussed in \cite{Shen93}. The reader may refer to \cite{Shen93} for the detailed background on the Hermitian function tower. The Hermitian tower is defined by the following recursive equations
\begin{equation}\label{eq:x2}x_{i+1}^r+x_{i+1}=x_i^{r+1},\quad i=1,2,\dots,e-1.\end{equation}
Put  $F_e=\F_q(x_1,x_2,\dots,x_{e})$ for $e\ge 2$. We will assume that $r \ge 2e$.

\subsubsection{Rational places}
The function field $F_e$ has $r^{e+1}+1$ rational places. One of these is the ``point at infinity" which is the unique pole $\Pin$ of $x_1$ (and is fully ramified). The other $r^{e+1}$ come from the rational places lying over the unique zero $P_\Ga$ of $x_1-\Ga$ for each $\Ga\in\F_q$. Note that for every $\Ga\in\F_q$, $P_\Ga$ splits completely in $F_e$, i.e., there are $r^{e-1}$ rational places lying over $P_\Ga$.
Intuitively, one can think of the rational places of $F_e$ (besides $\Pin$) as being given by $e$-tuples $(\Ga_1,\Ga_2,\dots,\Ga_e)\in \F_q^e$ that satisfy $\Ga_{i+1}^r+\Ga_{i+1}=\Ga_i^{r+1}$ for $i=1,2,\dots,e-1$. For each value of $\Ga \in \F_q$, there are precisely $r$ solutions to $\beta \in \F_q$ satisfying $\beta^r + \beta = \Ga^{r+1}$, so the number of such $e$-tuples is $r^{e+1}$ ($q=r^2$ choices for $\Ga_1$, and then $r$ choices for each successive $\Ga_i$, $2 \le i  \le e$).

\subsubsection{Riemann-Roch spaces}
For an integer $l$, we consider the Riemann-Roch space defined by
\[\cL(l\Pin):=\{h\in F_e\setminus\{0\}:\; \nu_{\Pin}(h)\ge -l\}\cup\{0\}.\]
By the Riemann-Roch theorem, its dimension $\ell(l\Pin)$ is at least $l-\g_e+1$ and furthermore,
\[ \ell(l\Pin)=l-\g_e+1 \quad \text{if} \quad l\ge 2\g_e-1 \  , \]
where $\g_e$ is the genus of the function field $F_e$ given by \eqref{eq:x4} below.

A basis over $\F_q$ of $\cL(l\Pin)$ can be explicitly constructed as follows
\begin{equation}\label{eq:x3}\left\{x_1^{j_1}\cdots x_e^{j_e}:\; (j_1,\dots,j_e)\in\ZZ^e_{\ge 0},\ \sum_{i=1}^ej_ir^{e-i}(r+1)^{i-1}\le l\right\}.
\end{equation}

We stress that evaluating elements of $\cL(l\Pin)$ at the rational places of $F_e$ (other than $\Pin$) is easy: we simply have to evaluate a linear combination of the monomials allowed in \eqref{eq:x3} at the tuples $(\Ga_1,\Ga_2,\dots,\Ga_e)\in \F_q^e$ mentioned above. In other words, it is just evaluating an $e$-variate polynomial at a specific subset of $r^{e+1}$ points of $\F_q^e$, and can be accomplished in polynomial time.

\subsubsection{Genus}
The genus $\g_e$ of the function field $F_e$ is given by
\begin{equation}
\label{eq:x4}
\g_e=\frac12\left(\sum_{i=1}^{e-1}r^e\left(1+\frac1r\right)^{i-1}-(r+1)^{e-1}+1\right)\le \frac {r^e}{2}\sum_{i=1}^e{e\choose i}\frac1{r^{i-1}} \le \frac{e r^e}{2} \sum_{i=1}^e \left(\frac{e}{r}\right)^{i-1} \le e r^e
\end{equation}
where the last step used $r \ge 2e$.

\subsubsection{A useful automorphism}
Let $\Gg$ be a primitive element of $\F_q$. Then for $i\ge 1$, one has $\Gg^{r(r+1)^{i}}=\Gg^{(r^2+r)(r+1)^{i-1}}=\Gg^{(1+r)(r+1)^{i-1}}=\Gg^{(r+1)^{i}}$. Consider the automorphism $\Gs\in{\rm Aut}(F_e/\F_q)$ defined by
\[\Gs:\; x_i\mapsto\Gg^{(r+1)^{i-1}}x_i \quad \mbox{for}\ i=1,2,\dots,e.\]
Indeed, $\Gs$ defines an automorphism $\Gs\in{\rm Aut}(F_e/\F_q)$ since after action of $\Gs$ the equation \eqref{eq:x2} becomes $(\Gg^{(r+1)^{i}}x_{i+1})^r+\Gg^{(r+1)^{i}}x_{i+1}=(\Gg^{(r+1)^{i-1}}x_i)^{r+1}$, i.e., 
$x_{i+1}^r+x_{i+1}=x_i^{r+1}$ by cancelling $\Gg^{(r+1)^{i}}$ in both the sides.
The order of $\s$ is $q-1$ and furthermore, we have the following facts:
\begin{itemize}
\item[(i)] Let $P_0$ be the unique common zero of $x_1,x_2,\dots,x_e$ (this corresponds to the $e$-tuple $(0,0,\dots,0)$), and $\Pin$ the unique pole of $x_1$. The automorphism $\Gs$ keeps $P_0$ and $\Pin$ unchanged, i.e., $P_0^{\s}=P_0$ and $\Pin^{\s}=\Pin$,
\item[(ii)] Let $\PP$ be the set of all the rational places which are neither $P_{\infty}$ nor zeros of $x_1$. Then $|\PP|=(q-1)r^{e-1}$. Moreover,  $\Gs$ divides $\PP$ into $r^{e-1}$ orbits and each orbit has $q-1$ places. For an integer $m$ with  $1\le m\le q-1$,  we can label $Nm$  distinct elements    $P_1,P_1^{\s},\dots,P_1^{\s^{m-1}},\dots,P_N,P_N^{\s},\dots,P_N^{\s^{m-1}}$ in $\PP$, as long as $N \le r^{e-1}\left\lfloor\frac{q-1}m\right\rfloor$.
\end{itemize}

\subsection{Background on Garcia-Stichtenoth tower}
\label{subsec:gs-prelim}
Again let $r$ be a prime power and let $q=r^2$. We denote by $\F_q$
the finite field with $q$ elements. The Garcia-Stichtenoth towers that
we are going to use for our code construction were discussed in
\cite{GS95,GS96}. The reader may refer to \cite{GS95,GS96} for the
detailed background on the Garcia-Stichtenoth function tower.  There
are two optimal Garcia-Stichtenoth towers that are equivalent. For
simplicity, we introduce the tower defined by the following recursive
equations \cite{GS96}
\begin{equation}\label{eq:x5}
x_{i+1}^r+x_{i+1}=\frac{x_i^{r}}{x_i^{r-1}+1},\quad i=1,2,\dots,e-1.
\end{equation}
Put  $K_e=\F_q(x_1,x_2,\dots,x_{e})$ for $e\ge 2$.

\subsubsection{Rational places}
The function field $K_e$ has at least $r^{e-1}(r^2-r)+1$ rational places. One of these is the ``point at infinity" which is the unique pole $\Pin$ of $x_1$ (and is fully ramified). The other $r^{e-1}(r^2-r)$ come from the rational places lying over the unique zero of $x_1-\Ga$ for each $\Ga\in\F_q$ with $\Ga^r+\Ga\not=0$. Note that for every $\Ga\in\F_q$ with $\Ga^r+\Ga\not=0$, the unique zero of $x_1-\Ga$ splits completely in $K_e$, i.e., there are $r^{e-1}$ rational places lying over the zero of $x_1-\Ga$. Let $\PP$ be the set of all the rational places lying over the zero of $x_1-\Ga$ for all $\Ga\in\F_q$ with $\Ga^r+\Ga\not=0$.
Then, intuitively, one can think of the $r^{e-1}(r^2-r)$ rational places in $\PP$ as being given by $e$-tuples $(\Ga_1,\Ga_2,\dots,\Ga_e)\in \F_q^e$ that satisfy $\Ga_{i+1}^r+\Ga_{i+1}=\frac{\Ga_i^{r}}{\Ga_i^{r-1}+1}$ for $i=1,2,\dots,e-1$ and $\Ga_1^r+\Ga_1\neq 0$. For each value of $\Ga \in \F_q$, there are precisely $r$ solutions to $\beta \in \F_q$ satisfying $\beta^r + \beta = \frac{\Ga^{r}}{\Ga^{r-1}+1}$, so the number of such $e$-tuples is $r^{e-1}(r^2-r)$ ($r^2-r$ choices for $\Ga_1$, and then $r$ choices for each successive $\Ga_i$, $2 \le i  \le e$).

\subsubsection{Riemann-Roch spaces}
As shown in \cite{SAKSD01}, every function of $K_e$ with a pole only at $\Pin$ has an expression of the form
\begin{equation}\label{eq:x6}x_1^a\left(\sum_{i_1=0}^{(e-2)r+1}
\sum_{i_2=0}^{r-1}\cdots\sum_{i_e=0}^{r-1}c_{\bi}h_1\frac{x_1^{i_1}x_2^{i_2}\cdots x_e^{i_e}}{\pi_2\dots\pi_{e-1}}\right),
\end{equation}
where $a\ge 0, c_{\bi}\in\F_q$, and for $1\le j<e$, $h_j=x_j^{r-1}+1$ and $\pi_j=h_1h_2\dots h_j$. Moreover, Shum et al. \cite{SAKSD01} present an algorithm running in time polynomial in $l$  that outputs a basis of  over $\F_q$ of $\cL(l\Pin)$  explicitly in the above form.

We stress that evaluating elements of $\cL(l\Pin)$ at the rational places of $\PP$  is easy: we simply have to evaluate a linear combination of the monomials allowed in \eqref{eq:x6} at the tuples $(\Ga_1,\Ga_2,\dots,\Ga_e)\in \PP$ (note that $h_i(P),\pi_j(P)\in\F_q^*$ for every $P\in \PP$). In other words, it is just evaluating an $e$-variate polynomial at a specific subset of $r^{e-1}(r^2-r)$ points of $\F_q^e$, and can be accomplished in polynomial time.

\subsubsection{Genus}
The genus $\g_e$ of the function field $K_e$ is given by
\[\g_e=\left\{\begin{array}{ll}
(r^{e/2}-1)^2&\mbox{if $e$ is even}\\
(r^{(e-1)/2}-1)(r^{(e+1)/2}-1)&\mbox{if $e$ is odd.}\end{array}
\right.\]
Thus the genus $\g_e$ is at most $r^e$. (Compare this with the $e r^e$ bound for the Hermitian tower; this smaller genus is what allows to pick $e$ as large as we want in the Garcia-Stichtenoth tower, while keeping the field size $q$ fixed.)

\subsubsection{A useful automorphism}
Let $\Gg$ be a primitive element of $\F_q$ and consider the automorphism $\Gs\in{\rm Aut}(K_e/\F_q)$ defined by
\[\Gs:\; x_i\mapsto\Gg^{r+1}x_i \quad \mbox{for}\ i=1,2,\dots,e.\]
Indeed, $\Gs$ defines an automorphism $\Gs\in{\rm Aut}(K_e/\F_q)$ since after action of $\Gs$ the equation \eqref{eq:x5} becomes $(\Gg^{r+1} x_{i+1})^r+\Gg^{r+1} x_{i+1}=\frac{(\Gg^{r+1}x_i)^{r}}{(\Gg^{r+1}x_i)^{r-1}+1}$, i.e, $x_{i+1}^r+x_{i+1}=\frac{x_i^{r}}{x_i^{r-1}+1}$ by cancelling $\Gg^{r+1}$ on both the sides (note the fact that $\Gg^{(r+1)r}=\Gg^{r+1}$ and $\Gg^{r^2-1}=1$).
The order of $\s$ is $r-1$ and furthermore, we have the following facts:
\begin{itemize}
\item[(i)] $\Gs$ keeps  $\Pin$ unchanged, i.e., $\Pin^{\s}=\Pin$;
\item[(ii)] Let $\PP$ be the set of all the rational places lying over $x_1-\Ga$ for all $\Ga\in\F_q$ with $\Ga^r+\Ga\not=0$.
Then $|\PP|=(r-1)r^{e}$. Moreover,  $\Gs$ divides $\PP$ into $r^{e}$ orbits and each orbit has $r-1$ places. For an integer $m$ with  $1\le m\le r-1$,  we can label $Nm$  distinct elements
\[ P_1,P_1^{\s},\dots,P_1^{\s^{m-1}},\dots,P_N,P_N^{\s},\dots,P_N^{\s^{m-1}} \]
 in $\PP$, as long as $N\le r^{e}\left\lfloor\frac{r-1}m\right\rfloor$.
\end{itemize}

\section{Local expansions and encoding}
\label{sec:local-exp-encoding}
%
Similar to the Laurent series expansion of a complex function $f(z)$ in the neighborhood of a complex number, one can write functions in a function field as a power series (with finitely many negative powers) around a place $P$, called the \emph{local expansion around $P$}. Local expansions play an important role in the encoding and decoding of the codes we construct, and we discuss them separately in this section.

\subsection{Local expansion at a place}
Let $F/\F_q$ be a function field and let
$P$ be a rational place. An element $t$ of $F$ is called a local parameter at $P$ if $\nu_p(t)=1$ (such a local parameter always exists) --- intuitively this is a function which has a simple zero at $P$, similar to how $(z-1)$ has a simple zero at $1$.
For a nonzero function $f\in F$ with $\nu_P(f)\ge
v$, we have
$\nu_P\left(\frac f{t^v}\right)\ge 0.$
Put
$f_v=\left(\frac f{t^v}\right)(P),$
i.e., $f_v$ is the value of the function $f/t^v$ at $P$.  Note that the function $f/t^v-f_v$ satisfies
$\nu_P\left(\frac f{t^v}-f_v\right)\ge 1,$
hence we know that
$ \nu_P\left(\frac {f-f_vt^v}{t^{v+1}}\right)\ge 0.$
Put
$f_{v+1}=\left(\frac{f-a_vt^v}{t^{v+1}}\right)(P).$
Then  $\nu_P(f-f_vt^v-f_{v+1}t^{v+1})\ge v+2$.

Assume that we have obtained a sequence $\{f_r\}_{r=v}^m$ ($m>v$)
of elements of $\F_q$ such that
$\nu_P(f-\sum_{r=v}^kf_rt^r)\ge k+1$
for all $v\le k\le m$.
Put
$f_{m+1}=\left(\frac{f-\sum_{r=v}^mf_rt^r}{t^{m+1}}\right)(P).$
Then  $\nu_P(f-\sum_{r=v}^{m+1}f_rt^r)\ge m+2$.
In this way we continue our construction of  $f_r$. Then we obtain
an infinite sequence $\{f_r\}_{r=v}^{\infty}$ of elements of $\F_q$ such
that
$
\nu_P(f-\sum_{r=v}^mf_rt^r)\ge m+1
$
for all $m\ge v$.
We summarize the above  construction in the formal expansion
\begin{equation}
\label{eq:x1}
f=\sum_{r=v}^{\infty}f_{r}t^r,
\end{equation}
which is called the \emph{local expansion} of $f$ at $P$.

It is clear that the local expansion of a function depends on the choice of the local
parameter $t$. Note that if a power series
$\sum_{i=v}^{\infty}a_it^i$ satisfies $\nu_P(f-\sum_{i=v}^ma_it^i)\ge
m+1$ for all $m\ge v$, then it is a local expansion of $f$. The above procedure shows that finding a local expansion at a rational place is very efficient as long as the computation of evaluations of functions at this place is easy.

If $f$ belongs to a Riemann-Roch space $\mL(G)$ with $\deg(G)=d$. Denote $\nu_P(G)$ by $v$,  then the first $d+1$ coefficients $a_{v},a_{v+1},\dots,a_{v+d}$ in \eqref{eq:x1} determines the function $f$. To see this, assume that $g$ is a function of $\mL(G)$ with the  first $d+1$ coefficients in its local expansion equal to those of $f$. Then we have $f-g\in\mL(G-(d+1)P)$ which is the zero vector space. This implies that $f=g$.

\subsection{Encodings using local expansion}\label{subsec:5.2local}

An algebraic-geometric code as defined in \eqref{eq:ag-code-defn} encodes messages which belong to a Riemann-Roch space $\cl(G)$. The most common instantiation, which suffices for most purposes, is to take $G= l \Pin$ for some rational place $\Pin$ (though of as the place at infinity).  For such spaces, one can compute bases for the Riemann-Roch spaces explicitly in many cases, including the Hermitian and Garcia-Stichtenoth towers as mentioned in Sections~\ref{subsec:herm-prelim} and \ref{subsec:gs-prelim}.
For $k$ linearly independent functions $g_1,g_2,\dots,g_k$ in $\cL(l \Pin)$, one can interpret a message vector $(a_1,\dots,a_k) \in \F_q^k$ as the function $f = \sum_{i=1}^k a_i g_i \in \cL(l \Pin)$ and then encode it.


For our decoding, we will actually recover the message $f \in \cL(l {\Pin})$ in terms of the coefficients of its local expansion around a rational place $P$
\begin{equation}\label{eq:f-expansion} f= x^{-\nu}(f_0 + f_1 x + f_2 x^2 + \cdots ) \end{equation}
where $x$ is a local parameter at $P$. The place $P$ may or may not
equal $P_\infty$ --- when we instantiate the algorithm of this section
for the Hermitian tower, we will use a place different than $P_\infty$
for $P$, whereas for the Garcia-Stichtenoth tower, we will use $P =
P_\infty$. The description of the algorithm and its analysis in this
section will be general and cover both cases. Let
$\nu_P({\Pin})={\nu}$ with $\nu=0$ if ${\Pin}\neq P$ and $\nu=l$ if
$P={\Pin}$. Realizing that one must work in this power series
representation is one of the key insights in this work behind the
extension of the linear-algebraic folded Reed-Solomon list decoding
algorithm~\cite{GW-tit13} to the algebraic-geometric setting.

Given this, we will find it convenient to let the message vector consist of $(f_0,f_1,\dots,f_{k-1})$ $ \in \F^k$ ($k$ being the dimension of the code), which we will then map to a function $f$ in an appropriate Riemann-Roch space. Here we denote the field by $\F$, to capture both $\F=\F_q$ and $\F=\F_{q^m}$ when we work with constant field extensions $F_m$ and seek functions $f \in \cL_m(l \Pin)$. Likewise, we use the common notation $\cL_\F(l \Pin)$ to denote $\cl(l \Pin)$ when $\F=\F_q$, and $\cl_m(l \Pin)$ when $\F=\F_{q^m}$.

If we seek a $k$-dimensional message space, it is natural to let the message functions belong to $\cL((k+\g-1)\Pin)$ which has dimension exactly $k$ by the Riemann-Roch theorem (when $k$ is at least the genus $\g$, which will always hold for our codes). However, we desire to index the messages of the code instead by the first $k$ coefficients $(f_0,f_1,\dots,f_{k-1})$ of the local expansion of the function $f$ at $P$. Therefore we require that for every $(f_0,f_1,\dots,f_{k-1})$ there is a $f \in \cL_{\F}(l {\Pin})$ whose local expansion at $P$ has the $f_i$'s as the first $k$ coefficients. We can ensure by taking a slightly larger value of $l$, namely $l = k+2\g-1$ as we argue below.
Since the genus will be much smaller than the code length, we can afford the resulting small loss in distance and list-decoding radius.

We will recover the message in terms of the coefficients of its local expansion at $P$.

\smallskip\noindent {\bf Restricting message functions using local expansions.}
In order to prune the subspace of possible solutions, we will pick a subcode that corresponds to restricting the coefficients to a carefully constructed subset of all possibilities. This requires us to index message functions in terms of the local expansion coefficients. However, not all $(k+2\g-1)$ tuples over $\F$ arise in the local expansion of functions in the $k$-dimensional subspace $\cL_\F((k+2\g-1)\Pin)$. Below we show that we can
find a $k$-dimensional subspace of $\cL_\F((k+2\g-1)\Pin)$ such that their top $k$ local expansion coefficients give rise to all $k$-tuples over $\F$.

\begin{lemma}
\label{lem:choice-basis}
There exist a set of functions $\{g_1,g_2,\dots,g_k\}$ in $\cL_\F((k+2\g-1)\Pin)$ such that the $k\times k$ matrix $A$ formed by taking the $i$th row of $A$ to be the first $k$ coefficients in the local expansion \eqref{eq:x1} for $g_i$ at $P$ is non-singular.
\end{lemma}
\begin{proof} Let $\{\psi_1,\psi_2,\dots,\psi_\g\}$ be a basis of $\cL_\F((k+2\g-1)\Pin-kP)$. Extend this basis to a basis $\{\psi_1,\psi_2,\dots,\psi_\g,g_1,g_2,\dots,g_k\}$ of $\cL_\F((k+2\g-1)\Pin)$. We claim that the functions $\{g_1,g_2,\dots,g_k\}$ are our desired functions.

Suppose that the matrix $A$ is obtained from expansion of functions $g_i$ and  it is singular. This implies that there exists elements $\{\lambda_i\}_{i=1}^k$ such that the function $\sum_{i=1}^k\lambda_i g_i$ has local expansion $\sum_{i=k}^{\infty}a_i T^i$ at $P$ for some $a_i\in \F$. Therefore, the function $\sum_{i=1}^k\lambda_i g_i$ belongs to the space $\cL_\F((k+2\g-1)\Pin-kP)$, i.e., $\sum_{i=1}^k\lambda_i g_i$ is a linear combination of $\psi_1,\psi_2,\dots,\psi_g$. This forces that all $\lambda_i$ are equal to $0$  since $\{\psi_1,\dots,\psi_g,g_1,g_2,\dots,g_k\}$  is linearly independent. This completes the proof.
\end{proof}

With the above lemma in place, we now describe our AG code in a manner convenient for pruning the possible local expansion coefficients.

\smallskip
\noindent {\bf Encoding.} Assume that we have found a set of functions $\{g_1,g_2,\dots,g_k\}$ of $\cL_\F((k+2\g-1)\Pin)$ as in Lemma \ref{lem:choice-basis}. After elementary row operations on the matrix $A$ defined in Lemma \ref{lem:choice-basis}, we may assume that $A$ is the $k\times k$ identity matrix, i.e., we assume that, for $1\le i\le k$, the function $g_i$ has local expansion $T^{i-1}+\sum_{j=k}^{\infty}\lambda_{ij}T^j$ for some $\lambda_{ij}\in\F$.
Now we encode
each message $(a_1,a_2,\dots,a_k)\in \F^k$ to the codeword $(f(P_1),f(P_2),\dots,f(P_N))$, where $f=\sum_{i=1}^ka_i g_i$.

Now define the map $\phi_{P} : \F^k \to \cL_\F((k+2\g-1)\Pin)$ by sending $(a_1,a_2,\dots,a_k)\in \F^k$ to  $\sum_{i=1}^ka_i g_i$. We record the above fact for easy reference below.

\begin{claim}
	\label{clm:x5.3}
	The map $\phi_{P} : \F^k \to \cL_\F((k+2\g-1)\Pin)$ is $\F_q$-linear and injective. Furthermore, we can compute a representation of this linear transformation using ${\rm poly}(N,\g)$ operations over $\F_q$, and the map itself can be evaluated  using ${\rm poly}(N,\g)$ operations over $\F_q$ provided that local expansion of the basis elements of $\cL((k+2\g-1)\Pin)$ at $P$ can be computed using ${\rm poly}(N,\g)$ operations over $\F_q$.
\end{claim}

\section{Folded algebraic-geometric codes and their list decoding}
\label{sec:hermitian}
In this section, we will describe a variation of algebraic-geometric codes, namely, folded algebraic-geometric codes and their list decoding. For convenience, we will focus on one-point algebraic-geometric codes though this is not in any way a necessary restriction for our approach.

\subsection{Folded algebraic-geometric codes}
Let $F/\F_q$ be a function field.
To construct our folded codes, we assume that there exists  a global function field $F$ with the full constant field $\F_q$ having the following property:
\begin{itemize}
\item [(i)]  There exists  an automorphism $\Gs$ in $\Aut(F/\F_q)$ of order at least $m$;
\item [(ii)] $F$ has $mN$ distinct rational places $P_1, P_1^{\Gs},\dots,  P_1^{\Gs^{m-1}}, P_2, P_2^{\Gs},\dots,
   P_2^{\Gs^{m-1}}, \dots, $\\ $P_N, P_N^{\Gs},\dots,  P_N^{\Gs^{m-1}}$;
    \item [(iii)] $F$ has a rational place ${\Pin}$  such that ${\Pin}$ is fixed under $\Gs$, i.e., ${\Pin}^{\Gs}={\Pin}$; and $P_i^{\Gs^j}\neq {\Pin}$ for all $1\le i\le N$ and $0\le j\le m-1$.
\end{itemize}

A folded algebraic geometric code can be defined as follows.
\begin{defn}[Folded AG codes]
\label{def:f-ag-code}{\rm
The  folded  code from $F$ with parameters $N,l,q,m$, denoted by ${\FD}(N,l,q,m)$,  encodes a message function $f \in \cL(l\Pin)$  as
\begin{equation}
\label{eq:x7} \pi:\quad
f \mapsto
\left(
\left[\begin{array}{c}  f(P_1) \\  f(P_1^{\s}) \\ \vdots \\  f(P_1^{\s^{m-1}})\end{array}\right],
\left[\begin{array}{c} f(P_2) \\  f(P_2^{\s}) \\ \vdots \\  f(P_2^{\s^{m-1}})\end{array}\right],
\ldots,
\left[\begin{array}{c} f(P_N) \\  f(P_N^{\s}) \\ \vdots \\  f(P_N^{\s^{m-1}})\end{array}\right]
\right)  \in \left( \F_{q}^m \right)^{N} \ .
\end{equation}
We will abuse notation and for clarity refer to the encoding map $\pi$ also as  ${\FD}(N,l,q,m)$.
}\end{defn}
Note that the folded code ${\FD}(N,l,q,m)$ has the alphabet $\F_q^m$ and it is $\F_q$-linear. Furthermore, ${\FD}(N,l,q,m)$ has the following parameters.
\begin{lemma}\label{lem:x5.1} If $l<mN$, then
the above code ${\FD}(N,l,q,m)$ is an $\F_q$-linear code with alphabet size $q^{m}$, rate at least $\frac{l-\g+1}{Nm}$, and minimum distance at least $N  - \frac{l}{m}$, where $\g$ is the genus of $F$.
\end{lemma}
\begin{proof} It is clear that the map $\pi$ in  (\ref{eq:x7}) is $\F_q$-linear and  the kernel of $\pi$ is
\[ \cL\Bigl( l\Pin-\sum_{i=1}^N\sum_{j=0}^{m-1}P_i^{\s^j}\Bigr) \]
 which is $\{0\}$ under the condition that $l<mN$. Thus, $\pi$ is injective. Hence, the rate is at least $\frac{l-\g+1}{Nm}$ by  the Riemann-Roch theorem. To see the minimum distance, let $f$ be a nonzero function in $\cL(l\Pin)$ and assume that $I$ is the support of $\pi(f)$. Then the Hamming weight ${\rm wt}_H(\pi(f))$ of $\pi(f)$ is $|I|$ and  $f\in \cL\left(l\Pin-\sum_{i\not\in I}\sum_{j=0}^{m-1}P_i^{\s^j}\right)$. Thus, $0\le \deg\left(l\Pin-\sum_{i\not\in I}\sum_{j=0}^{m-1}P_i^{\s^j}\right)=l-m(N-|I|)$, i.e., ${\rm wt}_H(\pi(f))=|I|\ge N-\frac{l}{m}$. This completes the proof.
\end{proof}

\subsection{Encoding of code using local expansions}
\label{subsec:foldedag-enc-le}

For our decoding, we will actually recover the message $f \in \cL(l {\Pin})$ in terms of the coefficients of its power series expansion around a rational place $P$
\begin{equation}
  \label{eq:temp-le}
  f= x^{-\nu}(f_0 + f_1 x + f_2 x^2 + \cdots )
  \end{equation}
where $x$ is a local parameter at $P$. The place $P$ may or may not
equal $P_\infty$ -- when we instantiate the algorithm of this section
for the Hermitian tower, we will use a place different than $P_\infty$
for $P$, whereas for the Garcia-Stichtenoth tower, we will use $P =
P_\infty$. The reason for different choice of $P$ is that we need an explicit and simple local parameter at $P$ such that this local parameter still has an explicit and simple form after action of automorphism.  The description of the algorithm and its analysis in this
section will be general and cover both cases by letting
$\nu_P({\Pin})={\nu}$ with $\nu=0$ if ${\Pin}\neq P$ and $\nu=l$ if
$P={\Pin}$. Realizing that one must work in this power series
representation is one of the key insights in this work behind the
extension of the linear-algebraic folded Reed-Solomon list decoding
algorithm~\cite{GW-tit13} to the algebraic-geometric setting.

As already mentioned in Section~\ref{sec:local-exp-encoding}, one can injectively map the top $k$ coefficients of the above local expansion \eqref{eq:temp-le} into functions in $\cL(l\Pin)$ for $l=k+2\g-1$.
We will now redefine a version of the folded algebraic-geometric code that maps $\F_q^k$ to $(\F_q^m)^N$ by composing the folded encoding \eqref{eq:x7} from the original Definition \ref{def:f-ag-code} with the map $\phi_{P} : \F_q^k \to \cL((k+2\g-1)\Pin)$ promised in Claim~\ref{clm:x5.3}.

\begin{defn}[Folded algebraic-geometric code using local expansion]
\label{def:folded-agcode}
The folded algebraic-geometric code  $\widetilde{\FD}(N,k,q,m)$ maps
\[  \mv{f} = (f_0,f_1,\dots,f_{k-1})\in \F_q^k \quad \mapsto \quad  {\FD}(N,k+2\g-1,q,m)(\phi_{P}(\mv{f})) \in (\F_q^m)^N \ , \]
where ${\FD}(\dots)$ is the folded AG code from Definition~\ref{def:f-ag-code}.
\end{defn}

\smallskip \noindent
The rate of the above code equals $k/(Nm)$ and its distance is at least $N-(k+2\g-1)/m$.

\subsection{List decoding folded algebraic-geometric codes}\label{subsec:LDAG}

We now present a list decoding algorithm for the above codes. The algorithm follows the linear-algebraic list decoding algorithm for folded Reed-Solomon codes.

Suppose a codeword (\ref{eq:x7}) encoded from $f\in  \cL((k+2\g-1){\Pin})$ was transmitted and received as
\begin{equation}
\label{eq:x8}
\mathbf{y} =
\left(
\begin{array}{ccccc}
y_{1,1} & y_{2,1} & & & y_{N,1}\\
y_{1,2} & y_{2,2} & & & \vdots\\
& & & \ddots & \\
 y_{1,m} & \cdots &&& y_{N,m}
\end{array}
\right),
\end{equation}
where some columns are erroneous.
Let $s \ge 1$ be an integer parameter associated with the decoder.
\begin{lemma}
\label{lem:x5.4}
Given a received word as in {\rm (\ref{eq:x8})},  we can find a nonzero linear polynomial in $F[Y_1,Y_2,\dots,Y_s]$ of the form
\begin{equation}
\label{eq:x9}
Q(Y_1,Y_2,\dots,Y_s) =
A_0 + A_1 Y_1 + A_2 Y_2 + \cdots + A_s Y_s \
\end{equation}
satisfying
\begin{equation}
\label{eq:x10}
Q(y_{i,j+1},y_{i,j+2},\cdots,y_{i,j+s}) =A_0(P_i^{\s^{j}})+A_1(P_i^{\s^{j}})y_{i,j+1}+\cdots+A_s(P_i^{\s^{j}})y_{i,j+s}= 0
\end{equation}
{for } $i=1,2,\dots,N$ { and }  $j =0,1,\dots,m-s$. The coefficients $A_i$ of $Q$ satisfy
$A_i \in \cL(\kappa {\Pin})$ for $i=1,2,\dots,s$ and $A_0\in \cL((\Gk+(k+2\g-1)){\Pin})$ for a ``degree" parameter $\Gk$ chosen as
\begin{equation}
\label{eq:x11}
 \Gk=\left\lceil \frac {N(m-s+1)- (k+2\g-1)+(s+1)(\g-1)+1}{s+1}\right\rceil .
\end{equation}
\end{lemma}
\begin{proof} Let $u$ and $v$ be dimensions of $\cL(\Gk {\Pin})$ and $\cL((\Gk+(k+2\g-1)){\Pin})$, respectively. Let $\{x_1,\dots,x_u\}$ be an $\F_q$-basis of $\cL(\Gk {\Pin})$ and extend it to an $\F_q$-basis $\{x_1,\dots,x_v\}$ of  $\cL((d+k+2\g-1){\Pin})$. Then $A_i$ is an $\F_q$-linear combination of $\{x_1,\dots,x_u\}$ for $i=1,2,\dots,s$ and $A_0$ is an $\F_q$-linear combination of $\{x_1,\dots,x_v\}$. Determining the functions $A_i$ is equivalent to determining the coefficients in the combinations of $A_i$. Thus, there are in total $su+v$ degrees of freedoms to determine $A_0,A_1,\dots,A_s$. By the Riemann-Roch theorem, the number of degrees of freedoms is at least $s(\Gk -\g+(k+2\g-1))+(\Gk+k+2\g-1)-\g+1$.

On the other hand, there are in total $N(m-s+1)$ equations in (\ref{eq:x10}). Thus, there must be one nonzero solution by the condition (\ref{eq:x11}), i.e., $Q(Y_1,Y_2,\dots,Y_s)$ is a nonzero polynomial.
\end{proof}

\begin{lemma}
\label{lem:x5.5}
If $f$ is a function in $\cL(l\Pin)$ whose encoding (\ref{eq:x7}) agrees
with the received word $\mathbf{y}$ in at least $t$ columns with
\[ t>\frac{\Gk+l}{m-s+1} \ , \]
 then $Q(f,f^{\s^{-1}},\dots,f^{\s^{-(s-1)}})$ is the zero function, i.e., \begin{equation}\label{eq:x12}A_0+A_1f+A_2f^{\s^{-1}}+\cdots+A_sf^{\s^{-(s-1)}}=0.\end{equation}
\end{lemma}
\begin{proof}  Since ${\Pin}={\Pin}^{\s}$, we have $f^{\s^i}\in \cL(l\Pin)$ for all $i\in \ZZ$. Thus, it is clear that $Q(f,f^{\s^{-1}},\dots,f^{\s^{-(s-1)}})$ is a function in $\cL((\Gk+l){\Pin})$.

Let us assume that $I\subseteq\{1,2,\dots,N\}$ is the index set such that the $i$th columns  of $\by$ and $\pi(f)$ agree if and only if $i\in I$.  Then we have $|I|\ge t$. For every $i\in I$ and $0\le j\le m-s$, we have by (\ref{eq:x10})
\begin{eqnarray*}
0&=&A_0(P_i^{\s^{j}})+A_1(P_i^{\s^{j}})y_{i,j+1}+A_2(P_i^{\s^{j}})y_{i,j+2}+\cdots+A_s(P_i^{\s^{j}})y_{i,j+s}\\
&=&A_0(P_i^{\s^{j}})+A_1(P_i^{\s^{j}})f(P_i^{\s^{j}})+A_2(P_i^{\s^{j}})f(P_i^{\s^{j+1}}))+\cdots+A_s(P_i^{\s^{j}})f(P_i^{\s^{j+s-1}})\\
&=&A_0(P_i^{\s^{j}})+A_1(P_i^{\s^{j}})f(P_i^{\s^{j}})+A_2(P_i^{\s^{j}})f^{\s^{-1}}(P_i^{\s^{j}})+\cdots+A_s(P_i^{\s^{j}})f^{\s^{-s+1}}(P_i^{\s^{j}})\\
&=&\left(A_0+A_1f+A_2f^{\s^{-1}}+\cdots+A_sf^{\s^{-s+1}}\right)(P_i^{\s^{j}}),\end{eqnarray*}
 i.e., $P_i^{\s^{j}}$ is a zero of $Q(f,f^{\s},\dots,f^{\s^{s-1}})$. Hence, $Q(f,f^{\s^{-1}},\dots,f^{\s^{-(s-1)}})$ is a function in $\cL\left((\Gk+l){\Pin}-\sum_{i\in I}\sum_{j=0}^{m-s}P_i^{\s^{j}}\right)$. Our desired result follows from the fact that \\ $\deg\left((\Gk+l){\Pin}-\sum_{i\in I}\sum_{j=0}^{m-s}P_i^{\s^{j}}\right)<0$.
\end{proof}

By Lemma \ref{lem:x5.5}, we know that all candidate functions $f$ in our list must satisfy  equation (\ref{eq:x12}). In other words,  we have to study the solution set of  equation (\ref{eq:x12}). The  method used in \cite{GR-FRS} for decoding the Reed-Solomon codes is to construct an irreducible polynomial $h(x)$ of degree $q-1$ such that every polynomial $f$ satisfies $f^{\s^{-1}}\equiv f^{q} \mod{h} $. Then the solution set of (\ref{eq:x7}) is the same as the solution set of the equation $ A_0+A_1f+A_2f^{q}+\cdots+A_sf^{q^{s-1}} \equiv 0 \mod{h}$ since $\deg(f)<q-1=\deg(h)$. Thus, there are at most $q^{s-1}$ solutions for equation (\ref{eq:x7}).
This method does not work for folded algebraic-geometric codes. To upper bound list size of a folded algebraic-geometric code,
 we require an automorphisms of $\Aut(F/F_q)$ with order proportional to the genus $\g$ of $F$. However, it was proved in \cite{MX17} that the order of an automorphisms of $\Aut(F/\F_q)$ is upper bounded $O(\g/\log \g)$.

In this paper, we will analyze the solutions of the equation (\ref{eq:x7}) by considering local expansions at a certain point. This local expansion method  guarantees  a structured list of exponential size. Through precoding by using the structure in the list, we will be able to obtain an explicit construction of subcodes of these codes with polynomial time list decoding.

\smallskip\noindent {\bf {\em Solving the functional equation for $f$.}}
Recall that our goal is to recover the top $k$ coefficients $(f_0,f_1,\dots,f_{k-1})$ of the local expansion $f=x^{-\nu}\sum_{j=0}^{\infty}f_{j}x^{j}$
at $P$, based on the functional equation \eqref{eq:x12} that $f$ satisfies.

We now prove that $(f_0,f_1,\dots,f_{k-1})$ for $f$ satisfying Equation \eqref{eq:x12} belong to a periodic subspace (in the sense of Definition~\ref{def:periodic-subspaces}) of not too large dimension.
\begin{lemma}
\label{lem:x5.6} Let $P$ and $\Pin$ be two rational places of $F$ ($P$ and $\Pin$ can be the same) and let $f\in  \cL((k+2\g-1){\Pin})$. Assume that $\Gs\in\Aut(F/\F_q)$ is an automorphism satisfying ${\Pin}^\Gs={\Pin}$. Let $x\in F$ be a local parameter at $P$ satisfying  $x^\Gs=\frac{x}{\xi}$ for an element $\xi\in\F_q^*$ of order $p$. Put $\nu=k+2\g-1$ if  $P=\Pin$ and $0$ otherwise.

Then the set of solutions $(f_0,f_1,\dots,f_{k-1}) \in \F_q^k$ such that $f = x^{-\nu}(f_0+f_1 x+ f_2 x^2 + \cdots) \in \cL((k+2\g-1){\Pin})$ obeys the equation
\begin{equation}
\label{eq:x13}
A_0+A_1f+A_2f^{\s^{-1}}+\cdots+A_sf^{\s^{-(s-1)}}=0,
\end{equation}
when the $A_i$'s obey the pole order restrictions of Lemma \ref{lem:x5.4} and at least one $A_i$ is nonzero, is an $(s-1,p)$-ultra periodic subspace of $\F_q^k$.

Further, there are at most $q^{Nm+s+1}$ possible choices of this subspace over varying choices of the $A_i$'s.
\end{lemma}

\begin{proof}
  Let $u=\min\{\nu_{P}(A_i):\; i=1,2,\dots,s\}$.
  Then we have  $\nu_{P}(A_0)=\nu_{P}(-\sum_{i=1}^sA_if^{-\Gs^{i-1}})\ge\min\{\nu_{P}(A_if^{-\Gs^{i-1}})):\; i=1,2,\dots,s\}\ge \min\{\nu_{P}(A_i)-\nu:\; i=1,2,\dots,s\}=u-\nu$.
  Each $A_i$ has a local expansion at $P$:
\[A_i=x^u\sum_{j=0}^{\infty}a_{i,j}x^{j}\]
for $i=1,\dots,s$, and $A_0=x^{u-\nu}\sum_{j=0}^{\infty}a_{0,j}x^{j}$ which can be efficiently computed from the basis representation of the $A_i$'s. From the definition of $u$, one knows that the polynomial
\[B_0(X):=a_{1,0}+a_{2,0}X+\cdots+a_{s,0}X^{s-1}\]
is nonzero.
Assume that at $P$, the function $f$ has a local expansion $x^{-\nu}\sum_{j=0}^\infty f_j x^j$. Then $f^{\s^{-i}}$ has a local expansion at $P$ as follows
\[f^{\s^{-i}}=\xi^{-i\nu}x^{-\nu} \sum_{j=0}^{\infty}\xi^{ij}f_{j}x^{j}.\]
By direct inspection, we see that for every $d\ge 0$, the coefficient of $x^{d+u-\nu}$ in the local expansion of $A_0+A_1f+A_2f^{\s^{-1}}+\cdots+A_sf^{\s^{-(s-1)}}$ equals 
\begin{equation}
\label{eq:x14}
0=B_0(\xi^{d-\nu})f_d+\sum_{j=1}^{d} B_j(\xi^{d-j-\mu}) f_{d-j} + a_{0,d},
\end{equation}
where similarly to $B_0(X)$, the degree $(s-1)$ polynomials $B_j(X)$, $j \ge 1$, are defined as
\[ B_j(X) = a_{1,j} + a_{2,j} X + \cdots + a_{s,j} X^{s-1}. \]
Hence, $f_d$ is uniquely determined by $f_0,\dots,f_{d-1}$  as long as $B_0(\xi^{d-\nu})\neq 0$.

Let $S: =\{0\le i \le p-1:\; B_0(\xi^i)=0\}$. Then it is clear that $|S|\le s-1$
since the order of $\xi$ is $p$ so the powers $\xi^i$ are distinct for $0 \le i\le p-1$, and $B_0(X)$ has degree at most $s-1$. Thus,  $B_0(\xi^{d-\nu})\neq 0$ if and only if $d -\nu \mod p \notin S$; and in this case $f_d$ is a fixed affine linear combination of $f_j$ for $0 \le j < d$.

Let $W$ be the solution space $(z_0,z_1,\dots,z_{p-1}) \in \F_q^p$  of the equation system
\begin{equation}
  \label{eq:x23}
  B_0(\xi^{d-\mu})z_d+\sum_{j=1}^{d}B_j(\xi^{d-\mu-j}) z_j =0 \text{ for } d=0,1,\dots,p-1  \ .
\end{equation}
The above argument shows that $W$ is a subspace of $\F_q^p$ of dimension at most $s-1$.


We now claim that the solutions to \eqref{eq:x14} for $0 \le d < k$
form an $(s-1,p)$-periodic subspace of $\F_q^k$ with $W \subset
\F_q^p$ as the recurring subspace. This is immediate by inspecting the
system of equations \eqref{eq:x14} satisfied by the $f_i$'s and the
system \eqref{eq:x23} defining the subspace $W$. Indeed, once the values of
$f_i$, $0 \le i < p(j-1)$ are fixed, the possible choices for the
$j$'th block of $p$ coordinates, $f_{p(j-1)},\cdots, f_{pj-1}$, lie in
an affine shift of $W$.
Further, this shift is an explicit affine combination of the $f_i$'s for $0 \le i < p(j-1)$ (i.e., the previous $j-1$ blocks).

A closer inspection of \eqref{eq:x14} reveals that the subspace is in fact $(s-1,p)$-\emph{ultra periodic}, and are defined by a system of equations with the periodic structure of \eqref{eqn:ultra-period} of Definition~\ref{def:ultra-periodic}.

Finally, we record the bound on the number of different possible
solution spaces (this will be useful when we prune these via h.s.e
sets later).  By the choice of $\kappa$ in \eqref{eq:x12}, the total
number of possible $(A_0,A_1,\dots,A_s)$ and hence the number of
possible functional equations \eqref{eq:x12}, is at most $q^{N(m-s+1)
  + s + 1} \le q^{Nm+s+1}$. Therefore, the number of possible
candidate solution spaces is also at most $q^{Nm+s+1}$.
\end{proof}

Combining Lemmas \ref{lem:x5.4} and \ref{lem:x5.5} together with some
simple calculations leads to the following statement concerning list
decoding folded algebraic-geometric codes. We will later instantiate
this with Hermitian and Garcia-Stichtenoth towers, and also combine
with appropriate hierarchical subspace evasive sets to prune the
periodic subspace of solutions into a small list size.

\begin{theorem}
  \label{thm:folded-ag-general-ld}
  Consider the folded algebraic-geometric code from Definition~\ref{def:folded-agcode} based on a function field $F/\F_q$ and automorphism $\sigma$. Let $P$ (possibly equal to $\Pin$) be a rational place for which $x^\sigma = x/\xi$ for some local parameter $x \in F$ at $P$ and $\xi$ of order $p \ge m$ in $\F_q^*$. Assume that local expansions of functions in $\cL((k+2\g-1)\Pin)$ at $P$ can be computed in polynomial time.

  Then one can find a representation \eqref{eqn:proj-period} of an $(s,p)$-periodic subspace of $\F_q^k$ containing all candidate messages $(f_0,f_1,\dots,f_{k-1})$ in polynomial time, when the fraction of errors $\tau = 1-t/N$ satisfies
\begin{equation}
\label{eq:x15}
 \tau \le \frac{s}{s+1}- \frac{s}{s+1} \frac{k}{N(m-s+1)} - \frac{3m}{m-s+1} \frac{\g}{mN} \ .
 \end{equation}
\end{theorem}

\section{List decoding algebraic-geometric codes with subfield evaluation points}
\label{sec:lin-rs}
In this section, we will present a linear-algebraic list decoding algorithm for algebraic-geometric (AG) codes based on evaluations of functions at rational points over a \emph{subfield}.

The strategy in this section is similar to that of folded
algebraic-geometric codes. For a folded algebraic-geometric code, once
a coordinate is received correctly, then we have correct information
on $f(P_i),f(P_i^\Gs)=f^{\Gs^{-1}}(P_i),\dots,
f(P_i^{\Gs^{m-1}})=f^{\Gs^{-m+1}}(P_i)$. For algebraic-geometric codes
in this section, we has a similar property. Namely, once we receive a
coordinate correctly, then we have correct information on
$f(P_i),f^{\Gs}(P_i),\dots, f^{\Gs^{m-1}}(P_i)$, where $\Gs$ is the
Frobenius automorphism of an extension field.

For simplicity, to illustrate the ideas in a self-contained way in the
setting of univariate polynomials, we begin with the case of
Reed-Solomon codes in Section~\ref{new-subsec:decoding-rs} . We then
extend it to a general framework for decoding (one-point)
algebraic-geometric codes based on constant field extensions in
Section~\ref{new-subsec:decoding-ag}. Later on in the paper, we will
instantiate the general framework to codes based on the
Garcia-Stichtenoth tower discussed in \ref{subsec:gs-prelim}.

\subsection{Decoding Reed-Solomon codes}
\label{new-subsec:decoding-rs}
Our list decoding algorithm will apply to Reed-Solomon codes with evaluation points in a subfield, defined below.
\begin{defn}\label{def:RS}
[Reed-Solomon code with evaluations in a subfield]
Let $\F_q$ be a finite field with $q$ elements, and $m$ a positive integer. Let $n,k$ be positive integers satisfying $1 \le k < n \le q$. The Reed-Solomon code $\rs^{(q,m)}[n,k]$ is a code over alphabet $\F_{q^m}$ that encodes a polynomial $f \in \F_{q^m}[X]$ of degree at most $k-1$ as
\[ f(X) \mapsto (f(\Ga_1),f(\Ga_2),\cdots,f(\Ga_n)) \]
where $\Ga_1,\Ga_2,\dots,\Ga_n$ are an arbitrary sequence of $n$ distinct elements of $\F_q$.
\end{defn}

Note that while the message polynomial has coefficients from
$\F_{q^m}$, the encoding only contains its evaluations at points in
the subfield $\F_q$. The above code has rate $k/n$, and minimum
distance $(n-k+1)$.

We now present a list decoding algorithm for the above Reed-Solomon
codes.  Suppose the codeword $(f(\Ga_1),f(\Ga_2),\cdots,f(\Ga_n))$ is
received as $(y_1,y_2,\dots,y_n) \in \F_{q^m}^n$ with at most $e =
\tau n$ errors (i.e., $y_i \neq f(\Ga_i)$ for at most $e$ values of
$i\in \{1,2,\dots,n\}$).  The goal is to recover the list of all
polynomials of degree less than $k$ whose encoding is within Hamming
distance $e$ from $y$. As is common in algebraic list decoders, the
algorithm will have two steps: (i) interpolation to find an algebraic
equation the message polynomials must satisfy, and (ii) solving the equation for the candidate message polynomials.

\medskip \noindent {\bf Interpolation step.} Let $1 \le s \le m$ be an integer parameter of the algorithm. Choose the ``degree parameter" $D$ to be
\begin{equation}
\label{eq:choice-of-D-rs}
 D=\left\lfloor \frac {n-k+1}{s+1}\right\rfloor .
\end{equation}

\begin{defn}[Space of interpolation polynomials]
Let $\mathcal{P}$ be the space of polynomials $Q \in \F_{q^m}[X,Y_1,Y_2,\dots,Y_s]$ of the form
\begin{equation}
\label{eq:form-of-Q-rs}
Q(X, Y_1,Y_2,\dots,Y_s) =
A_0(X) + A_1(X) Y_1 + A_2(X) Y_2 + \cdots + A_s( X) Y_s \ ,
\end{equation}
with each $A_i \in \F_{q^m}[X]$ and $\deg(A_0)\le D+k-1$ and $\deg(A_i)\le D$ for $i=1,2,\dots,s$.
\end{defn}

The lemma below follows because for our choice of $D$, the number of degrees of freedom for polynomials in $\mathcal{P}$ exceeds the number $n$ of interpolation conditions \eqref{eq:cond-on-Q-rs}. We include the easy proof for completeness.
\begin{lemma}
\label{lem:Q-rs}
There exists a nonzero polynomial $Q \in \mathcal{P}$ such that
\begin{equation}
\label{eq:cond-on-Q-rs}
Q(\Ga_i,y_i,y_i^q,y_i^{q^2},\cdots,y_i^{q^{s-1}}) = 0 \quad \text{for} \quad i=1,2,\dots,n \ .
\end{equation}
Further such a $Q$ can be found using $O(n^3)$ operations over $\F_{q^m}$.
\end{lemma}
\begin{proof}
Note that $\mathcal{P}$ is an $\F_{q^m}$-vector space of dimension
\[ (D+k) + s (D+1) = (D+1) (s+1) + k - 1 > n, \]
where the last inequality follows from our choice \eqref{eq:choice-of-D-rs}.
The interpolation conditions required in the lemma impose $n$ homogeneous linear conditions on $Q$. Since this is smaller than the dimension of $\mathcal{P}$, there must exist a nonzero $Q \in \mathcal{P}$ that meets the interpolation conditions
\[
Q(\Ga_i,y_i,y_i^q,y_i^{q^2},\cdots,y_i^{q^{s-1}}) = 0 \quad \text{for} \quad i=1,2,\dots,n \ .
\]
Finding such a $Q$ amounts to solving a homogeneous linear system over $\F_{q^m}$ with $n$ constraints and at most $\mathsf{dim}(\mathcal{P}) \le n+s+2$ unknowns, which can be done in $O(n^3)$ time.
\end{proof}

Lemma \ref{lem:Q-soln-rs} below shows that any polynomial $Q$ given by Lemma \ref{lem:Q-rs} yields an algebraic condition that the message functions $f$ we are interested in list decoding must satisfy.

\begin{defn}[Frobenius action on polynomials]
For a polynomial $f \in \F_{q^m}[X]$ with $f(X) = f_0 + f_1 X + \cdots + f_{k-1} X^{k-1}$, define the polynomial $f^\sigma \in \F_{q^m}[X]$ as $f^\sigma(X) = f_0^q + f_1^q X + \cdots + f_{k-1}^q X^{k-1}$.

For $i \ge 2$, we define $f^{\sigma^i}$ recursively as $(f^{\sigma^{i-1}})^\sigma$.
\end{defn}

The following simple fact is key to our analysis.
\begin{fact}
\label{fact:frob}
If $\Ga \in \F_q$, then $f(\Ga)^{q^j} = (f^{\sigma^j})(\Ga)$ for all $j=1,2,\dots$.
\end{fact}

\begin{lemma}
\label{lem:Q-soln-rs}
Suppose $Q \in \mathcal{P}$ satisfies the interpolation conditions \eqref{eq:cond-on-Q-rs}. Suppose $f \in \F_{q^m}[X]$ of degree less than $k$ satisfies $f(\Ga_i) \neq y_i$ for at most $e$ values of $i \in \{1,2,\dots,n\}$ with $e \le \frac{s}{s+1} ( n - k)$. Then
$Q(X, f(X), f^\sigma(X), f^{\sigma^2}(X), \cdots, f^{\sigma^{s-1}}(X)) = 0$.
\end{lemma}
\begin{proof}
Define the polynomial $\Phi \in \F_{q^m}[X]$ by $\Phi(X) := Q(X, f(X), f^\sigma(X), f^{\sigma^2}(X), \cdots,$ $ f^{\sigma^{s-1}}(X))$. By the construction of $Q$ and the fact that $\deg(f) \le k-1$, we have $\deg(\Phi) \le D+k-1 \le \frac{n-k+1}{s+1} + k - 1 = \frac{n}{s+1} + \frac{s}{s+1} (k-1)$.

Suppose $y_i = f(\Ga_i)$. By Fact \ref{fact:frob}, we have $y_i^q = f(\Ga_i)^q =  (f^\sigma)(\Ga_i)$, and similarly $y_i^{q^j} = (f^{\sigma^j})(\Ga_i)$ for $j=2,3,\dots$. Thus for each $i$ such that $f(\Ga_i)=y_i$, we have
$$ \Phi(\Ga_i) = Q(\Ga_i, f(\Ga_i), f^\sigma(\Ga_i), \cdots, f^{\sigma^{s-1}}(\Ga_i)) = Q(\Ga_i, y_i, y_i^q, \cdots, y_i^{q^{s-1}}) = 0 \ .$$
Thus $\Phi$ has at least $n-e \ge \frac{n}{s+1} + \frac{s}{s+1} k$ zeroes. Since this exceeds the upper bound on the degree of $\Phi$, $\Phi$ must be the zero polynomial.
\end{proof}

\smallskip \noindent {\bf Finding candidate solutions.} The previous two lemmas imply that the polynomials $f$ whose encodings
differ from $(y_1,\cdots,y_n)$ in at most $\frac{s}{s+1}(n-k)$
positions can be found amongst the solutions of the functional
equation $A_0 + A_1 f + A_2 f^\sigma + \cdots + A_s f^{\sigma^{s-1}} =
0$.  We now prove that these solutions form a well-structured affine
space over $\F_q$.

\begin{lemma}
\label{lem:space-of-solns-rs}
For integers $1 \le s \le m$, the set of solutions $f=\sum_{i=0}^{k-1}f_iX^{i}\in\F_{q^m}[X]$ to the equation
\begin{equation}
\label{eq:affine-cond-on-f-rs}
A_0(X)+A_1(X) f(X) +A_2(X) f^{\s}(X)+\cdots+A_s(X) f^{\s^{s-1}}(X)=0
\end{equation}
when at least one of $\{A_0,A_1,\dots,A_s\}$ is nonzero is an affine subspace over $\F_q$ of dimension at most $(s-1)k$. Further, fixing an $\F_q$-basis of $\F_{q^m}$ and viewing each $f_i$ as an element of $\F_q^m$, the solutions are an $(s-1,m,k)$-periodic subspace of $\F_q^{mk}$. A representation of this periodic subspace (in the form \eqref{eqn:proj-period} from Definition \ref{def:periodic-rep}) can be computed in $\mathrm{poly}(k,m,\log q)$ time.
\end{lemma}
\begin{proof}
If $f,g$ are two solutions to \eqref{eq:affine-cond-on-f-rs}, then so is $\alpha f + \beta g$ for any $\alpha,\beta \in \F_q$ with $\alpha + \beta = 1$. So the solutions to \eqref{eq:affine-cond-on-f-rs} form an affine $\F_q$-subspace. We now proceed to analyze the structure of the subspace.

First, by factoring out a common powers of $X$ that divide all of $A_0(X),A_1(X),\dots,A_s(X)$, we can assume that at least one $A_{i^*}(X)$ for some $i^* \in\{0,1,\dots,s\}$ is not divisible by $X$, and has nonzero constant term. Further, if $A_1(X),\dots,A_s(X)$ are all divisible by $X$, then so is $A_0(X)$, so we can take $i^* > 0$.

Let us denote $A_i(X) = a_{i,0} + a_{i,1} X + a_{i,2} X^2 + \cdots$ for $i = 0,1,2,\dots,s$.
%
%
%
For $l=0,1,2,\dots,D$, define the linearized polynomial
\begin{equation}
\label{eq:def-of-Bj}
B_l(X) = a_{1,l}X + a_{2,l} X^q + a_{3,l} X^{q^2} + \cdots + a_{s,l} X^{q^{s-1}} \ .
\end{equation}
We know that $a_{i^*,0} \neq 0$, and therefore $B_0 \neq 0$. This implies that the solutions $\beta \in \F_{q^m}$ to $B_0(\beta) = 0$ is a $\F_q$-subspace, say $W$, of $\F_{q^m}$ of dimension at most $s-1$.

Fix an $i \in \{0,1,\dots, k-1\}$. Expanding the equation \eqref{eq:affine-cond-on-f-rs} and equating the coefficient of $X^i$ to be $0$, we get
\begin{equation}
\label{eq:f_i}
 a_{0,i} + B_i(f_0) + B_{i-1}(f_1) + \cdots + B_1(f_{i-1}) + B_0(f_i) = 0 \ .
\end{equation}
Therefore, for each $i=0,1,\dots,k-1$, $f_i$ must belong to a coset of the subspace $W + \theta_i$ where $\theta_i$ is an affine combination of $f_0,f_1,\dots,f_{i-1}$. It follows that the solutions $(f_0,f_1,\dots,f_{k-1})$ to \ref{eq:affine-cond-on-f-rs} viewed as a vector in $\F_q^{mk}$ (w.r.t any fixed $\F_q$-basis of $\F_{q^m}$) belongs to an
form an $(s-1,m,k)$-periodic subspace. The equations \eqref{eq:f-expansion} give the desired representation  of this periodic subspace.
\end{proof}

Combining Lemmas \ref{lem:Q-soln-rs} and \ref{lem:space-of-solns-rs}, we see that one can find an affine space of dimension $(s-1)k$ that contains the coefficients of all polynomials whose encodings differ from the input $(y_1,\dots,y_n)$ in at most a fraction $\frac{s}{s+1}(1-R)$ of the positions. Note the dimension of the message space of the Reed-Solomon code
$\rs^{(q,m)}[n,k]$ over $\F_q$ is $km$. The above lemma pins down the
candidate polynomials to a space of dimension $(s-1) k$. For $s \ll
m$, this is a lot smaller. In particular, it implies one can list decode in time sub-linear in the code size (the proof follows by taking $s = \lceil 1/\eps\rceil$ and $m > \frac{s}{\gamma}$).
\begin{cor}
For every $R \in (0,1)$, and $\eps, \gamma > 0$, there is a positive integer $m$ such that for all large enough prime powers $q$, the Reed-Solomon code $C=\rs^{(q,m)}[q,Rq]$ can be list decoded from a fraction $(1-R-\eps)$ of errors in $|C|^\gamma$ time, outputting a list of size at most $|C|^\gamma$.
\end{cor}

Since the dimension of the subspace guaranteed by Lemma
\ref{lem:space-of-solns-rs} grows linearly in $k$, we cannot afford to
list this subspace as the decoder's output for polynomial time
decoding. However, using the periodic structure of the subspace, one
can prune it by using a ``pre-code'' that only allows polynomials with
coefficients in subspace designs or h.s.e sets as we will see in later
sections.

\subsection{Decoding algebraic-geometric codes}
\label{new-subsec:decoding-ag}

We now generalize the Reed-Solomon algorithm from the previous subsection to
algebraic-geometric codes. The description in this section will be for
a general abstract AG code. So we will focus on the algebraic ideas,
and not mention complexity estimates. Later, we will focus
on a specific AG code based on Garcia-Stichtenoth function fields,
which will require a small change to the setup, and where we will also
mention computational aspects. We refer to Subsection \ref{subsec:5.2local} for encoding and will focus on a decoding algorithm.

\subsubsection{AG codes with evaluation points in a subfield}
Let $F/\F_q$ be a function field of genus $\g$. Let $\Pin, P_1,P_2,\dots,P_N$ be $N+1$ distinct $\F_q$-rational places. Let $\Gs\in{\rm Gal}(\Fm/\F_q)$ be the Frobenius automorphism, i.e, $\Ga^{\Gs}=\Ga^q$ for all $\Ga\in\Fm$. Then we can extend $\Gs$ to an automorphism in ${\rm Gal}(F_m/F)$, where $F_m$ is the constant extension $\Fm\cdot F$. Note that $P^{\Gs}=P$ for any place of $F$.

Consider the Goppa geometric code defined by
\begin{equation}
\label{eq:agcode-subfield}
C(m;l):=\{(f(P_1),f(P_2),\dots,f(P_N)):\;f\in\cL_m(l\Pin)\} \ .
\end{equation}

We have the following well-known result on the parameters of the above
algebraic-geometric codes.
\begin{lemma}
\label{lem:para}
The above code $C(m;l)$ is an $\Fm$-linear code over $\Fm$, rate at least $\frac{l-\g+1}{N}$, and minimum distance at least $N  -l$.
\end{lemma}
\subsubsection{Encoding of code using local expansions}
As with the case of folded AG codes, the decoding algorithm will
recover the message function $f$ via the coefficients of its local expansion at some place $P$. Therefore, we will identify the message symbols with local expansion coefficiently of the function $f$ and encode into a subcode of $C(m;l)$.

\begin{defn}[Subfield algebraic-geometric code using local expansion]
	\label{def:subfield-agsubcode}
	The folded algebraic-geometric code  $\widetilde{C}(m;k)$ maps
	\[  \mv{f} = (f_0,f_1,\dots,f_{k-1})\in \F_{q^m}^k \quad \mapsto \quad  
	(\phi_{P}(\mv{f})(P_1), \phi_{P}(\mv{f})(P_2), \dots, \phi_{P}(\mv{f})(P_N))\in \F_{q^m}^N \ , \]
	where $\phi_P(\cdot)$ is the map converting a local expansion into an associated function guaranteed by Claim~\ref{clm:x5.3}.
\end{defn}

\subsubsection{A list decoding algorithm}
We now present a list decoding algorithm for the above codes. The algorithm follows the linear-algebraic list decoding algorithm for RS codes. It is quite similar to that of folded algebraic-geometric codes. Suppose a codeword encoding $f \in \cL_m((k+2\g-1)\Pin)$ is transmitted and received as $\by=(y_1,y_2,\dots,y_N)$.

Given such a received word, we will interpolate a nonzero linear polynomial over $F_m$
\begin{equation}
\label{eq:form-of-Q}
Q(Y_1,Y_2,\dots,Y_s) =
A_0 + A_1 Y_1 + A_2 Y_2 + \cdots + A_s Y_s \
\end{equation}
where $A_i \in \cL_m(D\Pin)$ for $i=1,2,\dots,s$ and $A_0\in \cL_m((D+k+2\g-1)\Pin)$ with the degree parameter $D$ chosen to be
\begin{equation}
\label{eq:choice-of-D}
 D=\left\lfloor \frac {N-k+(s-1)\g+1}{s+1}\right\rfloor .
\end{equation}
If we fix a basis of $\cL_m(D\Pin)$ and extend it to a basis of $\cL_m((D+k+2\g-1)\Pin)$, then
the number of freedoms of $A_0$ is at least $D+k+\g$ and the number of freedoms of $A_i$  is at least $D-\g+1$ for $i\ge 1$. Thus, the total number of freedoms in the polynomial $Q$ equals
\begin{equation}
\label{eq:affine-cond-on-deg}
s(D-\g+1)+D+k+\g=(s+1)(D+1)-(s-1)\g-1+k>N.
\end{equation}
for the above choice (\ref{eq:choice-of-D}) of $D$.
The interpolation requirements on $Q \in F_m[Y_1,\dots,Y_s]$ are
the following:
\begin{equation}
\label{eq:interpolation-cond}
Q(y_i,y_i^{\Gs},\dots,y_i^{\Gs^{s-1}}) =A_0(P_i)+A_1(P_i)y_i+A_2(P_i)y_i^{\Gs}+\cdots+A_s(P_i)y_i^{\Gs^{s-1}}= 0
\end{equation}
{for } $i=1,2,\dots,N$. Thus, we have a total of $N$ equations to satisfy.
Since this number is less than the number of  freedoms in $Q$, we can conclude that a nonzero linear function $Q \in F_m[Y_1,\dots,Y_s]$ of the form
  (\ref{eq:form-of-Q}) satisfying the interpolation conditions
  (\ref{eq:interpolation-cond}) can be found by solving a homogeneous
  linear system over $\Fm$ with at most $N$ constraints and at least $s(D-\g+1)+D+k+\g$
  variables.

The following lemma gives the algebraic condition that the message functions $f \in \cL_m((k+2\g-1)\Pin)$ we are interested in list decoding must satisfy.

\begin{lemma}
\label{lem:Q-is-good}
If $f$ is a function in $\cL_m((k+2\g-1)\Pin)$ whose encoding  agrees
with the received word $\mathbf{y}$ in at least $t$ positions with
$t>{D+k+2\g-1}$,
 then \begin{equation}\label{eq:alg-eqn}Q(f,f^{\s},\dots,f^{\s^{s-1}})=A_0+A_1f+A_2f^{\s}+\cdots+A_sf^{\s^{s-1}}=0.\end{equation}
\end{lemma}
\begin{proof}
The proof proceeds by comparing the number of zeros of the function $Q(f,f^{\s},$ $\dots,f^{\s^{s-1}})=A_0+A_1f+A_2f^{\s}+\cdots+A_sf^{\s^{s-1}}$ with $D+k+2\g-1$. Note that $Q(f,f^{\s},\dots,$ $f^{\s^{s-1}})$ is a function in $\cL_m((D+k+2\g-1)\Pin)$.  If position $i$ of the encoding  of $f$ agrees with $\mathbf{y}$, then
 \begin{eqnarray*}
0&=&A_0(P_i)+A_1(P_i)y_i+A_2(P_i)y_i^{\Gs}+\cdots+A_s(P_i)y_i^{\Gs^{s-1}}\\
&=&A_0(P_i)+A_1(P_i)f(P_i)+A_2(P_i)(f(P_i))^{\s}+\cdots+A_s(P_i)(f(P_i))^{\s^{s-1}}\\
&=&A_0(P_i)+A_1(P_i)f(P_i)+A_2(P_i)f^{\s}(P_i)+\cdots+A_s(P_i)f^{\s^{s-1}}(P_i)\\
&=&(A_0+A_1f+A_2f^{\s}+\cdots+A_sf^{\s^{s-1}})(P_i)\end{eqnarray*}
 i.e., $P_i$ is a zero of $Q(f,f^{\s},\dots,f^{\s^{s-1}})$. Thus, there are  at least $t$ zeros for all the agreeing positions.  Hence, $Q(f,f^{\s},\dots,f^{\s^{s-1}})$ must be the zero function when $t >  D+k+2\g-1$.
\end{proof}

\begin{lemma}
\label{lem:subs-struc}
Let $P$ be a rational place of $F$ with a local parameter $x\in F$ ($P$ may or may not be the same as $\Pin$).
The set of solutions $f\in\cL_m((k+2\g-1)\Pin)$ to the equation
\[A_0+A_1f+A_2f^{\s}+\cdots+A_sf^{\s^{s-1}}=0\]
when at least one $A_i$ is nonzero has size at most $q^{(s-1)(k+2\g-1)}$.  Further, the possible first $k$ coefficients $(f_0,f_1,\dots,$ $f_{k-1})$ of $f$'s local expansion at $P$ belong to an $(s-1,m)$-ultra periodic affine subspace of $\F_q^{mk}$.
\end{lemma}
\begin{proof}
The argument is very similar to Lemma \ref{lem:x5.6}. Define $\nu=k+2\g-1$ if $P=\Pin$ and $0$ otherwise.
Let $u=\min\{\nu_{P}(A_i):\; i=1,2,\dots,s\}$. Then we have  $\nu_{P}(A_0)=\nu_{P}(-\sum_{i=1}^sA_if^{-\Gs^{i-1}})\ge\min\{\nu_{P}(A_if^{-\Gs^{i-1}})):\; i=1,2,\dots,s\}\ge \min\{\nu_{P}(A_i)-\nu:\; i=1,2,\dots,s\}=u-\nu$.
 Each $A_i$ has a local expansion at $P$:
\[A_i=x^u\sum_{j=0}^{\infty}a_{i,j}x^{j}\]
for $i=1,\dots,s$ and $A_0$ has a local expansion $A_0=x^{u-\nu}\sum_{j=0}^{\infty}a_{0,j}x^{j}$.

Assume that at $P$, the function $f$ has a local expansion (\ref{eq:x1}).
Then $f^{\s^{i}}$ has a local expansion at $P$ as follows
\[f^{\s^{i}}=x^{-\nu}\sum_{j=0}^{\infty}f_{j}^{q^i}x^{j}.\]

For $l=0,1,\dots$, define the linearized polynomial
\[ B_l(X):=a_{1,l}X+a_{2,l}X^q+\cdots+a_{s,l}X^{q^{s-1}}\]
From the definition of $u$, one knows that $B_0(X)$
is nonzero. For $d\ge 0$, equating the coefficient of $x^{d+u-\nu}$ in $A_0+A_1f+A_2f^{\s}+\cdots+A_sf^{\s^{s-1}}$ to equal $0$ gives us the condition
\begin{equation}
\label{eq:ultra-per}
a_{0,d} + B_d(f_0) + B_{d-1}(f_1) + \cdots + B_0(f_d) = 0 \ .
\end{equation}
Let $W = \{\Ga\in\Fm:\; B_0(\Ga)=0\}$.
 Then $W$ is an $\F_q$-subspace  of $\Fm$ of dimension at most $s-1$, since $B_0$ is a nonzero linearized polynomial of $q$-degree at most $s-1$. As in Lemma \ref{lem:space-of-solns-rs}, for each fixed $f_0,f_1,\dots,f_{d-1}$, the coefficient $f_d$ must belong to a coset of the subspace $W$. This implies that the coefficients $(f_0,f_1,\dots,f_{k+2\g-1})$ belong to an $(s-1,m,k+2\g-1)$-periodic subspace of $\F_q^{m(k+2\g-1)}$. In particular, there are at most $q^{(s-1)(k+2\g-1)}$ solutions $f \in \cL_m((k+2\g-1)\Pin)$ to \eqref{eq:interpolation-cond}.

The equation \eqref{eq:ultra-per} also shows that each group of $\iota$ successive coefficients $f_{d-\iota+1},$ $f_{d-\iota+2},\cdots,f_d$ belong to cosets of the same underlying $\iota (s-1)$ dimensional subspace of $\F_q^{m\iota}$. This implies that $(f_0,f_1,\dots,f_{k})$ in fact belong to an $(s-1,m)$-ultra periodic subspace.\footnote{This ultra-periodicity was also true for the Reed-Solomon case in Lemma \ref{lem:space-of-solns-rs}, but we did not state it there as we will not make use of this extra property for picking a subcode in the case of Reed-Solomon codes.}
\end{proof}

\noindent {\bf Decoding.} Recall that the first $k$ coefficients of the local expansion of $f \in \cL_m(l P_\infty)$ around $P$ is precisely the message that was encoded in the code $\widetilde{C}(m;k)$ of Definition~\ref{def:subfield-agsubcode}.

Therefore, combining Lemmas \ref{lem:Q-is-good} and \ref{lem:subs-struc}, and recalling the choice of $D$ in \eqref{eq:choice-of-D}, we can conclude the following result about list-decodability of our code construction.
\begin{theorem}
\label{thm:ag-final}
For the code $\widetilde{C}(m;k)$, we can find an $(s-1,m)$-ultra periodic subspace of $\F_q^{mk}$ that includes all messages whose encoding differs from a received word $\mv{y} \in \F_{q^m}^N$ in at most
\[ \frac{s}{s+1} (N - k) - \frac{3s+1}{s+1} \g \]
positions.
\end{theorem}

\section{Instantiating with Hermitian and Garcia-Sticthenoth towers}\label{sec:FGS}
In Sections \ref{sec:hermitian} and \ref{sec:lin-rs}, we discussed list decoding of  folded algebraic-geometric codes and algebraic-geometric  codes with subfield evaluation points. In this section, we instantiate the codes and list decoding algorithms described in Sections \ref{sec:hermitian} and \ref{sec:lin-rs} with two important and explicit towers, i.e.,  the Hermitian and Garcia-Sticthenoth towers.

\subsection{Folded Hermitian codes}
\label{newsubsec:FH}
In this subsection, let us instantiate the list decoding algorithm of folded algebraic geometric codes from general algebraic function fields with the Hermtian tower. We refer to Subsection \ref{subsec:herm-prelim} for detailed background on the Hermitian tower. Let  $r$ be a prime power and let $q=r^2$. We denote by $\F_q$ the finite field with $q$ elements. Let $F_e=\F_q(x_1,x_2,\dots,x_e)$ be  the Hermitian tower defined by \eqref{eq:x2}. Let $\Gg$ be a primitive element of $\F_q$. Consider the automorphism $\Gs\in{\rm Aut}(F_e/\F_q)$ defined by
\[\Gs:\; x_i\mapsto\Gg^{(r+1)^{i-1}}x_i \quad \mbox{for}\ i=1,2,\dots,e.\]
For an integer $m$ with $1\le m\le q-1$, let $\Pin$ and $P_i^{\Gs_j}$ for $i=1,2,\dots, N$ and $j=0,1,\dots m-1$ be the same as defined in Subsection \ref{subsec:herm-prelim}. \begin{defn}[Folded codes from the Hermitian tower]
\label{def:fh-F}
Assume that $m,l,N$ are positive integers satisfying $1\le m\le q-1$ and $l/m \le N \le r^{e-1}\left\lfloor\frac{q-1}m\right\rfloor$.
The  folded  code from $F_e$ with parameters $N,l,q,e,m$, denoted by ${\FH}_e(N,l,q,m)$,  encodes a message function $f \in \cL(l\Pin)$ a folded codeword given in \eqref{eq:x7}.
\end{defn}
When $e=1$, the  folded  code ${\FH}_1(N,l,q,m)$ is in fact a folded Reed-Solomon code introduced in \cite{GR-FRS}.
\begin{lemma}
The above code ${\FH}_e(N,l,q,m)$ is an $\F_q$-linear code over alphabet size $q^{m}$, rate at least $\frac{l-g_e+1}{Nm}$, and minimum distance at least $N  - \frac{l}{m}$.
\end{lemma}
\begin{proof}
It is clear that the map \eqref{eq:x7} is an $\F_q$-linear map. The dimension over $\F_q$ of the message space $\cL(l\Pin)$ is at least $l-g_e+1$ by the Riemann-Roch theorem, which gives the claimed lower bound on rate.
For the distance property, observe that if the $i$-th column is zero, then $f$ has $m$ zeros.  This implies that the encoding of a nonzero function $f$ can have at most $l/m$ zero columns since $f \in \cL(l\Pin)$.
\end{proof}

Let $P_0$ be the common zero of $x_1,x_2,\dots,x_e$.
For our decoding, we will actually recover the message $f \in \cL(l \Pin)$ in terms of the coefficients of its power series expansion around $P_0$
\begin{equation}\label{eq:x20} f= f_0 + f_1 x + f_2 x^2 + \cdots  \end{equation}
where $x := x_1$ is the local parameter at $P_0$ (which means that $x_1$ has exactly one zero at $P_0$, i.e., $\nu_{P_0}(x_1)=1$). 

With this in mind, we now define the encoding into the above folded Hermitian code using the map $\phi_{P_0}$ from Claim~\ref{clm:x5.3}. 

\begin{defn}[Folded Hermitian code using local expansion]
\label{def:folded-herm-F}
The folded Hermitian code  $\widetilde{\FH}_e(N,k,q,m)$ maps 
\[ \mv{f} = (f_0,f_1,\dots,f_{k-1})\in \F_q^k \quad \text{to} \quad  \FH_e(N,k+2g_e-1,q,m)(\phi_{P_0}(\mv{f})) \in (\F_q^m)^N \ . \]
\end{defn}

Given the local expansions of a basis of $\cL(l P_\infty)$ at $P_0$, computing the map $\phi_{P_0}$ to convert from local expansion to some representative function in $\cL(l P_\infty)$ can be done in polynomial time by simply solving a system of linear equations. We now turn to the task of computing the local expansion at $P_0$ of a basis for $\cl(l\Pin)$.

\begin{lemma}
\label{lem:ps-basis-F}
For any $n$, one can compute the first $n$ terms of the local expansion of the basis elements \eqref{eq:x20} at $P_0$ using ${\rm poly}(n)$ operations over $\F_q$.
\end{lemma}
\begin{proof}
By the structure of the basis functions in \eqref{eq:x3}, it is sufficient to find an algorithm of efficiently finding local expansions of $x_i$ at $P_0$ for every $i=1,2,\dots,e$. We can inductively find the local expansions of $x_i$ at $P_0$ as follows.

For $i=1$, $x_1$ is the local parameter $x$ of $P_0$, so $x$ is the local expansion of $x_1$ at $P_0$.

Now assume that we know the local expansion of $x_i=\sum_{j=1}^{\infty}c_{i,j}x^j$ at $P_0$ for some $c_{i,j}\in\F_q$. Then
 we have
 \[\sum_{j=1}^{\infty}c_{i+1,j}^rx^{jr}+\sum_{j=1}^{\infty}c_{i+1,j}x^j=x_{i+1}^r+x_{i+1}=x_i^{r+1}=\left(\sum_{j=1}^{\infty}c_{i,j}^rx^{jr}\right)\left(\sum_{j=1}^{\infty}c_{i,j}x^j\right).\]
 Note that $r$ is a power of the characteristic and hence $r$ can be pushed into infinite sums. By comparing the coefficients of $x^j$ in the above identity, we can easily solve $c_{i+1,j}$'s from $c_{i, j}$'s. More specifically, the coefficient of $x^j$ at the left of the identity is
 \[\left\{\begin{array}{ll}c_{i+1,j}&\mbox{if $r\not| j$}\\
c_{i+1,j}+c_{i+1,j/r}^r&\mbox{if $r| j$.}
 \end{array}
 \right.
\]
Thus, all $c_{i+1,j}$'s can be easily solved recursively.
\end{proof}

By instantiating Theorem \ref{thm:folded-ag-general-ld} with our code $\widetilde{\FH}_e(N,k,q,m)$, we obtain the following result.
\begin{theorem}
  \label{thm:folded-Hermitian-ld}
 One can find a representation of an $(s,q-1)$-periodic subspace\footnote{In fact, this subspace will be $(s,q-1)$-ultra periodic.}  of $\F_q^k$ containing all candidate messages $\mv{f} = (f_0,f_1,\dots,f_{k-1})$ in polynomial time, when the fraction of errors $\tau = 1-t/N$ in its encoding by $\widetilde{\FH}_e(N,k,q,m)$ satisfies
\begin{equation}
\label{eq:herm-error-frac-F}
 \tau \le \frac{s}{s+1}- \frac{s}{s+1} \frac{k}{N(m-s+1)} - \frac{3m}{m-s+1} \frac{\g_e}{mN} \ .
 \end{equation}
\end{theorem}

\subsection{Folded codes from the Garcia-Stichtenoth tower}
\label{subsec:FGS}
Compared with the Hermitian tower of function fields, the
Garcia-Stichtenoth tower of function fields yields folded codes with
better parameters due to the fact that the Garcia-Stichtenoth
tower is an optimal one in the sense that the ratio of number of
rational places against genus achieves the maximal possible value. The
construction of folded codes from the Garcia-Stichtenoth tower is
almost identical to the one from the Hermitian tower except for one
major difference: the redefined code from the Garcia-Stichtenoth tower
is constructed in terms of the local expansion at point $\Pin$, while
in the Hermitian case local expansion at $P_0$ is considered. For
convenience of the reader, we give a parallel description of folded
codes from the Garcia-Stichtenoth tower, while only sketching the
identical parts.  We refer to Subsection \ref{subsec:gs-prelim} for background on the Garcia-Stichtenoth tower.

Let  $r$ be a prime power and let $q=r^2$. We denote by $\F_q$ the finite field with $q$ elements. Let $K_e=\F_q(x_1,x_2,\dots,x_e)$ be  the Garcia-Stichtenoth tower defined by \eqref{eq:x5}.  Let $\Gg$ be a primitive element of $\F_q$ and consider the automorphism $\Gs\in{\rm Aut}(K_e/\F_q)$ defined by
\[\Gs:\; x_i\mapsto\Gg^{r+1}x_i \quad \mbox{for}\ i=1,2,\dots,e.\]
For an integer $m$ with $1\le m\le q-1$, let $\Pin$ and $P_i^{\Gs_j}$ for $i=1,2,\dots, N$ and $j=0,1,\dots m-1$ be the same as defined in Subsection \ref{subsec:gs-prelim}.

The folded codes from the Garcia-Stichtenoth tower are defined similarly to the Hermitian case.

\begin{defn}[Folded codes from the Garcia-Stichtenoth tower]
\label{def:fgs}
Assume that $m,k,N$ are positive integers satisfying $1\le m\le r-1$ and $l/m < N \le r^{e}\left\lfloor\frac{r-1}m\right\rfloor$.
The  folded  code from $K_e$ with parameters $N,l,q,e,m$, denoted by ${\FGS}_e(N,l,q,m)$,   encodes a message function $f \in \cL(l\Pin)$ a folded codeword given in \eqref{eq:x7}.
\end{defn}

Then we have a similar result on parameters of ${\FGS}_e(N,l,q,m)$.
\begin{lemma}
The above code ${\FGS}_e(N,l,q,m)$ is an $\F_q$-linear code over alphabet size $q^{m}$, rate at least $\frac{l-\g_e+1}{Nm}$, and minimum distance at least $N  - \frac{l}{m}$.
\end{lemma}

Similar to the the Hermitian case, we need to redefine the code in terms of local expansion at a point. In the Hermitian case, we use coefficients of its power series expansion around $P_0$ which has a simple local parameter $x_1$. However, for the Garcia-Stichtenoth tower we do not have such a nice point $P_0$. Fortunately, we can use point $\Pin$ to achieve our mission, i.e., $\Pin$ has a simple local parameter $\frac 1{x_e}$.

\begin{defn}[Folded Garcia-Stichtenoth code using local expansion]
	\label{def:fgs-code}
	The folded Garcia-Stichtenoth code (FGS code for short) $\widetilde{\FGS}_e(N,k,q,m)$ maps 
	\[ \mv{f} = (f_0,f_1,\dots,f_{k-1}) \in \F_q^k \quad \text{to} \quad \FGS_e(N,k+2\g_e-1,q,m)(\phi_{\Pin}(\mv{f})) \in (\F_q^m)^N \ . \]
\end{defn}
The rate of the above code equals $k/(Nm)$ and its distance is at least $N-(k+2\g_e-1)/m$.

As in the Hermitian case, we now turn to the task of computing the local expansion around $\Pin$ of a basis for $\cl(l\Pin)$, which then suffices to compute the map $\phi_{\Pin}$ efficiently. The local expansion of $f \in \cL(l \Pin)$ 
around $\Pin$ is of the form
\begin{equation}\label{eq:x21} f= T^{-l}(f_0 + f_1 T + f_2 T^2 + \cdots ) \end{equation}
where $T := \frac 1{x_e}$ is the local parameter at $\Pin$ (the function $x_e$ has exactly one pole at $\Pin$).

\begin{lemma}
\label{lem:ps-basis-GS}
For any $n$, one can compute the first $n$ terms of the local expansion \eqref{eq:x21} of a basis of $\cl(l \Pin)$ at $\Pin$ using ${\rm poly}(n)$ operations over $\F_q$.
\end{lemma}
\begin{proof}
First let $h$ be a nonzero function in $\F_q(x_1,x_2,\dots,x_e)$ with $\nu_{\Pin}(h)=v\in\ZZ$. Assume that the local expansion
  $h=T^{v}\sum_{j=0}^{\infty}a_jT^j$ is known.
To find the local expansion $\frac 1h=T^{-v}\sum_{j=0}^{\infty}c_jT^j.$ Consider the identity
  \[1= \left(\sum_{j=0}^{\infty}c_jT^j\right)\left(\sum_{j=0}^{\infty}a_jT^j\right).\]
  Then by comparing the coefficients of $T^i$ in the above identity, one has $c_0=a_0^{-1}$ and $c_i=-a_0^{-1}(c_{i-1}a_1+\cdots+c_0a_i)$ can be easily computed recursively for all $i\ge 1$.

Thus, by the structure of the basis functions in (\ref{eq:x6}), it is sufficient to find an algorithm efficiently finding local expansions of $x_i$ at $\Pin$ for every $i=1,2,\dots,e$. We can inductively find the local expansions of $x_i$ at $\Pin$ as follows. We note that $\nu_{\Pin}(x_i)=-r^{e-i}$ for $i=1,2,\dots,e$.

For $i=e$, $x_e$ has the local expansion $\frac 1T$ at $\Pin$.

Now assume that we know the local expansion of $x_i$. Then we can easily compute the local expansion of $x_i^r+x_i$ and hence the local expansion of $1/(x_i^r+x_i)$. Let us assume that $1/(x_i^r+x_i)$ has local expansion
$1/(x_i^r+x_i)=T^{r^{e-i+1}}\sum_{j=0}^{\infty}\Ga_{j}T^j$ at $\Pin$ for some $\Ga_{i}\in\F_q$.
Assume that $1/x_{i-1}$ has the local expansion
$1/x_{i-1}=T^{r^{e-i+1}}\sum_{j=0}^{\infty}\Gb_{j}T^j$. To find $\Gb_j$, we consider the identity
 \[ T^{r^{e-i+1}}\sum_{j=0}^{\infty}\Gb_{j}T^j+ T^{r^{e-i+2}}\sum_{j=0}^{\infty}\Gb_{j}^rT^{rj} =\frac1{x_{i-1}}+\left(\frac1{x_{i-1}}\right)^r=\frac1{x_i^r+x_i}=T^{r^{e-i+1}}\sum_{j=0}^{\infty}\Ga_{j}T^j.\]
 By comparing the coefficients of $T^{j+r^{e-i+1}}$ in the above identity, we have that $\Gb_0=\Ga_0$ and $\Gb_j$ can be easily computed recursively by the following formula for all $i\ge 1$.
 \[\Gb_j=\left\{\begin{array}{ll}\Ga_j&\mbox{if $r\not| j$}\\
\Ga_j-\Gb_{j/r-1}^r&\mbox{if $r| j$.}
 \end{array}
 \right.
\]
Therefore, the local expansion of $x_{i-1}$ at $\Pin$ can be easily computed.
\end{proof}

Similar to the Hermitian case, by instantiating Theorem \ref{thm:folded-ag-general-ld} with our code $\widetilde{\FGS}_e(N,k,q,m)$, we obtain the following result.
\begin{theorem}
  \label{thm:folded-GS-ld}
One can find a representation of the $(s,r-1)$-ultra periodic subspace
containing all candidate messages $(f_0,f_1,\dots,f_{k-1})$ in polynomial time, when the fraction of errors $\tau = 1-t/N$ in its encoding by $\widetilde{\FGS}_e(N,k,q,m)$ satisfies
\begin{equation}
\label{eq:fgs-error-frac}
 \tau \le  \frac{s}{s+1} \left( 1 - \frac{k}{N(m-s+1)} \right) - \frac{3m}{m-s+1} \frac{r^e}{mN} \ .
 \end{equation}
\end{theorem}

\subsection{Subfield evaluation codes from the Garcia-Stichtenoth tower}
\label{subsec:decoding-gs}

Let $r$ be a prime power and let $q=r^2$. For $e\ge 2$, let $K_e$ be the function field $\F_q(x_1,x_2,\dots,x_{e})$  given by Garcia-Stichtenoth tower
\eqref{eq:x5}, with genus $\g_e \le r^e$.

Put $F=K_e$ and $F_m=\Fm\cdot K_e$.  Let  $P_1,P_2,\dots,P_N$ be the rational points of $F$ besides the place $\Pin$; we have $N \ge r^e(r-1)$.
 Let $k$ be the desired dimension of the code (where $1 \le k < N - 2\g_e$) and let $l = k + 2 \g_e -1$.
We will now instantiate the code $C(m;l)$ defined in \eqref{eq:agcode-subfield} with the Garcia-Stichtenoth function field $F_m$ and $P_1,P_2,\dots,P_N$ as evaluation points.  (We hide the dependence on $e$ and $N$ in the specification of the code $C(m;l)$ as implicit.)

Let us call the resulting code $C_{\text{GS}}(m;l)$. As in the previous sections, we will encode into subcode $\widetilde{C}_{\text{GS}}(m;k)$ using as message vector the first $k$ coefficients of the local expansion around $\Pin$.  The Garcia-Stichtenoth code with subfield evaluation using local expansion, $\widetilde{C}_{\text{GS}}(m;k)$ is defined from $C_{\text{GS}}(m;l)$ as in Definition~\ref{def:subfield-agsubcode}.
	

By virtue of Theorem~\ref{thm:ag-final}, we can now conclude the following:
\begin{cor}
\label{cor:gs-decoding}
The code $C_{\mathrm{GS}}(m;k+2g_e-1 )$ can be list decoded from up to
$\frac{s}{s+1} (N - k) - \frac{3s+1}{s+1} g_e$ errors, pinning down the messages to an $(s-1,m)$-ultra periodic subspace of $\F_q^{mk}$.
\end{cor}

We conclude the section by incorporating the trade-off between $g_e$ and $N$, and stating the rate vs. list decoding radius trade-off offered by these codes, in a form convenient for improvements to the list size using subspace evasive sets and subspace designs (see Section \ref{sec:sd}). The claim about the number of possible solution subspaces follows since the subspace is determined by $A_0,A_1,\dots,A_s$, and for our choice of parameter $D$, there are at most $q^{O(mN)}$ choices of those.
\begin{theorem}
\label{thm:gs-decoding}
Let $q$ be the even power of a prime. Let $1 \le s \le m$ be integers, and let $R\in (0,1)$. Then for infinitely many $N$ (all integers of the form $q^{e/2} (\sqrt{q}-1)$), there is a deterministic polynomial time construction of an $\F_{q^m}$-linear code $\mathrm{GS}^{(q,m)}[N,k]$ of block length $N$ and dimension $k = R\cdot N$ that can be list decoded in $\mathrm{poly}(N,m,\log q)$ time from \[ \frac{s}{s+1} (N-k) -  \frac{3N}{\sqrt{q}-1} \]
 errors, pinning down the messages to one of $q^{O(mN)}$ possible $(s-1,m)$-ultra periodic $\F_q$-affine subspaces of $\F_q^{mk}$.
\end{theorem}

\section{Hierarchical subspace-evasive sets}
\label{sec:hse}

Let us first recall the notion of ``ordinary" subspace-evasive sets from \cite{Gur-ccc11}.

\begin{defn}
A subset $S \subset \F_q^k$ is said to be $(d,\ell)$-subspace-evasive if for all $d$-dimensional affine subspaces $H$ of $\F_q^k$, we have $|S \cap H| \le \ell$.
\end{defn}

We next define the notion of evasiveness w.r.t a collection of subspaces instead of all subspaces of a particular dimension.
\newcommand{\calF}{\mathcal{F}}
\begin{defn}
Let $\calF$ be a family of (affine) subspaces of $\F_q^k$, each of dimension at most $d$.
A subset $S \subset \F_q^k$ is said to be $(\calF,d,\ell)$-evasive if for all $H \in \calF$, we have $|S \cap H| \le \ell$.
\end{defn}

The key to pruning the list to a small size is the notion of a {\em hierarchical subspace-evasive set}, which is defined as a subset of $\F_q^k$ with the property that some of its prefixes are subspace-evasive with respect to $(s,\period,b)$-periodic subspaces.
We will show how the special subspace-evasive sets help towards pruning the list in our list decoding context in Section~\ref{sec:hse-pruning}.
\begin{defn}
Let $\calF$ be a family of $(s,\period,b)$-periodic subspaces  of $\F_q^k$ with $k=b\period$.
A subset $S \subset \F_q^k$ is said to be $(\calF,s,\period,b, L)$-h.s.e (for hierarchically subspace evasive for block size $\period$) if for every affine subspace $H \in \calF$,  the following bound holds for $j=1,2,\dots,b$:
\[ |\proj_{j\period}(S) \cap \proj_{j\period}(H)| \le L  \ . \]
\end{defn}

\begin{rmk}
\label{rmk:periodic-hse}
For h.s.e based pruning, a property weaker than $(s,\period)$-periodicity of $H$ suffices. Namely, it is enough if for each prefix $a \in \F_q^{j\Delta}$, the extensions of $a$ in $H$ form an affine space of dimension $s$ (it is not necessary that this be a coset of the \emph{same} subspace of $\F_q^\period$ for every $j$). However, we stick with the periodicity assumption since it is available to us in the subspaces output by the list decoder, and is also necessary for the subspace design based pruning of the next section.
\end{rmk}


\subsection{Random sets are 
subspace evasive}

Our goal is to give a randomized construction of large h.s.e sets that works with high probability, with the further properties that one can index into elements of this set efficiently (necessary for efficient encoding), and one can check membership in the set efficiently (which is important for efficient decoding).

An easy probabilistic argument, see \cite{Gur-ccc11}, shows that a random subset of $\F_q^k$ of size about $q^{(1-\zeta)k}$ is $(d,O(d/\zeta))$-subspace evasive with high probability. As a warmup, let us work out the similar proof for the case when we have only to avoid a not too large family $\calF$ of all possible $d$-dimensional affine subspaces. The advantage is that the guarantee on the intersection size is now $O(1/\zeta)$ and independent of the dimension $d$ of the subspaces one is trying to evade.
%
\begin{lemma}
\label{lem:prob-subspace-avoiding}
Let $\zeta \in (0,1)$ and $k$ be a large enough positive integer.
Let $\calF$ be a family of affine subspaces of $\F_q^k$, each of dimension at most $d \le \zeta k/2$, with $|\calF| \le q^{c k}$ for some positive constant $c$.

Let $S$ be a random subset of $\F_q^k$ chosen by including each $x
  \in \F_q^k$ in $S$ with probability $q^{-\zeta k}$. Then with probability at least $1-q^{-ck}$,
  $S$ satisfies both the following conditions: (i) $|S| \ge
  q^{(1-2\zeta) k}$, and (ii) $S$ is
  $(\calF,d,4c/\zeta)$-evasive.
\end{lemma}   
\begin{proof}
  The first part follows by noting that the expected size of $S$ equals $q^{(1-\zeta)k}$ and a standard Chernoff bound calculation. For the second part, fix an affine subspace $H \subseteq \calF$ of dimension at most $d$, and a subset $T \subseteq H$ of size $t$, for some parameter $t$ to be specified shortly.
  The probability that $S \supseteq T$ equals $q^{-\zeta k t}$.
  By a union bound over the at most $q^{c k}$ choices for the affine subspace $H \in \calF$, and the at most $q^{d t}$ choices of $t$-element subsets $T$ of $H$, we get that the probability that $S$ is not  $(\calF,d,t)$-evasive is at most
$q^{c k +d t}  \cdot q^{-\zeta k t} \le q^{c k } q^{-\zeta k t/2}$ since $d \le \zeta k/2$.
Choosing $t = \lceil 4c/\zeta \rceil$, this quantity is bounded from above by $q^{-ck}$.
\end{proof}

  \subsection{Pseudorandom construction of large h.s.e subsets}
We next turn to the pseudorandom construction of large h.s.e subsets. Suppose, for some fixed subset $\calF$ of $(s,\period,b)$-periodic subspaces of $\F_q^k$ with $k=b\period$, we are interested  in an $(\calF,s,\period,b,L)$-h.s.e subset of $\F_q^k$ of size $\approx q^{(1-\zeta)k}$ for a constant $\zeta$, $1/\period < \zeta < 1/3$.
(Bwlow, we will ignore floors and ceilings in the description to avoid notational clutter; those are easy to accommodate and do not affect any of the claims.)

Denote $\period' =(1-\zeta)\period$, $b' = (1-\zeta) b$, and $k' = b' \period = (1-\zeta) k$.

The random part of the construction will consist of mutually independent, random univariate polynomials $P_1,P_2,\dots,P_{b'}$ and $Q$, where $P_j \in \F_{q^{j\period'}}[T]$ for $1\le j \le b'$  and $Q \in \F_{q^{k'}}[T]$ are random polynomials of degree $\degree$.\footnote{We will assume that representations of the necessary extension fields $\F_q^{i\period'}$ are all available. For this purpose, we only need irreducible polynomials over $\F_q$ of appropriate degrees, which can be constructed by picking random polynomials and checking them for irreducibility. Our final construction is anyway randomized, so the randomized nature of this step does not affect the results.}
The degree parameter will be chosen to be $\degree = \Theta(k)$.\footnote{The degree of $Q$ can in fact be just $O(1/\zeta)$, but for uniformity we fix the degree of all polynomials to be the same.}

The key fact we will use about random polynomials is the following, which follows by virtue of the $\lambda$-wise independence of the values of a random degree $\lambda$ polynomial.
\begin{fact}
\label{fact:poly-indep}
Let $P \in \mathbb{K}[T]$ be a polynomial of degree $\degree$ whose coefficients are picked uniformly and independently at random from the field $\mathbb{K}$.
For a fixed subset $T \subseteq \mathbb{K}$ with $|T| \le \degree$, the values $\{P(\alpha)\}_{\alpha \in T}$ are independent random values in $\mathbb{K}$.
\end{fact}

We remark that this property of low-degree polynomials was also the basis of the pseudorandom construction of subspace evasive sets in \cite{GW-tit13}. However, since we require the h.s.e property, and need to exploit the periodicity of the subspaces we are trying to evade (which can have large dimension), the construction here is more complicated, and needs to use several polynomials $P_j$'s evaluated in a nested fashion, and one further polynomial $Q$ to further bring down the list size to a constant (this final use of $Q$ is similar in spirit to the construction in \cite{GW-tit13}). We remark that the construction presented here is a bit simpler and cleaner than the one in the conference version~\cite{GX-stoc12}, and comes with efficient encoding automatically by construction. In contrast, the construction in \cite{GX-stoc12} required some additional work in order to allow for efficient encoding.


In what follows we assume that, for $j=1,2,\dots,b'$, some fixed bases of the fields $\F_{q^{j \period'}}$ have been chosen, giving us some canonical $\F_q$-linear injective maps
\[ \rho_j  : \F_q^{j\period'} \to \F_{q^{j \period'}} \ . \]
Also, for $j=1,2,\dots,b'$, let
\[ \xi_j : \F_{q^{j\period'}}\to \F_q^{\zeta \period} \]
be some arbitrary $\F_q$-linear surjective map (thus $\xi_j$ just
outputs the first $\zeta \period$ coordinates of the representation of
elements of $\F_{q^{j\period'}}$ as vectors in $\F_q^{j\period'}$
w.r.t some fixed basis). Finally, let $\rho : \F_q^{k'} \to
\F_{q^{k'}}$ be some fixed $\F_q$-linear injective map, and $\xi :
\F_{q^{k'}} \to \F_q^{\zeta k}$ be an arbitrary $\F_q$-linear
surjective map.

We are now ready to describe our construction of h.s.e set based on the random polynomials $P_1,P_2,\dots,P_{b'},Q$.

\begin{defn}[h.s.e set construction]
\label{def:hse-set}
Given the polynomials $P_j \in \F_{q^{j\period'}}[T] $ for $i=1,2,\dots,b'$ and $Q \in \F_{q^{k'}}[T]$, define the subset  $\Gamma(P_1,P_2,\dots,P_b; Q)$ by
 \begin{align*}
  \Bigl\{ (y_1,z_1,y_2,z_2,\dots , y_{b'},z_{b'}; w) \in \F_q^k ~~ & \vline ~~ \text{ for } j=1,2,\dots,b': y_j \in \F_q^{\period'}, \\
  &  ~z_j = \xi_j(P_j(\rho_j(y_1\circ y_2 \circ \cdots \circ y_j)))  \in \F_q^{\zeta\period}; \text{ and} \\
  & ~ w =  \xi(Q(\rho(y_1,z_1,\dots, y_{b'},z_{b'}))) \in\F_q^{\zeta k}\Bigr\} \ .
  \end{align*}
  \end{defn}

By construction, once suitable representations of the extension fields are available by pre-processing and the choice of $P_1,\dots,P_{b'},Q$ is made, we can efficiently compute a bijective encoding map
$\hse: \F_q^{(1-\zeta)^2 k} \to \Gamma(P_1,P_2,\dots,P_b;Q)$. Indeed, we can view the input $\mv{y} \in \F_q^{b'\period'}$ as $(y_1,y_2,\dots,y_{b'})$ with $y_j \in \F_q^{\period'}$ and then compute the $z_j$'s and $w$ efficiently using $\mathrm{poly}(k)$ operations over $\F_q$ (recall that the degree of the polynomials is $\degree = \Theta(k)$).

We now move on to the main claim about the h.s.e property of our construction.

\begin{theorem}
\label{thm:hse-proof}
Let $c$ be a positive constant. Let $\zeta \in (0,1/3)$ and $s$ be a positive integer satisfying $s < \zeta \period/10$. Let $\calF$ be a subset of at most $q^{c k}$ $(s,\period, b)$-periodic subspaces of $\F_q^k$ for $k=b\period$ that is much bigger than $1/\zeta$. Suppose that the parameters satisfy the condition $q^{\zeta \period} \ge ( 2 q^2 c k)^{10/9}$. Then with probability $1-q^{-\Omega(k)}$ over the choice of random polynomials $\{P_i\}_{1 \le i \le b}$ and $Q$ each of degree $\lambda = \lceil ck \rceil$, the set $\Gamma(P_1,P_2,\dots,P_b;Q)$ from Definition~\ref{def:hse-set} is
\[ (\calF,s,\period,b,L) \mbox{-h.s.e ~~{\em and}~~ }  (\calF,sb, \ell) \text{-evasive} \]
 for $L = \lceil 2c k \rceil$ and $\ell = \lceil 4 c/\zeta \rceil$ (note that (i) $L\gg\ell$ as $k\gg 1/\zeta$; and (ii) $Q$ trims down the intersection size from $L$ to $\ell$).
\end{theorem}
\begin{proof}
Note that the first $k'=(1-\zeta)k$ symbols of vectors in $\Gamma(P_1,\dots,P_{b'};Q)$ only depend on the $P_j$'s.
We will first prove that with high probability over the choice of the $P_j$'s the following holds (call such a choice of $P_j$'s as \emph{good}):
\begin{quote}
  For every $H \in \calF$, $|\proj_{k'}(H) \cap \proj_{k'}(\Gamma)| < L$, where we denote $\Gamma$ as shorthand for $\Gamma(P_1,\dots,P_{b'};Q)$.
\end{quote}
Then, conditioned on a good choice of $P_j$'s, we will prove that with high probability over the choice of the random polynomial $Q$, $|H \cap \Gamma| < \ell$.
Together, these steps will imply that the set $\Gamma(P_1,P_2,\dots,P_{b'};Q)$ is $(\calF,sb, \ell) \text{-evasive}$. (Note that every subspace in $\calF$ has dimension at most $s b$ by Claim \ref{clm:periodic-subspace}.) We will return to the $(\calF,s,\period,b,L) \text{-h.s.e}$ property at the end of the proof.

Let us first establish the second step. Fix a good choice of $P_1,\dots,P_{b'}$, and suppose we pick $Q$ randomly. Fix a subspace $H \in \calF$.  Since $|\proj_{k'}(H) \cap \proj_{k'}(\Gamma)| < L$ (recall that $\proj_{k'}(\Gamma)$ only depends on the $P_j$'s and thus is already determined), the number of elements of $H$ that could possibly belong to $\Gamma$ (after the choice of $Q$) is at most $L \cdot q^{s (b-b')} = L q^{\zeta s b}$; indeed for each prefix belonging to $\proj_{k'}(\Gamma) \cap \proj_{k'}(H)$, there are most $q^{s (b-b')}$ extensions that can fall in $H$ since $H$ is $(s,\period,b)$-periodic. Further, the probability over the choice of $Q$ that any such fixed extension belongs to $\Gamma$ is at most $q^{-\zeta k}$, and any $\ell$ of these events are independent. (Note that for a fixed prefix, there can be at most one extension that falls in $\Gamma$, so for $\ell$ different strings to fall in $\Gamma$, their prefixes must be distinct and are mapped to independent locations by the random polynomial $\Gamma$.)
Therefore, the probability over the choice of $Q$ that  $|H \cap \Gamma| \ge \ell$ is at most $(L q^{\zeta s b})^\ell q^{-\zeta k \ell}$. By a union bound over all $H \in \calF$, we conclude that $|H \cap \Gamma| < \ell$ for every $H \in \calF$ simultaneously, except with probability at most
\[ q^{ck} L^\ell q^{\zeta (s -\period) b \ell} \le q^{ck} (ck)^\ell q^{-\zeta \period b \ell/2} \le q^{ck} q^{-\zeta k \ell/4}
\]
where in the first inequality we used $s \le \period/2$ and in the next one  $ck \le q^{\zeta k/4}$ both of which hold comfortably. For $\ell \ge 4 c/\zeta$, the above probability upper bound is at most $q^{-ck}$.

We now turn to the first step, on the $P_j$'s being good with high probability. Fix some $H \in \calF$; we will prove by induction on $j$ that
\begin{equation}
\label{eq:proj-ind}
|\proj_{j \period}(H) \cap \proj_{j \period}(\Gamma)| < L
\end{equation}
w.h.p over the choice of $P_1,P_2,\dots,P_j$, for $1 \le j \le b'$ (note that $\proj_{j \period}(\Gamma)$ only depends on $P_1,\dots,P_j$, so this event is well defined). For the base case $j=1$, $|\proj_{\period}(H)| \le q^s$ as $H$ is $(s,\period,b)$-periodic, and the probability that some $L$ of these $q^s$ elements belong to $\proj_{\period}(\Gamma)$ is at most $q^{sL}$ times the probability that $L$ distinct elements in $\F_{q^{\period'}}$ are mapped to specific values in $\F_q^{\zeta \period}$ by $\xi_1\circ P_1$,
which is at most $\bigl( q^{-\zeta \period}\bigr)^L$. So the overall probability that $|\proj_{\period}(H) \cap \proj_{\period}(\Gamma)| \ge L$ is at most $q^{(s -\zeta \period)L}$.

Now let $j \ge 2$ and assume $|\proj_{(j-1) \period}(H) \cap \proj_{(j-1) \period}(\Gamma)| < L$. By the $(s,\period,b)$-periodicity of $H$, for each of the (less than $L$) prefixes in $\proj_{(j-1) \period}(H) \cap \proj_{(j-1) \period}(\Gamma)$, there are at most $q^s$ extensions that fall in $\proj_{j \period}(H)$. Similarly to the argument used for second step above, the probability that some $L$ of these belong to $\proj_{j \period}(\Gamma)$ is at most $(L q^s)^L \cdot q^{-\zeta \period L}$. Thus, the probability that $|\proj_{j \period}(H) \cap \proj_{j \period}(\Gamma)| \ge L$ is at most $\bigl(L  \cdot q^{(s-\zeta \period)}\bigr)^L$.

Combining these arguments, we conclude that the probability over the choice of the $P_j$'s that $|\proj_{b' \period}(H) \cap \proj_{b' \period}(\Gamma)| \ge L$ is at most
\[ b' (L \cdot q^{(s -\zeta \period)})^L \le (2ck q^{-0.9\zeta \period})^L \le q^{-2L} \]
where the last step used the assumption that $q^{\zeta \period} \ge (2 q^2 c k)^{10/9}$.

Finally, since there are at most $q^{ck}$ subspaces $H \in \calF$, by a union bound we have that for all $H \in \calF$ simultaneously,  $|\proj_{k'}(H) \cap \proj_{k'}(\Gamma)| < L$ with probability at least $1- q^{ck} q^{-L} \ge 1-q^{-ck}$ over the choice of $P_1,\dots,P_{b'}$.

To finish the proof, we need to verify the $(\calF,s,\period,b,L) \text{-h.s.e}$ property. That is, we need to prove that w.h.p,
$|\proj_{j \period}(H) \cap \proj_{j \period}(\Gamma)| \le L$ for every $H \in \calF$ and $j=1,2,\dots,b$.  By \eqref{eq:proj-ind}, this holds for $j=1,2,\dots,b'$. By construction, the last $\zeta k$ symbols of any vector in $\Gamma$ is a function of the first $(1-\zeta)k = b' \period$ symbols, so
$|\proj_{j \period}(H) \cap \proj_{j \period}(\Gamma)| \le L$ also holds for $b' < j \le b$.
\end{proof}

\subsection{Efficient computation of intersection with h.s.e. subsets}
\label{sec:hse-pruning}
The key aspect which makes h.s.e subsets useful in our context to prune the affine space of candidate messages, and indeed motivated the exact specifics of the definition and aspects of its construction, is the following claim which shows that intersection of a $(s,\period,b)$-periodic subspace with our h.s.e set can be found efficiently.
\begin{lemma}(h.s.e.-intersection)
\label{lem:hse-intersection}
There is an algorithm running in time $\mathrm{poly}(k, q^{\zeta \period})$ that provides the following guarantee. Given as input the polynomials $P_1,\dots,P_{b'}$ and $Q$ underlying the construction of an
$(\calF,s,\period,b,L) \mbox{-h.s.e}$ and $(\calF,sb, \ell) \text{-evasive}$ set $\Gamma = \Gamma(P_1,\dots,P_{b'};Q)$ and an $(s,\period,b)$-periodic subspace $H \subseteq \F_q^k$ belonging to $\calF$, the algorithm computes the at most $\ell$ elements of $H \cap \Gamma$.
\end{lemma}
\begin{proof}
The proof essentially follows from the observations made in the proof of Theorem~\ref{thm:hse-proof}. First note  that $|H \cap \ \Gamma(P_1,\dots,P_{b'};Q)| \le \ell$ just follows from the $(\calF,sb, \ell) \text{-evasiveness}$ of $\Gamma$. To compute $H \cap \Gamma$, the algorithm iteratively computes the intersections $\proj_{j\period}(H) \cap \proj_{j\period}(\Gamma)$ for $1 \le j \le b'$. As $\Gamma$ is
$(\calF,s,\period,b,L) \mbox{-h.s.e}$, this intersection has size at most $L$. To compute $\proj_{j\period}(H) \cap \proj_{j\period}(\Gamma)$, the algorithm runs over the at most $q^s$ possible extensions of each element of $\proj_{(j-1)\period}(H) \cap \proj_{(j-1)\period}(\Gamma)$ that can belong to $\proj_{j \period}(H)$ (due to the $(s,\period,b)$-periodicity of $H$), and checks which ones also belong to $\proj_{j\period}(\Gamma)$. The complexity amounts to $q^{O(s)}$ evaluations of degree $O(k)$ polynomials, and thus takes $q^{O(\zeta \period)} \mathrm{poly}(k)$ time. To compute $H \cap \Gamma$ from $\proj_{b'\period}(H) \cap \proj_{b' \period}(\Gamma)$, we recall the earlier observation that the construction of $\Gamma$ implies that there is a unique extension of an element in $\proj_{b' \period}(\Gamma)$ that belongs to $\Gamma$.
\end{proof}

We conclude this section by recording in convenient form all necessary properties of our h.s.e set construction, which follow from Theorem~\ref{thm:hse-proof} and Lemma~\ref{lem:hse-intersection}.  (We can remove the restriction that $k$ is a multiple of $\period$ by constructing a subspace in $\F_q^{k^{\#}}$ for $k^{\#} = \period \lceil \frac{k}{\period}\rceil$ and dropping the last $k^{\#}-k$ coordinates, so we remove that restriction in the final statement below on h.s.e sets.)

\begin{theorem}
	\label{thm:hse-set-final} Let $c$ be a constant.
	Let $\zeta \in (0,1)$, and $\period, s,k$ be positive integers satisfying $s < \zeta \period/10$ and $k \le q^{\zeta \period/2}$.
	Let $\calF$ be a family of at most $q^{ck}$ $(s,\period)$-periodic subspaces of $\F_q^k$. Then there is $\mathrm{poly}(k,\log q)$ time randomized construction of an
	 injective map $\hse : \F_q^{(1-\zeta)^2 k} \to \F_q^k$ such that:
	\begin{enumerate}
		\item Given $\mv{x} \in \F_q^{(1-\zeta)^2 k}$, $\hse(\mv{x})$ can be computed using $\mathrm{poly}(k)$ operations over $\F_q$.
		\item With probability at least $1-q^{-\Omega(k)}$ over the construction of $\hse$, the following holds: for every $H \in \calF$, the set $\{ \mv{x} \in \F_q^{(1-\zeta)^2 k} \mid \hse(\mv{x}) \in H\}$ has size at most $O(c/\zeta)$, and further can be computed in $\mathrm{poly}(k,q^{\zeta \period})$ time.
	\end{enumerate}
\end{theorem}

\section{Subspace designs}
\label{sec:sd}
The linear-algebraic list decoder discussed in the previous sections
pins down the coefficients of the message to a periodic subspace.
 We already saw, in Section~\ref{sec:hse}, an approach using h.s.e. sets to prune the periodic subspace to a small list. In this section, we will develop an alternate approach based a special collection of
subspaces, which we call a {\em subspace design}, for pruning the periodic subspaces.
Further, we will extend the construction to a ``cascaded" variant that enables more effective pruning of ultra-periodic subspaces. The advantage of using subspace designs is they can be explicitly constructed, a feature which is (currently) lacking for h.s.e sets.

We begin with the definition of the central object of study in this section, subspace designs, introduced in the conference version~\cite{GX-stoc13}.\footnote{While we were not aware of it when we coined this term to refer to our subspace collections, subspace designs were used to denote the $q$-analogs of combinatorial designs~\cite{BKW-subs-design}. We apologize for unknowlingly repeating this nomenclature in our (very different) context.}
\begin{defn}[Subspace design]
Let $\dm$ be a positive integer, and $q$ a prime power.
For positive integers $r < \dm$ and $d$, an $(r,d)$-\sd\ in $\F_q^\dm$ is a collection $H$ of subspaces of $\F_q^\dm$ such that for every $r$-dimensional subspace $W \subset \F_q^\dm$, we have
\begin{equation}
\label{eq:subs-design-cond}
\sum_{H \in \mathcal{H}} \dim(W \cap H) \le d \ .
\end{equation}
The cardinality of a \sd\ $\mathcal{H}$ is the number of subspaces in its collection, i.e., $|\mathcal{H}|$. If all subspaces in $\mathcal{H}$ have the same dimension $t$, then we refer to $t$ as the {\em dimension} of the \sd\ $\mathcal{H}$.
\end{defn}

\noindent Note that the condition \eqref{eq:subs-design-cond} in particular implies for every $r$-dimensional subspace $W$,  at most $d$ of the subspaces in an $(r,d)$-subspace design non-trivially intersect it. This weaker property was subsequently called a ``weak subspace design" in \cite{GK-combinatorica} which gave explicit constructions of subspace designs following our original definition in \cite{GX-stoc13}. For our list decoding application, the stronger property \eqref{eq:subs-design-cond} is required. Note though that the weak subspace design property implies the stronger \eqref{eq:subs-design-cond} with the r.h.s upper bound $d$ replaced by $dr$.

\subsection{Subspace designs to prune periodic subspaces}
The usefulness of \sds\ defined above, in the context of pruning periodic subspaces, is captured by the following key lemma.

\begin{lemma}[Periodic subspaces intersected with a subspace design]
\label{lem:sd-pruning}
Suppose $H_1,H_2,$ $\dots,H_b$ are subspaces in an $(r,d)$-\sd\ in $\F_q^\dm$, and $T$ is an $(r,\dm,b)$-periodic affine subspace of $\F_q^{\dm b}$ with recurring subspace $S\subseteq \F_q^{\dm }$. Then the set
\[ \mathcal{T} = \{ (\mv{f_1},\mv{f_2},\dots,\mv{f_b}) \in T \mid \mv{f_j} \in H_j \text{ for } j=1,2,\dots,b \} \]
is an affine subspace of $\F_q^{\dm b}$ of dimension at most $d$. Also, the
underlying subspace of $\mathcal{T}$ is contained in $\mathcal{S} \eqdef S^b \cap (H_1 \times H_2 \times \cdots \times H_b)$.
\end{lemma}
\begin{proof}
It is clear that $\mathcal{T}$ is an affine subspace, since its
elements are restricted by the set of linear constraints defining $T$
and the $H_j$'s. Also, the difference of two elements in $\mathcal{T}$
is contained in both the subspaces $S^b$ and $(H_1 \times H_2 \times
\cdots \times H_b)$, which implies that the underlying subspace of
$\mathcal{T}$ is contained in $\mathcal{S}$.

We will prove the bound on dimension by proving that $|\mathcal{T}|
\le q^d$. To prove this, we will imagine the elements of $\mathcal{T}$
as the leaves of a tree of depth $b$, with the nodes at level $j$
representing the possible projections of $\mathcal{T}$ onto the first
$j$ blocks. The root of this tree has as children the elements of the
affine space $\proj_{[1,\dm]}(T) \cap H_1$.  Let $W$ be the subspace
of $\F_q^\dm$ of dimension at most $r$ associated with the periodic
subspace $T$ (in the sense of Definition
\ref{def:periodic-subspaces}). Note that the underlying subspace of
the affine space $\proj_{[1,\dm]}(T) \cap H_1$ is contained in the
subspace $W \cap H_1$.

Continuing this argument, the children of an element $\mv{a} \in \F_q^{j \dm}$ at level $j$ will be $\mv{a}$ followed by the  possible extensions of $\mv{a}$ to the $(j+1)$'th block, given by
\[ \{ \proj_{[j\period+1,(j+1)\period]}(\mv{x}) \mid \mv{x} \in T \mbox{ and } \proj_{j\period}(\mv{x}) =\mv{a} \} \cap H_{j+1} \ .  \]
The periodic property of $T$ and the fact that $H_{j+1}$ is a subspace
implies that the possible extensions of $\mv{a}$ are given by a coset of
a subspace of $W \cap H_{j+1}$.  Thus the nodes at level $j$ have
degree at most $q^{\dim(W \cap H_{j+1})}$ for $j=0,1,\dots,b-1$. Since
the $H_j$'s belong to an $(r,d)$-\sd\, we have $\sum_{j=1}^b \dim(W
\cap H_j) \le d$. Therefore, the tree has at most $q^d$ leaves, which
is also an upper bound on $|\mathcal{T}|$.
\end{proof}

\subsection{Existence and probabilistic construction of subspace designs}
We now turn to the construction of \sds\ of large size and dimension. We first analyze the performance of a random collection of subspaces.

\begin{lemma}
\label{lem:sd-random}
Let $\eta > 0$ and $q$ be a prime power. Let $r,\dm$ be integers $\dm
\ge 8/\eta$ and $r \le \eta \dm/2$. Consider a collection
$\mathcal{H}$ of subspaces of $\F_q^\dm$ obtained by picking,
independently at random, $q^{\eta \dm/8}$ subspaces of $\F_q^\dm$ of
dimension $(1-\eta)\dm$ each. Then, with probability at least
$1-q^{-\dm r}$, $\mathcal{H}$ is an $(r,8r/\eta)$-\sd.
\end{lemma}
\begin{proof}
Let $\ell = 8 r/\eta$, and let $M=q^{\eta \dm/8}$ denote the number of randomly chosen subspaces.\footnote{For simplicity, we ignore the floor and ceil signs in defining integers; these can be easily incorporated.} Let $H_1,H_2,\dots,H_M$ be the subspaces in the collection $\mathcal{H}$.
Fix a subspace $W$ of $\F_q^\dm$ of dimension $r$. Fix a tuple of non-negative integers $(a_1,a_2,\dots,a_M)$ summing up to $\ell$. For each $j \in \{1,2,\dots,M\}$, the probability that $\dim(W \cap H_j) \ge a_j$ is at most $q^{r a_j} q^{-\eta \dm a_j}$. Since the choice of the different $H_j$'s are independent, the probability that $\dim(W \cap H_j) \ge a_j$ for every $j$ is at most $q^{(r-\eta \dm) \ell} \le q^{-\eta \dm \ell/2}$ (the last step uses $r \le \eta \dm/2$).

A union bound over the at most $q^{\dm r}$ subspaces $W \subset \F_q^\dm$ of dimension $r$, and the at most ${{\ell + M} \choose \ell}\le (M+\ell)^\ell \le M^{2\ell}$ choices of the tuples $(a_1,a_2,\dots,a_M)$, we get the probability that $\mathcal{H}$ is {\em not} an $(r,\ell)$-\sd\ is at most
\[ q^{\dm r} \cdot q^{-\eta \dm \ell/2} \cdot (q^{\eta \dm/8})^{2\ell} = q^{\dm r} \cdot q^{-\eta \dm \ell /4} \le q^{-\dm r} \]
where the last step uses $\ell \ge 8r/\eta$.
\end{proof}

Note that given a collection $\mathcal{H}$ of subspaces, one can
deterministically check if it is an $(r,d)$-\sd\ in $\F_q^\dm$ in
$q^{O(\dm r)} |\mathcal{H}|$ time by doing a brute-force check of all
$r$-dimensional subspaces $W$ of $\F_q^{\dm}$, and for each computing
$\sum_{H \in \mathcal{H}} \dim(W \cap H)$ using $|\mathcal{H}|
\dm^{O(1)}$ operations over $\F_q$.  Thus the above lemma
gives a $q^{O(\dm r)}$ time {\em Las Vegas} construction of an $(r,d)$-\sd\ with many
subspaces each of large dimension $(1-\eta)m$.

\begin{lemma}
\label{lem:sd-derandomized}
For parameters $\eta,r,\dm$ as in Lemma \ref{lem:sd-random}, for any
$b \le q^{\eta \dm/8}$, one can compute an $(r,8r/\eta)$-\sd\ in
$\F_q^\dm$ of dimension $(1-\eta)\dm$ and cardinality $b$ in $q^{O(\dm r)}$ {\em Las Vegas} time.
\end{lemma}

As noted in the conference version \cite{GX-stoc13} of this paper, the  construction can be derandomized
using the method of
conditional expectations to successively find good subspaces $H_i$ to add to the subspace design. However, as each step involves searching over all $(1-\eta)\dm$-dimensional subspaces of $\F_q^{\dm}$, the construction time would be $q^{O(\dm^2)}$ even for constructing subspace designs with few subspaces. For our application to reducing the list size for long algebraic-geometric codes (either folded or with rational points in a subfield), we will need subspace designs for ambient dimension $\dm$ growing at least logarithmically in the code length. The  $q^{O(\dm^2)}$ complexity will thus lead to a quasi-polynomial code construction time, as claimed in the conference version~\cite{GX-stoc13}. In fact, even the Las Vegas construction time of $q^{O(\dm r)}$ will be super-polynomial for the parameters used in the construction.

\subsection{Explicit subspace design constructions}

The question of explicit (polynomial time) constructions of subspace designs naturally arose following \cite{GX-stoc13} and was addressed in the follow-up work by Guruswami and Kopparty~\cite{GK-combinatorica}, who proved the following.

\begin{theorem}[Explicit subspace designs~\cite{GK-combinatorica}]
\label{thm:sd-explicit}
For every $\eta > 0$, integers $r,\dm$ with $r \le \eta \dm/4$, and prime powers $q$ satisfying $q^{\eta \dm/(2r)} > 2r/\eta$, for any $b \le q^{\eta \dm/(4r)}$, there exists an explicit $(r, r^2/\eta)$-\sd\ of cardinality $b$ and dimension $(1-\eta) \dm$, that can be constructed deterministically in time $\mathrm{poly}(b,q)$ time. In the case when $q > \dm$, one can explicitly construct an $(r,2r/\eta)$-\sd\ with the same parameters.
\end{theorem}

We note a couple of senses in which the parameters offered by the
explicit construction are weaker than those guaranteed by the
probabilistic construction. First, the total intersection dimension
\eqref{eq:subs-design-cond} is $r^2/\eta$ rather than $O(r/\eta)$
(except when $q$ is large). This is because, for small fields, their
construction yields only a weak subspace design, incurring a factor
$r$ loss when passing to a subspace design. Second, the number of
subspaces in the design is smaller, roughly $q^{\Omega(\eta \dm/r)}$
instead of $q^{\Omega(\eta \dm)}$. Finally, there is a modest
restriction the field size $q$, and we need to pick $r,\dm$ suitably
to allow for fixed $q$. Fortunately, all these restrictions can be
accommodated for our application. We remark that a recent construction
of subspace designs based on cyclotomic function
fields~\cite{GXY-tams} gives an $(r,O(r
\log_q \dm/\eta))$-\sd\ over an \emph{arbitrary} field $\F_q$; for our application, however, the $r^2/\eta$
bound is more useful as $r \ll \dm$, and we cannot afford the
dependence on $\dm$ in the bound.

Let us now record a construction of a subspace that has large dimension and yet has
low-dimensional
intersection with every periodic subspace. The construction is based on the above \sds.
This form will be convenient for later use in Section~\ref{subsec:put-together-rs} for pre-coding Reed-Solomon codes with evaluation points in a subfield.

\begin{theorem}
\label{thm:final-sd-subspace}
Let $\eta \in (0,1)$ and $q$ be a prime power, and $r,\dm,b$ be integers such that $r\le \eta \dm/4$ and $b < q$.
Then, one can construct a subspace $V$ of $\F_q^{b \dm}$ of dimension at least $(1-\eta) b\dm$ in  deterministic $q^{O(\dm)}$ time such that for every $(r,\dm,b)$-periodic subspace $T \subset \F_q^{b\dm}$, $V \cap T$ is an $\F_q$-affine subspace of dimension at most $2r/\eta$.
\end{theorem}
\begin{proof}
We will take $V = H_1 \times H_2 \times \cdots H_b$ where the $H_i$'s belong to a $(r,2r/\eta)$-\sd\ in $\F_q^\period$ of cardinality $b$ and dimension at least $(1-\eta)\dm$ as guaranteed by Theorem~\ref{thm:sd-explicit} when $q > \dm$. er Clearly $\dim(V) \ge (1-\eta)b\dm$ since each $H_i$ has dimension at least $(1-\eta)\dm$. The claim now follows using Lemma \ref{lem:sd-pruning}.
\end{proof}

\subsection{Cascaded subspace designs}
\label{subsec:csd}
\newcommand{\csd}{cascaded subspace design}
In preparation for our results about algebraic-geometric codes, whose block length $\gg q^m$
is much larger than the possible size of subspace designs in $\F_q^m$,
we now formalize a notion that combines several ``levels'' of subspace
designs. The definition might seem somewhat technical, but it has a
natural use in our application to list-size reduction for AG codes. Note that there is no ``consistency'' requirement between \sds\ at different levels other than the lengths and cardinalities matching.

\begin{defn}[Subspace designs of increasing length]
Let $l$ be a positive integer.  For positive integers $r_0 \le r_1 \le
\cdots \le r_l$ and $m_0\le m_1 \le \cdots \le m_l$ such that
$m_{\iota-1} | m_\iota$ for $1 \le \iota \le l$, an
$(r_0,r_1,\dots,r_l)$-\csd\ with length-vector $(m_0,m_1,\dots,m_l)$
and dimension vector $(d_0,d_1,\dots,d_{l-1})$ is a collection of $l$
\sds, specifically an $(r_{\iota-1},r_\iota)$-\sd\ in
$\F_{q^{m_{\iota-1}}}$ of cardinality $m_\iota/m_{\iota-1}$ and dimension
$d_{\iota-1}$ for each $\iota=1,2,\dots,l$.
\end{defn}

Note that the $l=1$ case of the above definition corresponds to an $(r_0,r_1)$-\sd\ in $\F_q^{m_0}$ of dimension $d_0$ and cardinality $m_1/m_0$.
In Lemma \ref{lem:sd-pruning}, we used the subspace $H_1 \times H_2 \times \cdots \times H_b$ based on a \sd\ consisting of the $H_i$'s to prune a periodic subspace. Generalizing this, we now define a subspace associated with a \csd\ based on the subspace designs comprising it.
\begin{defn}[Canonical subspace]
\label{def:csd}
Let $\mathcal{M}$ be a \csd\ with length-vector $(m_0,m_1,\dots,m_l)$ such that the $\iota$'th \sd\ in $\mathcal{M}$ has subspaces
\[ H^{(\iota)}_1, H^{(\iota)}_2,
\cdots, H^{(\iota)}_{m_\iota/m_{\iota-1}} \subset \F_q^{m_{\iota-1}} \ , \text{  for } 1 \le \iota \le l \ . \]
The {\em canonical subspace} associated with such a \csd, denoted
$U(\mathcal{M})$, is a subspace of $F_q^{m_l}$ defined as
  follows:
  \begin{quote}
  A vector $\mv{x} \in \F_q^{m_l}$ belongs to $U(\mathcal{M})$ if and only
  if for every $\iota \in \{1,2,\dots,l\}$, each of the $m_\iota$-sized blocks of
  $\mv{x}$ given $\proj_{[j m_\iota +1 , (j+1)m_\iota]}(\mv{x})$ for $0\le j < m_l/m_\iota$) belongs $H^{(\iota)}_1 \times H^{(\iota)}_2 \times \cdots
  \times H^{(\iota)}_{m_\iota/m_{\iota-1}}$.
\end{quote}
\end{defn}

In other words, we apply the construction of Lemma \ref{lem:sd-pruning} for (disjoint) intervals of length $m_\iota$ at each level $\iota \in \{1,2,\dots,l\}$.

\noindent The following simple fact, which follows by counting number of linear constraints imposed, gives a lower bound on the dimension of a canonical subspace.

\begin{obs}
\label{obs:csd-dim}
For a \csd\ $\mathcal{M}$ as above, if the $\iota$'th subspace design has dimension at least $(1-\xi_{\iota-1}) m_{\iota-1}$ for $1 \le \iota \le l$, then the dimension of the canonical subspace $U(\mathcal{M})$ is at least $\Bigl(1 - (\xi_0+\xi_1+\cdots+\xi_{l-1})\Bigr) m_l$.
\end{obs}

The following is the crucial claim about pruning ultra-periodic subspaces using (the canonical subspace of) a \csd. It generalizes Lemma \ref{lem:sd-pruning} which corresponds to the $l=1$ case.

\begin{lemma}
\label{lem:csd-pruning}
Suppose $\mathcal{M}$ is a $(r_0,r_1,\dots,r_l)$-\csd\ with length-vector $(m_0,m_1,\dots,m_l)$. Let $T$ be a $(r_0,m_0)$-ultra periodic affine subspace of $\F_q^{m_l}$. Then the dimension of the affine space $T \cap U(\mathcal{M})$ is at most $r_l$.
\end{lemma}
\begin{proof}
The idea will be to apply Lemma \ref{lem:sd-pruning} inductively, for increasing periods $m_0,m_1,$ $\dots,m_{l-1}$.
Since $T$ is $(r_0,m_0)$-ultra periodic, it is $(r_0,m_0)$-periodic and $((m_1/m_0) r_0,m_1)$-periodic. Using this together with Lemma \ref{lem:sd-pruning}, it follows that
\[ T \cap \{ \mv{x} \in \F_q^{m_l} \mid \proj_{[j m_1 +1 , (j+1)m_1]}(\mv{x}) \in  H^{(1)}_1 \times H^{(1)}_2 \times \cdots
  \times H^{(1)}_{m_1/m_{0}} \text{ for }  0\le j < m_l/m_1\} \]
  is an affine subspace that is $(r_1,m_1)$-periodic. Continuing this argument, the affine subspace of $T$ formed by restricting each $m_\iota$-block to belong to $H^{(\iota)}_1 \times H^{(\iota)}_2 \times \cdots
  \times H^{(\iota)}_{m_\iota/m_{\iota-1}}$ for $1 \le \iota \le j$ is $(r_j,m_j)$-periodic. For $j=l$, we get the intersection $T \cap U(\mathcal{M}) \subset \F_q^{m_l}$ will be $(r_l,m_l)$-periodic, which simply means that it is an $r_l$-dimensional affine subspace of $\F_q^{m_l}$.
\end{proof}

We conclude this section by constructing a canonical subspace that has low-dimensional intersection with ultra-periodic subspaces based on the explicit subspace designs of Theorem~\ref{thm:sd-explicit}.
This statement will be used in Section~\ref{subsec:put-together-ag} for pre-coding algebraic-geometric codes based on the Garcia-Stichtenoth tower.
\begin{theorem}
\label{thm:final-csd-subspace}
Let $q \ge 4$ be a prime power.
Let $\eta\in (0,1)$ and integers $\dm,r\ge 2$ satisfy
$\dm \ge c r \eta^{-1} \log (r/\eta)$ for a large enough (absolute) constant $c > 0$.
For all large enough multiples $\dims$ of $\dm$, we can construct a subspace $U$ of $\F_q^\dims$ of dimension at least $(1-\eta)\dims$ such that for every $(r,\dm)$-ultra periodic affine subspace $T \subset \F_q^\dims$, the dimension of the affine subspace $U \cap T$ is at most
$(r/\eta)^{2^{O(\log^* \dims)}}$.
The subspace $U$ can be constructed in deterministically in $\mathrm{poly}(\dims,q)$ time.
\end{theorem}
\begin{proof}
We will take $U$ to the canonical subspace $U(\mathcal{M})$ of an appropriate \csd\ $\mathcal{M}$. To this end, given our work so far, the main remaining task is to pick the parameters of $\mathcal{M}$ carefully. Let $\eta_\iota = \frac{\eta}{4 \cdot 2^{\iota}}$ for $\iota=0,1,2,\dots$.

Let $m_0=\dm$, $m_1 = m_0 \cdot \lfloor (r/\eta)^{c/4} \rfloor$, and for $\iota \ge 0$, $m_{\iota+1} = m_{\iota}\cdot q^{\lceil \sqrt{m_{\iota}}\rceil}$. Let $r_0=r$, and for $\iota\ge 0$, $r_{\iota+1} = \lceil r^2_{\iota}/\eta_{\iota} \rceil$.
For this choice of parameters, one can verify that (i) $r_\iota \le \eta_\iota m_\iota/4$, and
(ii) $q^{\eta_\iota m_\iota/(4 r_\iota)} \ge m_{\iota + 1}/m_{\iota}$ for all $\iota \ge 0$.
Indeed, to verify the first condition by induction, one only needs to check that $m_{\iota+1} \ge m_\iota^2$, which is true for $\iota=0$ for a large enough choice of $c$, and for $\iota \ge 1$, $m_{\iota+1}$ in fact grows exponentially in $\sqrt{m_\iota}$. For the second condition, for $\iota =0$ it follows from our assumption that $\dm \ge c r \eta^{-1} \log (r/\eta)$. For $\iota \ge 1$, it is implied by $r_\iota/\eta_\iota \ll \sqrt{m_\iota}/4$, which is true for $\iota=1$ for large enough $c$, and for $\iota > 1$ by induction since $r_\iota/\eta_\iota$ grows quadratically in each step, whereas $m_\iota$ grows exponentially.

We can therefore conclude by Theorem~\ref{thm:sd-explicit} that we can construct a $(r_{\iota},r_{\iota+1})$-\sd\ of cardinality $m_{\iota+1}/m_\iota$ in $\F_q^{m_\iota}$ of dimension $(1-\eta_\iota)m_\iota$.


Pick $l$ to the smallest integer so that $m_{l-1}\ge (\log_q \dims)^2$. Since $m_0 =\dm \ge 2$ and $m_{\iota+1} \ge q^{\sqrt{m_\iota}}$ for $1 \le \iota < l$, it is easy to see that that $l \le O(\log^* \dims)$ 
Redefine $m_{l-1}$ to equal $m'_{l-1}$ which is the smallest multiple of $m_{l-2}$ that is at least $(\log_q \dims)^2$. Since $m_{l-2} < (\log_q \dims)^2$, we have $(\log_q \dims)^2 \le m'_{l-1} < 2 (\log_q \dims)^2$. We also redefine $m_l$ to equal the largest multiple $m'_l$ of $m'_{l-1}$ that is at most $\dims$. This implies $\dims -m'_l < m'_{l-1}$. Note that $m'_{l-1} \le m_{l-2} q^{\lceil \sqrt{m_{l-2}}\rceil}$ and $m'_l \le q^{\sqrt{m'_{l-1}}}$. For notational simplicity, let us re-denote $m'_{l-1}$ and $m'_l$ by $m_{l-1}$ and $m_l$.

Thus for these parameters, we can construct an $(r_0,r_1,\dots,r_l)$-\csd\ $\mathcal{M}_l$ with length-vector $(m_0,m_1,\dots,m_l)$ and dimension-vector $(d_0,d_1,\dots,d_{l-1})$ where $d_\iota \ge (1-\eta/2^{\iota+2}) m_\iota$.

The construction time for \sds\ guaranteed by Theorem~\ref{thm:sd-explicit} implies that $\mathcal{M}_l$ can be constructed in $\mathrm{poly}(m_l,q)$ time.
We define the desired subspace $U \subset \F_q^\dims$ as $U(\mathcal{M}_l) \times 0^{\dims-m_l}$, i.e., $U$ consists of the vectors in the canonical subspace $U(\mathcal{M}_l) \subset \F_q^{m_l}$ padded with $\dims - m_l$ zeroes at the end. By Observation \ref{obs:csd-dim}, the dimension of $U$ is at least
\begin{eqnarray*} \left(1- \sum_{\iota=0}^{l-1} \frac{\eta}{4 \cdot 2^\iota}\right) m_l &\ge& (1-\eta/2)m_l
> (1-\eta/2) (\dims-m_{l-1})\\ & >& (1-\eta/2)\dims - 2 (\log_q \dims)^2 > (1-\eta) \dims \end{eqnarray*}
for large enough $\dims$. This proves that the subspace $U$ has dimension at least $(1-\eta)\dims$, and can be constructed deterministically in $\mathrm{poly}(q,\dims)$ time.

It remains to prove the claimed intersection property with ultra-periodic subspaces. Let $T$ be an arbitrary $(r,\dm)$-ultra periodic affine subspace of $\F_q^\dims$. By Lemma \ref{lem:csd-pruning}, $\proj_{m_l}(T) \cap U(\mathcal{M})$ is an affine subspace of $\F_q^{m_l}$ of dimension at most $r_l$. Clearly, the same dimension bound also holds for $T \cap U$ since the last $\dims - m_l$ coordinates for vectors in $U$ are set to $0$. The proof is complete by noting that for our choice of parameters, $r_l \le (2r/\eta)^{2^l}$ and $l \le O(\log^* \dims)$.
\end{proof}

\section{Pre-coding AG codes using h.s.e sets and subspace designs}
\label{sec:PL}

In this final section, we combine the algebraic list decoding results (for folded AG codes and AG codes with subfield evaluation points) with subspace evasive sets (h.s.e sets and subspace designs) to deduce our main results on optimal rate list-decodable codes (Theorem~\ref{thm:main-intro}). The idea is to pre-code the messages of the algebraic codes to belong to the subspace evasive sets, so that only a small number of candidates fall in the periodic (or ultrae-periodic) subspaces that arise in algebraic decoding and further they can be enumerated efficiently.

We stress that for our final code constructions either h.s.e sets or subspace designs can be used in combination with either the folded variant or the subspace evaluation variant. For concreteness though, below we focus on the following two combinations:
\begin{enumerate}
\item folded codes with h.s.e sets, and
\item subspace evaluation codes with subspace designs.
\end{enumerate}
We note that the use of h.s.e sets leads to smaller final list size but their construction is randomized. Subspace designs lead to slightly larger list size (which in particular grows, albeit very slowly, with the block length) but the advantage is that they can be explicitly constructed.

\subsection{Pruning with h.s.e sets}
\label{subsec:fag-hse}

We begin with pruning via h.s.e sets, applied to the folded Hermitian and folded Garcia-Stichtenoth codes from Sections~\ref{newsubsec:FH} and \ref{subsec:FGS} respectively. In particular, the combination of folded Garcia-Stichtenoth codes with h.s.e sets will give us our final main Monte Carlo code construction (part (i) of Theorem~\ref{thm:main-intro}). We start with the folded Hermitian case as a warmup.

\subsubsection{Combining folded Hermitian codes and h.s.e sets}
\label{sec:herm+se}

Instead of encoding arbitrary $\mv{f} \in \F_q^k$ by the folded Hermitian code (Definition \ref{def:folded-herm-F}), we can restrict the messages $\mv{f}$ to belong to the range of our h.s.e set, so that the affine space of solutions guaranteed by Lemma \ref{lem:space-of-solns-rs} can be efficiently pruned to a small list. The formal claim is below.

\begin{theorem}
	\label{thm:final-herm+se}
	Let $e\ge 2$ be an integer, $r \ge 2e$ be a large enough prime power, $q =r^2$, and $\zeta \in (1/q,1)$. Let $k \le q^{\zeta q/2}$ be a positive integer. Let $s,m$ be positive integers satisfying $1 \le s \le m \le q-1$ and $s < \zeta q/12$. Finally let $N$ be an integer satisfying $k + 2 e r^e \le N m \le (q-1) r^e$.
	
	Consider the code $C_1$ with encoding $E_1 : \F_q^{(1-\zeta)^2 k} \to (\F_q^m)^N$ defined as
	\[ E_1(\mv{x})=\widetilde{FH}_e(N,k,q,m)(\hse(\mv{x})) \ , \] 
	for a random map $\hse : \F_q^{(1-\zeta)^2 k} \to \F_q^k$ as constructed in Theorem~\ref{thm:hse-set-final} for a period size $\period=q-1$ and $b = \lceil \frac{k}{q-1} \rceil$.

	Then, the code $C_1$ code has rate $R=(1-\zeta)^2 k/(Nm)$, can be encoded in $\mathrm{poly}(N m q^{\zeta q})$ time, and with high probability over the choice of $\hse$, it is $(\tau,\ell)$-list decodable in time $\mathrm{poly}(N m q^{\zeta q})$ for $\ell \le O(1/(R\zeta))$ and
	\[ \tau = \frac{s}{s+1} \left( 1 - \frac{k}{N(m-s+1)} \right) - \frac{3m}{m-s+1} \frac{e r^e}{m N} \ . \]
\end{theorem}
\begin{proof}
	The claim about the rate is clear, and the encoding time follows from the time to compute $\hse$ recorded in Theorem~\ref{thm:hse-set-final}.
	
	By \eqref{eq:x4}, the genus $\g_e \le e r^e$, and so the condition on $N,m$ meets the requirement for the construction of the folded Hermitian tower based code in Definition \ref{def:fh-F}, and the claimed value of the error fraction $\tau$ satisfies \eqref{eq:herm-error-frac-F}. By Lemma \ref{lem:x5.6}, we know that the candidate messages found by the decoder lie in one of at most $q^{2Nm}$ possible  $(s,q-1)$-periodic subspaces of $\F_q^k$. 
	
	One can check that the conditions of Theorem \ref{thm:hse-set-final} are met for our choice of $\zeta,s,q,k,\period$.
	Appealing to Theorem \ref{thm:hse-set-final} with the choice $c = 2Nm/k =O(1/R)$, we conclude that, with high probability over the choice of $\hse$, there is a decoding algorithm running in time $\mathrm{poly}(N m q^{\zeta q})$ to list decode $C_1$ from a fraction $\tau$ of errors, outputting at most $O(1/(R\zeta))$ messages in the worst-case.
\end{proof}

\paragraph{\textsl{Choosing parameters.}}
Let $\eps > 0$ be a small positive constant, and a family of codes of length $N$ (assumed large enough) and rate $R\in (0,1)$ is sought. Pick $n$ to be a growing parameter.

By picking $s = \Theta(1/\eps)$, $m = \Theta(1/\eps^2)$,  $r = \lfloor \log n \rfloor$, $e  = \lceil \frac{\log n}{\log \log n} \rceil$, $\zeta = (\log n \log \log n)^{-1}$, $N = \lfloor\frac{(r^2-1)r^e}{m}\rfloor$, and $k$ proportional to $Nm$ in Theorem~\ref{thm:final-herm+se}, we can conclude the following.

\begin{cor}
	For any $R \in (0,1)$ and positive constant $\eps \in (0,1)$, there is a randomized construction of a family of codes of rate at least $R$ over an alphabet size $(\log N)^{O(1/\eps^2)}$  that are encodable and $(1-R-\eps, O(R^{-1} \log N \log \log N))$-list decodable in $\mathrm{poly}(N,1/\eps)$ time, where $N$ is the block length of the code.
\end{cor}

Our promised main result (Theorem~\ref{thm:main-intro}) achieves better parameters than the above, namely an alphabet size of $\exp(\tilde{O}(1/\eps^2))$ and list-size of $O(1/(R\eps))$.
This is based on the Garcia-Stichtenoth tower and is described next.

\subsubsection{Combining folded Garcia-Stichtenoth codes and h.s.e sets}
\label{sec:fgs+se}

Similarly to Section \ref{sec:herm+se}, we now show how to pre-code the messages of the FGS code with a h.s.e subset. Here we will work with a base field $\F_q$ whose size is fixed and independent of the code dimension $k$, which will lead both to constant alphabet size and constant list size. To accommodate the requirement that $k \le q^{O(\zeta \period)}$, we work with a larger ``period" size for the h.s.e sets to evade.
 Recall Observation~\ref{obs:ultra-period-scales} which implies that if $H$ is an $(s,\period,b)$-ultra periodic subspace of $\F_q^k$ for $k=b\period$, then $H$ is also $(su,\period u, b/u)$-periodic for every integer $u\ge 1$ with $u|b$. Thus we can scale up the period size using the ultra-periodicity of the subspace guaranteed by the decoder of Theorem~\ref{thm:folded-GS-ld}. Following Remark~\ref{rmk:periodic-hse}, we actually only need a weaker property to prune via h.s.e sets which can also be ensured without ultra-periodicity. But since we have the stronger property available, we make use of it (this is also for sake of uniformity with the pruning based on subspace designs of Section~\ref{subsec:subfield-subs-design} which can also be applied to the FGS code).


As in the Hermitian case, instead of encoding arbitrary $\mv{f} \in \F_q^k$ by the folded Garcia-Stichtenoth code, we will restrict the messages $\mv{f}$ to belong to the range of our h.s.e set. This will ensure that the affine space of solutions  can be efficiently pruned to a small list.

\begin{theorem}
  	\label{thm:final-fgs+hse}
	Let $r$ be a prime power, $q=r^2$, and $e\ge 2$ be an integer, and $\zeta \in (0,1)$.  Let $\period \le k$ be a multiple of $(r-1)$, say $\period = u (r-1)$ for a positive integer $u$.
	Let $k \le q^{\zeta \period/2}$ be a positive integer.

	Let $s,m$ be positive integers satisfying $1 \le s \le m \le r-1$ and $s < \zeta r/12$. Finally let $N$ be an integer satisfying $k + 2 r^e \le N m \le (r-1) r^e$.
	
	Consider the code $C_2$ with encoding $E_2 : \F_q^{(1-\zeta)^2 k} \to (\F_q^m)^N$ defined as
	\[ E_2(\mv{x})=\widetilde{FGS}_e(N,k,q,m)(\hse(\mv{x})) \ , \]
	
	for a random map $\hse : \F_q^{(1-\zeta)^2 k} \to \F_q^k$ as constructed in Theorem~\ref{thm:hse-set-final}
	for a period size $\period$ and $b = \lceil \frac{k}{\period} \rceil$.
	
	
	The code $C_2$ has rate $R=(1-\zeta)^2k/(Nm)$, can be encoded in $\mathrm{poly}(N m q^{\zeta \period})$ time, and w.h.p over the choice of $\hse$, it is $(\tau,\ell)$-list decodable in time $\mathrm{poly}(N m q^{\zeta \period})$ for $\ell \le O(1/(R\zeta))$ and
	\begin{equation}
	\label{eq:tau-fgs}
	\tau = \frac{s}{s+1} \left( 1 - \frac{k}{N(m-s+1)} \right) - \frac{3m}{m-s+1} \frac{r^e}{m N} \ .
	\end{equation}
\end{theorem}
\begin{proof}
	The proof is very similar to that of Theorem~\ref{thm:final-herm+se}.
	The claim about the rate is clear, and the encoding time follows from the time to compute $\hse$ recorded in Theorem~\ref{thm:hse-set-final}.
	
	The genus $\g_e$ is now upper bounded by $r^e$, and so the condition on $N,m$ meets the requirement for the construction of the folded  the Garcia-Stichtenoth tower based code in Definition \ref{def:fgs}, and the claimed value of the error fraction $\tau$ satisfies \eqref{eq:fgs-error-frac}. By Lemma \ref{lem:x5.6}, we know that the candidate messages found by the decoder lie in one of at most $q^{2Nm}$ possible  $(s,r-1, \lceil \frac{k}{r-1}\rceil)$-periodic subspaces.\footnote{Technically, it will belong to $\proj_k(W)$ of such a periodic subspace $W$, but we may pretend that there are $(r-1) \lceil k/(r-1)\rceil - k$ extra dummy coordinates which we decode. Or we can just assume for convenience that $r-1$ divides $k$.}
	Now by Observation \ref{obs:scale-up}, each of these subspaces is also $(su,\period, \lceil \frac{k}{\period}\rceil)$-periodic.
	One can check that the conditions of Theorem \ref{thm:hse-set-final} are met for our choice of $\zeta,s,q,k,\period$ and taking $su$ to play the role of $s$ (since $s < \zeta r/12$, we have $s u < \zeta \period/10$).
	
	Appealing to Theorem \ref{thm:hse-set-final} with the choice $c = 2Nm/k =O(1/R)$, we conclude that there is a decoding algorithm running in time $\mathrm{poly}(N m q^{\zeta \period})$ to list decode $C_2$ from a fraction $\tau$ of errors, outputting at most $O(1/(R\zeta))$ messages in the worst-case.
\end{proof}

\paragraph{\textsl{Choosing parameters.}}
Finally, all that is left to be done is to pick parameters to show how the above can lead to optimal rate list-decodable codes over a constant-sized alphabet which further achieve very good list-size.

Let $\eps > 0$ be a small positive constant, and a family of codes of length $N$ (assumed large enough) and rate $R\in (0,1)$ is sought. Pick $n$ to be a growing parameter.

Let us pick $s = \Theta(1/\eps)$, $m = \Theta(1/\eps^2)$,  $\zeta = \eps/12$, $r = \Theta(1/\eps)$, $q=r^2$, and $e = \lceil \frac{\log n}{\log r} \rceil$, $N = \lfloor\frac{(r-1)r^e}{m}\rfloor$, and $k = R N m (1+\eps)$. This ensures that (i) there are at least $n=Nm$ rational places and so we get a code of length at least $n/m=N$, (ii) the rate of the code $C_2$ is at least $R$, and (iii) the error fraction \eqref{eq:tau-fgs} is at least $1-R-\eps$.

The remaining part is to pick a multiple $\period$ of $(r-1)$ so that the $k \le q^{\zeta \period/2}$ condition is met. This can be achieved by choosing $u = \lceil \frac{\log n}{\log (1/\eps)} \rceil$ and $\period =  (r-1) u$.
With these choices, we can conclude the following, which is our main randomized code construction promised in part (i) of Theorem~\ref{thm:main-intro}.

\begin{theorem}[Main; Corollary to Theorem \ref{thm:final-fgs+hse} with above choice of parameters]
\label{thm:final-monte-carlo}
	For any $R \in (0,1)$ and positive constant $\eps \in (0,1)$, there is a Monte Carlo construction of a family of codes of rate at least $R$ over an alphabet size $\exp(O(\log(1/\eps)/\eps^2))$ that are encodable and $(1-R-\eps, O(1/(R\eps))$-list decodable in $\mathrm{poly}(N)$ time, where $N$ is the block length of the code.
\end{theorem}

It may be instructive to recap why the Hermitian tower could not give a result like the above one. In the Hermitian case, the ratio $\g_e/n$ of the genus to the number of rational places was about $e/r=e/\sqrt{q}$, and thus we needed $q > e^2$. Since the period $\period$ was about $q$, the running time of the decoder was bigger than $q^{\Omega(\zeta q)}$, whereas the length of the code was at most $q^{O(\sqrt{q})}$. This dictated the choice of $q \approx \log^2 n$, and then to keep the running time polynomial, we had to take $\zeta \approx (\log n \log\log n)^{-1}$.

\subsection{Pruning using subspace designs}
\label{sec:put-together}
\label{subsec:subfield-subs-design}
We now combine our the constructions of Reed-Solomon and Garcia-Stichtenoth codes with evaluation points in a subfield (from Section~\ref{new-subsec:decoding-rs} and Section~\ref{subsec:decoding-gs} respectively) with a pre-coding step that restricts the message coefficients to (respectively) the subspaces constructed in Theorem~\ref{thm:final-sd-subspace} (using subspace designs) and Theorem~\ref{thm:final-csd-subspace} (using cascaded subspace designs). These subcodes will then be list decodable with smaller list-size in polynomial time.

In particular, the combination of Garcia-Stichtenoth codes with cascaded subspace design will give us our final main deterministic code construction (part (ii) of Theorem~\ref{thm:main-intro}).

\subsubsection{Subcodes of Reed-Solomon codes}
\label{subsec:put-together-rs}
We begin with the case of Reed-Solomon codes as considered in Section~\ref{new-subsec:decoding-rs}.  For a finite field
$\F_q$, constant $\eps > 0$, integers $n,k,m,s$ satisfying $1 \le k <
n \le q$ and $1 \le s \le \eps m/12$, we will define subcodes of
$\mathrm{RS}^{(q,m)}[n,k]$.  Below for a polynomial $f \in
\F_{q^m}[X]$ with $k$ coefficients $f_0,f_1,\dots,f_{k-1}$, we denote
by $\mv{f_0},\mv{f_1},\dots,\mv{f_{k-1}}$ the representation of these
coefficients as vectors in $\F_q^m$ by fixing some $\F_q$-basis of
$\Fm$.

\newcommand{\wRS}{\widehat{\mathsf{RS}}}

Define the subcode $\wRS$ of $\mathrm{RS}^{(q,m)}[n,k]$ consisting of the encodings of $f \in \F_{q^m}[X]$ such that $(\mv{f_0},\mv{f_1},\dots,\mv{f_{k-1}}) \in V$ for a subspace $V \subseteq \F_q^{mk}$ guaranteed by Theorem \ref{thm:final-sd-subspace}, when applied with the parameter choices
\[ \dm =m; \quad b=k; \quad r=s-1; \quad \eta = \eps \ . \]
Note that $\wRS$ is an $\F_q$-linear code over the alphabet $\Fm$ of rate
$(1-\eps)k/n$, and it can be constructed in deterministic $q^{O(m^2)}$ time, or Las Vegas $q^{O(ms)}$ time.\footnote{It can also be constructed in Monte Carlo $(q/\eps)^{O(1)}$ time by randomly picking  subspaces for the \sd\ used to construct $V$ in Theorem \ref{thm:final-sd-subspace}.}
\begin{theorem}
\label{thm:sd-subcode-rs}
Given an input string $\mv{y} \in \F_{q^m}^n$, a basis of an affine subspace of dimension at most $O(s/\eps)$ that includes all codewords of the above subcode $\wRS$ that lie within Hamming distance $\frac{s}{s+1} (n-k)$ from $\mv{y}$ can be found in deterministic $\mathrm{poly}(n,\log q, m)$ time.
\end{theorem}
\begin{proof}
By Lemma \ref{lem:space-of-solns-rs}, we can compute the $(s-1,m,k)$-periodic subspace $T$ of messages whose Reed-Solomon encodings can be within Hamming distance $\frac{s}{s+1} (n-k)$ from $\mv{y}$. By Theorem \ref{thm:final-sd-subspace}, the intersection $T \cap V$ is
is an affine subspace over $\F_q$ of dimension $d = O(s/\eps)$. Since both steps involve only basic linear algebra, they can be accomplished using $\mathrm{poly}(n,m)$ operations over $\F_q$.
\end{proof}
By picking $s = \Theta(1/\eps)$ and $m = \Theta(1/\eps^2)$ in the above construction, we can conclude the following.
\begin{cor}
\label{cor:cap-achieving-rs-deterministic}
For every $R \in (0,1)$ and $\eps > 0$, and all large enough integers $n < q$ with $q$ a prime power, one can construct a rate $R$ $\F_q$-linear subcode of a Reed-Solomon code of length $n$ over $\Fm$, such that the code can be (i) encoded in $(n/\eps)^{O(1)}$ time and (ii) list decoded from a fraction $(1-\eps)(1-R)$ of errors in $(n/\eps)^{O(1)}$ time, outputting a subspace over $\F_q$ of dimension $O(1/\eps^2)$ including all close-by codewords.
The code can be constructed deterministically in $\mathrm{poly}(q)$ time.
\end{cor}

We note that the above list decoding guarantee is in fact weaker than what is achieved for folded Reed-Solomon codes in \cite{GW-tit13}, where the codewords were pinned down to a dimension $O(1/\eps)$ subspace. We can improve the list size above to $\mathrm{poly}(1/\eps)$ using pseudorandom subspace-evasive sets
as in \cite{GW-tit13}, or to $\exp(\eps^{-O(1)})$ using the explicit subspace-evasive sets from \cite{DL-stoc12}. The main point of the above result is not the parameters but that an explicit subcode of RS codes has optimal list decoding radius with polynomial complexity.

\subsubsection{Subcodes of Garcia-Stichtenoth codes}
\label{subsec:put-together-ag}
We now pre-code the codes constructed in Section \ref{subsec:decoding-gs}.
For a finite field $\F_q$, constant $\eps > 0$, and integers $s,m$ satisfying
$1 \le s \le O(\eps m/\log(1/\eps))$ and $m \ge \Omega(1/\eps^2)$, we will define subcodes of $\mathrm{GS}^{(q,m)}[N,k]$ guaranteed by Theorem \ref{thm:gs-decoding}. Note that messages space of this code can be identified with $\F_q^{mk}$.

\newcommand{\wGS}{\widehat{GS}}
Define the subcode $\wGS$ of $\mathrm{GS}^{(q,m)}[N,k]$ consisting of the encodings of
a subspace $U \subseteq \F_q^{mk}$ guaranteed by Theorem \ref{thm:final-csd-subspace}, when applied with the parameter choices
\begin{equation}
\label{eq:params-choice-gs-csd}
 \eta=\eps; \quad r=s-1; \quad \dm=m; \quad \dims = km \ .
 \end{equation}
Note that $\wGS$ is an $\F_q$-linear code over the alphabet $\Fm$ of rate
$(1-\eps)k/N$. Also, it can be constructed in $\mathrm{poly}(k,m,q)$ time by virtue of the construction complexity of $U$.
\begin{lemma}
\label{lem:csd-subcode-gs}
Given an input string $\mv{y} \in \F_{q^m}^N$, a basis of an affine subspace of dimension at most
\[ (s/\eps)^{2^{O(\log^* (km)}} \]
that includes all codewords of the above subcode within Hamming distance $\frac{s}{s+1} (N-k) - 3N/(\sqrt{q}-1)$ from $\mv{y}$ can be found in deterministic $\mathrm{poly}(n,\log q, m)$ time.
\end{lemma}
\begin{proof}
By Theorem \ref{thm:gs-decoding}, we can compute the $(s-1,m)$-ultra periodic subspace $T$ of messages whose encodings are within Hamming distance $\frac{s}{s+1} (N-k) - 3 N/(\sqrt{q}-1)$ from $\mv{y}$. By Theorem \ref{thm:final-csd-subspace}, for the above choice of parameters \eqref{eq:params-choice-gs-csd}, the intersection $T \cap U$
is an affine subspace over $\F_q$ of dimension  $(s/\eps)^{2^{O(\log^* (km)}}$.
Since both steps involve only basic linear algebra, they can be accomplished using $\mathrm{poly}(N,m)$ operations over $\F_q$.
\end{proof}

By taking $q = \Theta(1/\eps^2)$, and choosing $s = \Theta(1/\eps)$ and $m = \Theta(\eps^{-2} \log(1/\eps))$ in the above lemma, we conclude our main result (stated informally as part (ii) of Theorem~\ref{thm:main-intro}) concerning the explicit construction of codes list decodable up to the Singleton bound over fixed alphabets and very slowly growing list-size.

\begin{theorem}[Main deterministic code construction]
\label{thm:cap-achieving-gs-deterministic}
For every $R \in (0,1)$ and $\eps > 0$, there is a deterministic polynomial time constructible family
of error-correcting codes of rate $R$ over an alphabet of size $\exp(O(\Ge^{-2} \log^2 (1/\Ge)))$ that can be list decoded in polynomial time from a fraction $(1-R-\Ge)$ of errors, outputting a list of size at most $\exp_{1/\Ge}\left(\exp_{1/\Ge}(\exp(O(\log^*N)))\right)$, where $N$ is block length of the code.
\end{theorem}

\medskip \noindent \textbf{\large Acknowledgments.} We are grateful to the anonymous reviewers for detailed and perceptive comments which led to significant improvements in the organization and presentation of the paper.

\providecommand{\bysame}{\leavevmode\hbox to3em{\hrulefill}\thinspace}
\providecommand{\MR}{\relax\ifhmode\unskip\space\fi MR }
\providecommand{\MRhref}[2]{%
  \href{http://www.ams.org/mathscinet-getitem?mr=#1}{#2}
}
\providecommand{\href}[2]{#2}


\begin{thebibliography}{10}

\bibitem{AS14}
Avraham Ben-Aroya and Igor Shinkar, \emph{A note on subspace evasive sets},
  Chicago Journal of Theoretical Computer Science (2014), 1--11.

\bibitem{BKW-subs-design}
Michael Braun, Michael Kiermaier, and Alfred Wassermann, \emph{$q$-analogs of
  designs: Subspace designs}, Network Coding and Subspace Designs (Silberstein
  N. Vazquez-Castro~M. Greferath~M., Pavcevic~M., ed.), Signals and
  Communication Technology, Springer, Cham, 2018.

\bibitem{DL-stoc12}
Zeev Dvir and Shachar Lovett, \emph{Subspace evasive sets}, Proceedings of the
  44th Symposium on Theory of Computing Conference (STOC), 2012, pp.~351--358.

\bibitem{elias91}
Peter Elias, \emph{Error-correcting codes for list decoding}, IEEE Transactions
  on Information Theory \textbf{37} (1991), 5--12.

\bibitem{FG-random15}
Michael~A. Forbes and Venkatesan Guruswami, \emph{Dimension expanders via rank
  condensers}, 19th International Workshop on Randomization and Computation
  (RANDOM), 2015, pp.~800--814.

\bibitem{GS95}
Arnaldo Garcia and Henning Stichtenoth, \emph{A tower of {A}rtin-{S}chreier
  extensions of function fields attaining the {D}rinfeld-{V}l\u{a}dut bound},
  Inventiones Mathematicae \textbf{121} (1995), 211--222.

\bibitem{GS96}
\bysame, \emph{On the asymptotic behavior of some towers of function fields
  over finite fields}, Journal of Number Theory \textbf{61} (1996), no.~2,
  248--273.

\bibitem{GuoRonZewi20}
Zeyu Guo and Noga Ron{-}Zewi, \emph{Efficient list-decoding with constant
  alphabet and list sizes}, CoRR \textbf{abs/2011.05884} (2020).

\bibitem{Gur-cyclo-ANT}
Venkatesan Guruswami, \emph{Cyclotomic function fields, {A}rtin-{F}robenius
  automorphisms, and list error-correction with optimal rate}, Algebra and
  Number Theory \textbf{4} (2010), no.~4, 433--463.

\bibitem{Gur-ccc11}
\bysame, \emph{Linear-algebraic list decoding of folded {R}eed-{S}olomon
  codes}, Proceedings of the 26th IEEE Conference on Computational Complexity,
  June 2011.

\bibitem{GK-combinatorica}
Venkatesan Guruswami and Swastik Kopparty, \emph{Explicit subspace designs},
  Combinatorica \textbf{36} (2016), no.~2, 161--185, Preliminary version in
  FOCS 2013.

\bibitem{GN14}
Venkatesan Guruswami and Srivatsan Narayanan, \emph{Combinatorial limitations
  of average-radius list-decoding}, {IEEE} Transactions on Information Theory
  \textbf{60} (2014), no.~10, 5827--5842.

\bibitem{GR-FRS}
Venkatesan Guruswami and Atri Rudra, \emph{Explicit codes achieving list
  decoding capacity: {E}rror-correction with optimal redundancy}, IEEE
  Transactions on Information Theory \textbf{54} (2008), no.~1, 135--150.

\bibitem{GR-BlokhZyablov}
\bysame, \emph{Better binary list decodable codes via multilevel
  concatenation}, {IEEE} Trans. Information Theory \textbf{55} (2009), no.~1,
  19--26.

\bibitem{GS99}
Venkatesan Guruswami and Madhu Sudan, \emph{Improved decoding of
  {R}eed-{S}olomon and {A}lgebraic-geometric codes}, IEEE Transactions on
  Information Theory \textbf{45} (1999), no.~6, 1757--1767.

\bibitem{GW-tit13}
Venkatesan Guruswami and Carol Wang, \emph{Linear-algebraic list decoding for
  variants of {R}eed-{S}olomon codes}, {IEEE} Transactions on Information
  Theory \textbf{59} (2013), no.~6, 3257--3268.

\bibitem{GWX16}
Venkatesan Guruswami, Carol Wang, and Chaoping Xing, \emph{Explicit
  list-decodable rank-metric and subspace codes via subspace designs}, {IEEE}
  Trans. Information Theory \textbf{62} (2016), no.~5, 2707--2718.

\bibitem{GX-stoc12}
Venkatesan Guruswami and Chaoping Xing, \emph{Folded codes from function field
  towers and improved optimal rate list decoding}, Proceedings of the 44th
  Symposium on Theory of Computing Conference (STOC), 2012, pp.~339--350.

\bibitem{GX-stoc13}
\bysame, \emph{List decoding {R}eed-{S}olomon, algebraic-geometric, and
  {G}abidulin subcodes up to the {S}ingleton bound}, Proceedings of the ACM
  Symposium on Theory of Computing Conference (STOC), 2013, pp.~843--852.

\bibitem{GXY-tams}
Venkatesan Guruswami, Chaoping Xing, and Chen Yuan, \emph{Constructions of
  subspace designs via algebraic function fields}, Trans. Amer. Math. Soc.
  \textbf{370} (2018), 8757--8775.

\bibitem{KV}
Ralf Koetter and Alexander Vardy, \emph{Algebraic soft-decision decoding of
  {R}eed-{S}olomon codes}, IEEE Transactions on Information Theory \textbf{49}
  (2003), no.~11, 2809--2825.

\bibitem{Kopparty15}
Swastik Kopparty, \emph{List-decoding multiplicity codes}, Theory Comput.
  \textbf{11} (2015), 149--182.

\bibitem{KRS-pc}
Swastik Kopparty, Noga Ron{-}Zewi, Shubhangi Saraf, and Mary Wootters,
  \emph{Improved decoding of folded {R}eed-{S}olomon and multiplicity codes},
  Proceedings of the 59th {IEEE} Annual Symposium on Foundations of Computer
  Science, 2018, pp.~212--223.

\bibitem{MX17}
Liming Ma and Chaoping Xing, \emph{The asymptotic behavior of automorphism
  groups of function fields over finite fields}, Transactions of the Amercan
  Mathematical Society \textbf{372} (2019), 35--52.

\bibitem{NX01}
Harald Niederreiter and Chaoping Xing, \emph{Rational points on curves over
  finite fields--theory and applications}, Cambridge University Press, 2001.

\bibitem{PV-focs05}
Farzad Parvaresh and Alexander Vardy, \emph{Correcting errors beyond the
  {G}uruswami-{S}udan radius in polynomial time}, Proceedings of the 46th
  Annual IEEE Symposium on Foundations of Computer Science, 2005, pp.~285--294.

\bibitem{Shen93}
Ba-Zhong Shen, \emph{A {J}ustesen construction of binary concatanated codes
  that asymptotically meet the {Z}yablov bound for low rate}, IEEE Transactions
  on Information Theory \textbf{39} (1993), 239--242.

\bibitem{SAKSD01}
Kenneth Shum, Ilia Aleshnikov, P.~Vijay Kumar, Henning Stichtenoth, and Vinay
  Deolalikar, \emph{A low-complexity algorithm for the construction of
  algebraic-geometric codes better than the {G}ilbert-{V}arshamov bound}, IEEE
  Transactions on Information Theory \textbf{47} (2001), no.~6, 2225--2241.

\bibitem{Sil85}
Joseph~H. Silverman, \emph{The arithmetic of elliptic curves}, Springer, New
  York, 1985.

\bibitem{stich-book}
Henning Stichtenoth, \emph{Algebraic function fields and codes}, Universitext,
  Springer-Verlag, Berlin, 1993.

\bibitem{sudan}
Madhu Sudan, \emph{Decoding of {R}eed-{S}olomon codes beyond the
  error-correction bound}, Journal of Complexity \textbf{13} (1997), no.~1,
  180--193.

\bibitem{salil-NOW}
Salil Vadhan, \emph{Pseudorandomness}, Foundations and Trends in Theoretical
  Computer Science (FnT-TCS), NOW publishers, 2010, To appear. Draft available
  at {\tt http://people.seas.harvard.edu/\~\ salil/pseudorandomness/}.

\end{thebibliography}

\end{document}